\documentclass[12pt]{article}
\usepackage{graphicx,amssymb}
\usepackage[bookmarks=true]{hyperref}
\title{A quasi-potential for conservation laws with boundary conditions}
\author{C. Bahadoran
\thanks{
Laboratoire de Math\'ematiques CNRS UMR 6620, Universit\'e
Clermont-Ferrand 2, F-63177 Aubi\`{e}re. e-mail:
Christophe.Bahadoran@math.univ-bpclermont.fr } }
\date{}
\newcommand{\dsp}{\displaystyle}
\newcommand{\bd}{\begin{displaymath}}
\newcommand{\be}{\begin{equation}}
\newcommand{\ba}{\begin{array}}
\newcommand{\ed}{\end{displaymath}}
\newcommand{\ee}{\end{equation}}
\newcommand{\ea}{\end{array}}
\newcommand{\espace}{\mbox{ }}

\newcommand{\abs}[1]{\left|#1\right|}
\def\N{\mathbb{N}}
\def\Z{\mathbb{Z}}
\def\R{\mathbb{R}}

\newcommand{\bln}{{\mathcal E}}
\newcommand{\emsol}{\overline{\mathcal E}}
\newcommand{\qsol}{\overline{\mathcal Q}}
\newcommand{\gsol}{\mathcal G}
\newcommand{\eqref}[1]{(\ref{#1})}

\newtheorem{theorem}{Theorem}[section]

\newtheorem{proposition}{Proposition}[section]

\newtheorem{lemma}{Lemma}[section]
\newtheorem{corollary}{Corollary}[section]

\newenvironment{proof}[2]{\espace\\{\em Proof of #1 \ref{#2}.}}{\hfill\mbox{$\square$}}
\newenvironment{remark}{{\em Remark. }}{\espace \\ \\}

\begin{document}
\maketitle
\begin{abstract}
We compute the quasi-potential and determine minimizing paths for an action functional related to scalar conservation laws on an interval with boundary conditions 
in the sense of Bardos et al. (\cite{bln}). Taking as input an exclusion-like flux function $f(\rho)$, a strictly convex entropy $h(\rho)$, and boundary data $\rho_l,\rho_r$,
we obtain a generalization of the functional derived in \cite{dls} for the stationary large deviations of the asymmetric exclusion process.
\end{abstract}
\section{Introduction}
Let be given $K>0$ and a flux function $f$ on $[0,K]$ satisfying
\be\label{properties_flux}
f\in C^2([0,K]),\quad f(0)=f(K)=0,\quad f''\leq c<0\ee
We consider the scalar conservation law
\be \label{conservation_law}
\partial_t\rho(t,x)+\partial_x f(\rho(t,x))=0
\ee
on $(0,1)$, with boundary data 
\be\label{boundary_data}
\rho(t,0^+)=\rho_l,\quad
\rho(t,1^-)=\rho_r\ee
where $\rho_l,\,\rho_r\in[0,K]$.
It is well-known (see e.g. \cite{ser}) that, for given Cauchy data, \eqref{conservation_law}
has in general no classical solution, as even for smooth data, shocks appear in finite time. On the other hand, uniqueness fails for weak solutions. 
For the problem on $\R$, the unique physical solution (called the entropy solution) is selected among weak solutions by the so-called entropy condition. For the initial-boundary problem \eqref{conservation_law}--\eqref{boundary_data},
there is in general no entropy solution satisfying \eqref{boundary_data} in usual sense, but a unique entropy solution that satisfies \eqref{boundary_data} in the weaker sense of Bardos et al. (\cite{bln}), hereafter called BLN boundary conditions. In short, $\rho_l$ and $\rho_r$ in \eqref{boundary_data} must be viewed as sets rather than single values, and the definition of these sets depends on the boundary location through the outer normal. This phenomenon arises from the formation of boundary layers in the inviscid limit for the viscous parabolic approximation of \eqref{conservation_law}.\\ \\
In statistical mechanics, the asymmetric exclusion process (ASEP) and related interacting particle systems in the sense of \cite{lig} are microscopic models with a single conservation law, whose large-scale behavior (the so-called hydrodynamic limit, see \cite{kl}) is governed by \eqref{conservation_law}. For many such models, see e.g. \cite{rez},
it is established that the empirical particle density converges in probability to the unique entropy solution of \eqref{conservation_law}. It was also established
(\cite{ba0,ba}) that boundary conditions of \cite{bln} are obtained if the particle system is coupled to reservoirs at the boundary of some open domain.
Let us consider \eqref{conservation_law}--\eqref{boundary_data} on $(-\infty,0)_t\times(0,1)_x$. We introduce an action functional $I[\rho(.,.)]$  that quantifies deviation from $\rho(.,.)$ being a BLN entropy solution. This functional is believed to be a dynamical large deviation (LD) rate function for the time evolution of the empirical particle measure in ASEP-like models. 
%
%
Our goal is to compute the {\em quasi-potential}
$V[\rho(.)]$ for a function $x\in(0,1)\mapsto\rho(x)\in[0,K]$, defined  as the minimum energy of a path leading from a stationary state to $\rho(.)$:
\be\label{def_quasi}
V[(\rho(.)]:=\inf\{I[\rho(.,.)]:\,\rho(-\infty,.)\in\mathcal S,\,\rho(0,.)=\rho(.)\}
\ee
and determine the minimizer(s) in \eqref{def_quasi}, i.e. paths followed by the system to create a fluctuation from expected stationary states.
In \eqref{def_quasi}, $\mathcal S$ denotes the set of stationary solutions of \eqref{conservation_law}--\eqref{boundary_data}. This set contains a unique constant function $\rho_s(.)$ outside the phase transition line $\rho_l<\rho_r,\,f(\rho_l)=f(\rho_r)$, where it contains arbitrarily located shocks connecting $\rho_l$ and $\rho_r$ (\cite{ba}).
An analogous picture was established in \cite{dehp} for the stationary state of ASEP. This picture was extended heuristically to more general models and fluces in \cite{ps}, and a rigorous proof of this extension given in \cite{ba} using BLN boundary conditions.
Under suitable assumptions (see e.g. \cite{bg} in the context of lattice gases), if one can prove the dynamical LD principle with rate function $I[\rho(.,.)]$ for the empirical measure under the stochastic dynamics, this implies the static LD principle $S[\rho(.)]$ for the empirical measure under the invariant measure of the stochastic process.\\ \\
The quasi-potential was introduced in \cite{fw} in the context of stochastic perturbations of dynamical systems, like for instance
\be\label{diffusion}
dx_t^\varepsilon=-\nabla U(x_t^\varepsilon)dt+\varepsilon^{1/2}dW_t
\ee
where $W_t$ is a standard Brownian motion and $U(x)$ a potential on $\R$ with a single minimum at (say) $x=0$. In this case, instead of \eqref{conservation_law}--\eqref{boundary_data}, the limiting evolution problem is the ODE
\be\label{ode}\dot{x}_t=-\nabla U(x_t)\,,\ee
the dynamical action functional for a path $x_.=(x_t,\,t\in(-\infty,0])$ given by
\be\label{action_fw}
I[(x_.)]=\frac{1}{2}\int_{-\infty}^0[\dot{x}_t+\nabla U(x_t)]^2 dt\,,
\ee
and the quasi-potential $V(x)$, where $x\in\R$, reads
\be\label{quasi_fw}
V(x)=\inf\{I(x_.):\,x_{-\infty}=0,\,x_0=x\}
\ee
A standard computation leads to $V(x)=2U(x)$, in agreement with the here explicit invariant (and reversible) measure $\mu^\varepsilon=Z(\varepsilon)^{-1}e^{-2\varepsilon^{-1}U(x)}dx$ for \eqref{diffusion}, which identifies $2U$ as the stationary LD rate function. The same computation 
shows that the unique minimizing path (or ``fluctuation path'') $x_.$ in \eqref{quasi_fw} is obtained as the time reversal of the solution of \eqref{ode} starting from $x$. This is a particular case of the Onsager-Machlup principle, which states that for reversible dynamics the fluctuation path is the time-reversed relaxation path.\\ \\
Our problem is motivated by the study of nonequilibirum stationary states (NESS) in exclusion-like models coupled to reservoirs. A general feature of NESS is the occurence of long-range correlations, one signature of which is the non-locality of the stationary LD rate function. The first example of such a rate function was obtained in \cite{dls0} for the symmetric exclusion process coupled to reservoirs. The same authors (\cite{dls}) later considered the asymmetric exclusion process, which shares the non-locality feature, but is somewhat more complex due to the presence of two distinct regimes $\rho_l<\rho_r$ and $\rho_l>\rho_r$. The weakly asymmetric exclusion process was treated in the same spirit in \cite{de} in the regime $\rho_l>\rho_r$. The result of \cite{dls} is our primary motivation. 
%
%
In \cite{bdgjl} was introduced the so-called ``macroscopic fluctuation theory''. One outcome of this theory was to recover the static functional of \cite{dls0} 
by computing a quasi-potential and show that the fluctuation paths obey a nonequilibrium  generalization of the Onsager-Machlup principle, namely they are time-reversals of  relaxation paths of some {\em adjoint} evolution problem. Only in the reversible case does the adjoint problem coincide with the original problem. This approach does not depend directly on microscopic details of the particle dynamics, but rather on the dynamic functional. Thus it is applicable to certain other models. It was recently (\cite{bdgjlw1,bdgjlw2}) extended to the weakly asymmetric exclusion process, whose fluctuations in the vanishing noise limit $\varepsilon\to 0$ are described formally by the infinite-dimensional diffusion
\begin{eqnarray}
\partial_t\rho+\partial_xf(\rho) & = & \nu\partial_{x}[D(\rho)\partial_x\rho]-\varepsilon^{1/2}\partial_x[a(\rho)^{1/2}B(t,x)]\label{spde}\\
\rho (t,0^+) & = & \rho_l\label{spde_left}\\
\rho(t,1^-) & = & \rho_r\label{spde_right}
\end{eqnarray}
where $\rho=\rho(t,x)$, $B(t,x)$ is a space-time white noise, and $\nu$ a viscosity coefficient,
with the particular choice 
of flux $f$, diffusion coefficient $D$ and mobility $a$ given by
\be\label{particular_choice}
f(\rho)=\rho(1-\rho),\,D(\rho)=1,\,a(\rho)=\rho(1-\rho),
\ee
The action functional $I_\nu$ for \eqref{spde} is an infinite-dimensional version of \eqref{action_fw} that forbids violation of \eqref{spde_left}--\eqref{spde_right}. Its quasi-potential has an explicit local expression on the torus (\cite{bdgjl2}) that is valid for the most general $f$, $D$, $\rho$:
\be\label{quasi_eq}S[\rho(.)]=\int_0^1 h[\rho(x),\rho_0]dx\ee
where $h(\rho,\rho_0)$ is the relative version (see \eqref{relative_entropy}) of the entropy $h$ given by Einstein's relation
\be\label{einstein}h''(\rho)=D(\rho)/a(\rho)\ee
in agreement with the uncorrelated stady state of WASEP on the torus.
In the nonequilibrium case  \eqref{spde_left}--\eqref{spde_right} with $\rho_l\neq\rho_r$, for \eqref{particular_choice},
the non-local functional of \cite{de} is recovered in the regime $\rho_l>\rho_r$ (\cite{bdgjlw1}) and another one (\cite{bdgjlw2}) obtained in the regime $\rho_l<\rho_r$.\\ \\
Our purpose is threefold: to  generalize the static LD rate function of \cite{dls}, identify it as the quasi-potential for the action functional associated with \eqref{conservation_law}--\eqref{boundary_data}, and determine the adjoint evolution problem that yields the fluctuation path. The variational problem induced here is very different from \eqref{spde},\eqref{spde_left}--\eqref{spde_right} because the dynamic action functional is no longer an infinite-dimensional version of \eqref{action_fw}. In particular, Hamilton-Jacobi techniques from \cite{bdgjl} are not effective here. We take as basic input for our problem the exclusion-like (that is, satisfying \eqref{properties_flux}) flux $f$ and equilibrium entropy function $h(\rho)$ defined for $[0,K]$.
For the ASEP case studied in \cite{dls}, 
\be\label{asep_case} K=1,\quad f(\rho)=\rho(1-\rho),\quad h(\rho)=\rho\log\rho+(1-\rho)\log(1-\rho)\ee
The action functional is decomposed into a bulk and a boundary term. The former is determined only by the pair $(h,f)$, the latter also by the boundary data $\rho_l,\rho_r$.\\ \\
The bulk term is based on  \cite{jen,var,bbmn,mar} where two  closely related functionals were introduced as LD rate functions on the torus for the ASEP and for \eqref{spde} in the limit $\nu=\varepsilon\to 0$. These functionals are supported on weak solutions of \eqref{conservation_law} and measure the positive entropy production, i.e. violation of entropy conditions for \eqref{conservation_law}. We are currently able to treat the functional of \cite{bbmn,mar} but not the weaker one of \cite{jen,var}. For the latter we would need 
a single-entropy version of the structure results of \cite{dow,pan2} for convex conservation laws.
Therefore our result in its present form is not sufficient to deal with ASEP large deviations, but presumably effective 
for \eqref{spde} in the limit of \cite{mar}. We hope it could be used also for vanishing viscosity WASEP (\cite{fri}), which is formally similar to \eqref{spde}.
The boundary terms are a generalization of those introduced in \cite{bd2} for the ASEP case \eqref{asep_case}, and measure violation of BLN boundary conditions \eqref{boundary_data}. Although we only study the quasi-potential for fluces satisfying \eqref{properties_flux}, the dynamic action potential that we define is relevant to general fluces. In particular, like the bulk functional of \cite{bbmn}, our definition of the boundary part does not assume any convexity for $f$. This general joint bulk/boundary action functional has a natural variational formula.
%
%
The quasi-potential on the torus for the inside part of the functional was derived in \cite{bcm} and yields the equilibrium entropy \eqref{quasi_eq}, again in agreement with the uncorrelated steady states of ASEP-like systems on the torus.\\ \\
In Theorem \ref{theorem_1}, we obtain  a closed variational expression for the nonequilibrium quasi-potential relative to \eqref{conservation_law}--\eqref{boundary_data}. To do so we need a joint symmetry assumption on $f$ and $h$, see \eqref{symmetry_f}--\eqref{symmetry}, satisfied in particular 
by \eqref{asep_case}. This assumption can be interpreted as the flux being an even function of the chemical potential $\theta=h'(\rho)$ (that is the variable dual to $\rho$ w.r.t. the entropy $h$), see \eqref{symmetry_2}.
In a similar spirit, it was observed in \cite{bgl} that some (more stringent) relation between diffusion and mobility was required to get a closed expression for the quasi-potential
in \eqref{spde}-\eqref{spde_left}--\eqref{spde_right}. 
Our expression is a natural generalization of \cite{dls}, which gives a candidate static LD functional for ASEP-like systems satisfying the symmetry assumption. Besides, by definition of the quasi-potential, for given flux $f$, it yields a whole family of Lyapunov functionals for the evolution problem \eqref{conservation_law}--\eqref{boundary_data}
(namely for all entropies $h$ compatible with \eqref{symmetry}). 
The adjoint evolution problem is then described in Theorems \ref{theorem_2}--\ref{theorem_3} for the space-time reversal of the fluctuation path $\rho(.,.)$. Space-time reversal is more natural here than time reversal alone, because it leaves \eqref{conservation_law} invariant, hence relaxation paths  can be described with reference to the original problem \eqref{conservation_law}. Our results also cover the
phase transition line  $\rho_l<\rho_r$, $f(\rho_l)=f(\rho_r)$, where there is a continuum of stationary solutions for \eqref{conservation_law}--\eqref{boundary_data} parametrized by $y\in[0,1]$.\\ \\
There is a  connection between our results and those of \cite{bdgjlw1}--\cite{bdgjlw2}. For given viscosity $\nu$ in \eqref{spde}-\eqref{spde_left}--\eqref{spde_right}, in the case \eqref{particular_choice}, a quasi-potential $V_\nu$ is obtained in \cite{bdgjlw1}--\cite{bdgjlw2} from the diffusion-like action functional $I_\nu$. For \eqref{conservation_law}--\eqref{boundary_data}, we obtain a quasi-potential $V$ from the singular action functional $I$. 
Despite the difference in nature between $I_\nu$ and $I$, there is strong evidence that $I_\nu\to I$ holds in the sense of $\Gamma$-convergence.
$\Gamma$-convergence of the inside part of the functional is established on a set of ``nice'' functions  (\cite{bbmn}). On the other hand,  $\Gamma$-convergence of the boundary part is suggested in \cite{bd2}. With  a result similar to \cite{bbmn} at the boundary, for \eqref{asep_case} and the action functional $I$ restricted to nice functions, one could deduce Theorem \ref{theorem_1} from \cite{bdgjlw1}--\cite{bdgjlw2} in the case \eqref{particular_choice}--\eqref{asep_case}.
Conversely, Theorems \ref{theorem_2}--\ref{theorem_3} would determine the limit $\nu\to 0$ of the fluctuation paths in \cite{bdgjlw1}--\cite{bdgjlw2}. 
Strictly speaking this is true outside the phase transition, which is not seen at positive viscosity (but can be recovered differently by sending boundary conditions to $\pm\infty$, see \cite{bp}). \\ \\
We end up with an informal description of the fluctuation paths described in Theorems \ref{theorem_2}--\ref{theorem_3}.
A significant difference with known results at positive viscosity, and with the equilibrium case of \cite{bcm} is that, outside the maximum current phase,
the fluctuation is reached in finite time.
In the regime $\rho_l>\rho_r$, the quasi-potential is expressed as a supremum over an auxiliary function that depends nonlocally on $\rho(.)$. First, the  reversed evolution of this auxiliary function is determined   
as the entropy solution of \eqref{conservation_law} with space-reversed BLN conditions \eqref{boundary_data}. Next,  the reversed (unique) fluctuation path is obtained as the entropy solution to \eqref{conservation_law} with space-reversed, time-dependent BLN boundary data depending on the boundary values of the auxiliary function. The first step is reminiscent of \cite{bdgjl}--\cite{bdgjlw1}, but the specific second step is necessary because (unlike at positive viscosity) one cannot invert the relation $\rho(.)\mapsto F(.)$.
In the regime $\rho_l<\rho_r$, 
the quasi-potential is expressed as an infimum over an auxiliary position $y\in[0,1]$. We show that to each minimizing $y$ one can associate a unique fluctuation path, and that there are no other fluctuations paths than these.
This non-uniqueness phenomenon can be understood as the $\nu\to 0$ limit of the one exhibited in \cite{bdgjlw2}. The adjoint evolution problem for the reversed path here is quite different from prior results and consists of two distinct phases. In a first
step an antishock $\varphi(\rho_l)>\varphi(\rho_r)$ is instantly created at $y$ and travels up to one of the boundaries at Rankine-Hugoniot speed. During this step
the solution is entropic outside this antishock and follows space-reversed BLN boundary conditions $\rho_l,\rho_r$ at each boundary. Once the antishock reaches a boundary
there is an abrupt change in the rule. Suppose e.g. $f(\rho_l)<f(\rho_r)$, so the stationary state is $\rho_l$. Then the solution becomes entropic inside and the BLN data $\rho_r$ at one boundary changes to $\rho_l$, so that $\rho_l$ is now imposed on both sides and the solution eventually relaxes to $\rho_l$. On the phase transition line $f(\rho_l)=f(\rho_r)$ the antishick does not move, and the stationary state eventually reached is the shock $(\rho_l,\rho_r)$ at initial minimizing position $y$. \\ \\
The paper is organized as follows. In Section \ref{sec_results}, we define the setting and state the main results. Proofs are carried out in Sections  \ref{sec_shock}
and \ref{sec_relax}, respectively for the shock regime $\rho_l<\rho_r$ and the rarefaction regime $\rho_l\geq\rho_r$.
These two sections are similarly structured into three subsections,
that respectively contain preliminary computations and each of the two inequalities between the quasi-potential and the static functional.
In appendix \ref{dynvar}, we establish a variational formula
for the dynamic functional, which implies it is a good functional in the terminology of LD theory. Eventually, in Appendices \ref{section_explicit_1} and \ref{section_explicit_2}, we explicitely compute fluctuation paths leading to arbitrary uniform profiles, which involves interaction of Riemann waves issued from the boundaries.
\section{The setting and results}\label{sec_results}
%
%
%
%
%
%
Let $(a,b)$ be a time interval, $-\infty\leq a<b\leq +\infty$.
We define an action functional $I_{(a,b)}$ on
the set $L^{\infty,K}((a,b)\times(0,1))$ of $[0,K]$-valued Borel functions $\rho(.,.)$ on $(a,b)\times(0,1)$.
This functional consists of a bulk term $I^0$
and boundary terms $I^{l}$ and $I^{r}$:
\be
\label{decompose_functional}I_{(a,b)}=I^{l}_{(a,b)}+I^{0}_{(a,b)}+I^{r}_{(a,b)}\ee
For its definition and basic properties, we assume first that $f$ in \eqref{conservation_law} is a function of class $C^1$  that is not affine on any nonempty open subinterval of $[0,K]$.\\ \\
\textbf{Dynamical functional: bulk part.}
An entropy is a convex function 
$\eta:[0,K]\to\R$, and the associated entropy flux $q$ is defined up
to an additive constant by
\be \label{def_entropy_flux} q'(\rho)=\eta'(\rho)f'(\rho) \ee
$(\eta,q)$ is called an entropy-flux pair. 
The Kru\v{z}kov entropy-flux pair (\cite{kru}) is given by
\be\label{kruzkov_entropy}\eta_v(\rho):=(\rho-v)^+,\quad
q_v(\rho):=1_{\rho>v}[f(\rho)-f(v)]
\ee
Let $\Omega$ be an open subset of $\R^2$.
The $\eta$-entropy production of $\rho(.,.)\in L^{\infty,K}(\Omega)$ is the distribution $\mu_\eta[\rho(.,.)]$ on $\Omega$ defined by
\be \label{production} \mu_\eta[\rho(.,.)]:=\partial_t
\eta(\rho(t,x))+\partial_x q(\rho(t,x)) \ee
In particular, $(id,f)$ is an entropy-flux pair, and $\rho(.,.)$ is
a weak solution to \eqref{conservation_law} in $\Omega$ iff.
$\mu_{id}[\rho(.,.)]=0$.
Notice the behavior of $\mu_\eta$ under space-time reversal:
\be\label{reversal_prod}
\mu_\eta[\rho\circ\Phi]=-\mu_\eta[\rho]\circ\Phi,\quad
\Phi:(t,x)\mapsto(-t,1-x)
\ee
%
%
%
%
%
Let $\overline{M}(\Omega)$ denote the set of Radon measures $\mu$ on $\Omega$ that are locally bounded up to the boundary of $\Omega$, i.e. $|\mu|(K\cap\Omega)<+\infty$ for every compact subset $K$ of $\R^2$. We extend the definition of \cite{bbmn} to the boundary, by calling  $\rho(.,.)\in L^{\infty,K}(\Omega)$ an entropy-measure solution of \eqref{conservation_law},
iff. it is a weak solution of \eqref{conservation_law} in $\Omega$ such that $\mu_\eta[\rho(.,.)]\in\overline{M}(\Omega)$ for every $C^2$ entropy $\eta$. 
The set of entropy-measure solutions is denoted by $\emsol(\Omega)$.
By a straightforward extension of  (\cite[Proposition 2.3]{bbmn}), $\rho(.,.)\in\emsol(\Omega)$ iff. there exists a bounded measurable mapping
$v\mapsto m_{\rho(.,.)}(v;dt,dx)$ from $[0,K]$ to $\overline{M}(\Omega)$ such that, for every $C^2$ entropy $\eta$,
\be\label{kinetic}
\mu_\eta[\rho(.,.)](dt,dx)=\int_0^K\eta''(v) m_{\rho(.,.)}(v;dt,dx)dv
\ee
%
%
%
%
$m_{\rho(.,.)}$ is the kinetic entropy defect measure of $\rho(.,.)$ in the sense of \cite{plt}. 
%
%
%
In the sequel, $(h,g)$ denotes a given entropy-flux pair with uniformly convex $h\in C^2([0,K])$, i.e. $h''\geq c>0$.
For notational simplicity, we often write $\mu_\eta$ for 
$\mu_\eta[\rho(.,.)]$, $\mu$ for $\mu_h$, and $m$ for $m_{\rho(.,.)}$. 
Now assume $\Omega=(a,b)\times(0,1)$. We set
\be\label{functional_bulk}
I^0_{(a,b)}(\rho(.,.))=\left\{
\ba{lll}
+\infty & \mbox{if} & \rho(.,.)\not\in\emsol((a,b)\times(0,1))\\
\dsp\int_{[0,K]}\int_{(a,b)\times(0,1)} h''(v)m^+(v;dt,dx)dv & \mbox{if} &
\rho(.,.)\in\emsol((a,b)\times(0,1))
\ea
\right.
\ee
where $m^+(v;dt,dx)$ denotes the positive part of $m(v;dt,dx)$. 
%
Thus, zeroes of $I^0_{(a,b)}$ are exactly entropy
solution to \eqref{conservation_law} in $(a,b)\times(0,1)$. This functional was introduced in \cite{bbmn} and shown in
\cite{mar} to be the LD rate function for a  stochastic
perturbation of \eqref{conservation_law}. \\ \\ 
%
%
%
%
%
%
%
%
%
\textbf{Dynamical functional: boundary part.}
We set
\begin{eqnarray}
\label{functional_left} I^{l}_{(a,b)}[\rho(.,.)] & = &
\int_a^b i^l(\rho(t,0^+),\rho_l)dt\\ \nonumber \\
\label{functional_right} I^{r}_{(a,b)}[\rho(.,.)] & = & \int_a^b
i^r(\rho(t,1^-),\rho_r)dt
\end{eqnarray}
with functions $i^l$ and $i^r$ defined below. We point out that \eqref{functional_left}--\eqref{functional_right} makes sense for 
$\rho(.,.)\in\emsol((a,b)\times(0,1))$. Indeed, $\emsol(\Omega)$ is contained in the set $\qsol(\Omega)$
of quasi-solutions of \eqref{conservation_law} introduced in \cite{pan, pan2}, i.e. functions $\rho(.,.)\in L^{\infty,K}(\Omega)$
such that $\mu_{\eta_v}[\rho(.,.)]\in \overline{M}(\Omega)$
for a dense subset of $v\in[0,K]$. For such functions, existence of  boundary traces in $L^1$ sense (see \eqref{lonelims} below) is
established in \cite{pan2}.
Given an entropy-flux pair $(\eta,q)$ and a continuity point $\rho_0\in[0,K]$ of $\eta$,
we define the relative entropy-flux pair by
\begin{eqnarray}
\label{relative_entropy} \eta(\rho,\rho_0) & := &
\eta(\rho)-\eta(\rho_0)-\eta'(\rho_0)(\rho-\rho_0)\\
\label{relative_flux} g(\rho,\rho_0) & := &
q(\rho)-q(\rho_0)-\eta'(\rho_0)[f(\rho)-f(\rho_0)]
\end{eqnarray}
In particular, for $(\eta_v,q_v)$ in \eqref{kruzkov_entropy}, $\eta_v(\rho,\rho_0)$ and $q_v(\rho,\rho_0)$ are well defined for $v\neq\rho_0$, and, for every $C^2$ entropy-flux pair $(\eta,q)$,
\be\label{decomp_relative}
\eta(\rho,\rho_0)=\int_0^K \eta''(v)\eta_v(\rho,\rho_0)dv,\quad
q(\rho,\rho_0)=\int_0^K \eta''(v)q_v(\rho,\rho_0)dv
\ee
We then set in \eqref{functional_left}--\eqref{functional_right}
\begin{eqnarray}\label{boundarycost_1}
i^l(\rho,\rho_l) & := & \int_0^K h''(v)\max[q_v(\rho,\rho_l),0]dv\\ 
\label{boundarycost_2}
i^r(\rho,\rho_r) & := & \int_0^K h''(v)\max[-q_v(\rho,\rho_r),0]dv
\end{eqnarray}
Define the sets
\begin{eqnarray}\label{adset_left}\bln^+(\rho_l) & = &
\{\rho\in[0,K]:\,i^l(\rho,\rho_l)=0\}\\
%
%
\label{adset_right}\bln^-(\rho_r) & = &
\{\rho\in[0,K]:\,i^r(\rho,\rho_r)=0\}
%
%
\end{eqnarray}
It follows from \eqref{decomp_relative} that $\rho\in\bln^+(\rho_l)$ (resp. $\rho\in\bln^-(\rho_r))$  iff. $q(\rho,\rho_l)\leq 0$ (resp. $-q(\rho,\rho_r)\leq 0)$ for every entropy-flux pair $(\eta,q)$. Thus (see e.g. \cite[Chapter 15]{ser}) $\rho(t,0^+)\in\bln^+(\rho_l)$, $\rho(t,1^-)\in\bln^-(\rho_r)$ are equivalent to BLN boundary conditions 
for \eqref{boundary_data}.\\ \\ 
In Appendix \ref{dynvar}, we establish a variational expression for $I_{(a,b)}$ and prove that $I_{(a,b)}$ is l.s.c. with compact level sets with respect to 
the local $L^1$-topology on $L^{\infty,K}((a,b)\times(0,1))$. One can show (see \cite{pan2} and beginning of Subsection \ref{proof_inequality}) that $\{I^0_{(a,b)}<+\infty\}$  
is contained  in $C^0([a,b]\cap\R, L^1(0,1))$ .\\ \\
%
%
\textbf{Stationary set.} The set $\mathcal S$ of stationary solutions of \eqref{conservation_law}--\eqref{boundary_data} is studied in \cite{ba} and the following results established.
Let $\mathcal M$ denote the set of minimizers (resp. maximizers) of $f$ on $[\rho_l,\rho_r]$ (resp. $[\rho_r,\rho_l])$ if $\rho_l\leq\rho_r$ (resp. $\rho_l\geq\rho_r$). 
Then we have
\begin{theorem}\label{th_stat}(\cite{ba})
(i) If $\rho_l\leq\rho_r$ (resp. $\rho_l\geq\rho_r$), $\mathcal S$ is the set of piecewise constant, nondecreasing (resp. nonincreasing), $\mathcal M$-valued functions $\rho_s(.)$
on $[0,1]$. In particular, if $\mathcal M$ is a singleton, $\mathcal S$ is a singleton whose unique element is the function $\rho_s(.)$ with constant value
\be\label{optimizer}
R(\rho_l,\rho_r)  :=  \left\{
\ba{lll}
{\rm argmin}_{[\rho_l,\rho_r]}f & \mbox{if} & \rho_l\leq\rho_r\\
{\rm argmax}_{[\rho_l,\rho_r]}f & \mbox{if} & \rho_l\geq\rho_r
\ea
\right.
\ee
(ii) If $\mathcal M$ is a singleton, the entropy solution $\rho(t,.)$ to \eqref{conservation_law}--\eqref{boundary_data}
with Cauchy datum $\rho_0(.)\in L^{\infty,K}(0,1)$ converges to $\rho_s(.)$ in $L^1(0,1)$ as $t\to\infty$.  If in addition $R(\rho_l,\rho_r)$ is not a local extremum of $f$, there exists $T\in[0,+\infty)$ such that $\rho(t,.)=\rho_s(.)$ for all $t\geq T$. 
\end{theorem}
\textbf{Exclusion-like case.} From now on we assume that $f$ is an ``exclusion-like'' flux function, i.e. satisfies \eqref{properties_flux}.
Thus there exists a unique decreasing function $\varphi:[0,K]\to[0,K]$ such that
\be \label{symmetry_f} f\circ\varphi=f\ee
In this case, by elementary computations, one finds the following expressions for \eqref{boundarycost_1}--\eqref{boundarycost_2}, \eqref{adset_left}--\eqref{adset_right} and $\mathcal S$.\\ \\
If $0\leq\rho_l\leq\rho^*$, then
\be \label{left_1} i^l(\rho,\rho_l)=\left\{
\ba{lll}
0 & \mbox{if} & \rho=\rho_l\mbox{ or }\varphi(\rho_l)\leq\rho\leq
K\\
g(\rho,\rho_l) & \mbox{if} & 0\leq\rho\leq \rho^*\\
g(\varphi(\rho),\rho_l) & \mbox{if} &
\rho^*\leq\rho\leq\varphi(\rho_l)
\ea
\right.\ee
If $\rho^*\leq\rho_l\leq K$, then
\be\label{left_2} i^l(\rho,\rho_l)=\left\{
\ba{lll}
0 & \mbox{if} & \rho^*\leq\rho\leq K \\
g(\rho,\rho_l) & \mbox{if} & 0\leq\rho\leq\varphi(\rho_l)\\
g(\rho)-g(\varphi(\rho)) & \mbox{if} &
\varphi(\rho_l)\leq\rho\leq\rho^*
\ea
\right. \ee
If $0\leq\rho_r\leq\rho^*$, then
\be \label{right_1}
i^r(\rho,\rho_r)=\left\{
\ba{lll}
0 & \mbox{if} & 0\leq\rho\leq\rho^*\\
-g(\rho,\rho_r) & \mbox{if} & \varphi(\rho_r)\leq\rho\leq K\\
g(\varphi(\rho))-g(\rho) & \mbox{if} &
\rho^*\leq\rho\leq\varphi(\rho_r)
\ea
\right.
\ee
If $\rho^*\leq\rho_r\leq K$, then
\be \label{right_2} i^r(\rho,\rho_r)=\left\{
\ba{lll}
0 & \mbox{if} & \rho=\rho_r\mbox{ or
}0\leq\rho\leq\varphi(\rho_r)\\
-g(\rho,\rho_r) & \mbox{if} & \rho^*\leq\rho\leq 1\\
-g(\varphi(\rho),\rho_r) & \mbox{if} &
\varphi(\rho_r)\leq\rho\leq\rho^*
\ea
\right. \ee
\eqref{left_1}--\eqref{right_2} were introduced in \cite{bd2} in the special case \eqref{asep_case}. For sets of admissible boundary values, we obtain
\be\label{adset_left_2}\bln^+(\rho_l)=
\left\{
\ba{lll}
\{\rho_l\}\cup[\varphi(\rho_l),K] & \mbox{ if } & \rho_l\leq\rho^*\\
\left[
\rho^*,K \right] & \mbox{ if } & \rho_l\geq\rho^*
\ea
\right.
\ee
\be\label{adset_right_2}
\bln^-(\rho_r)=
\left\{
\ba{lll}
[0,\rho^*] & \mbox{ if } & \rho_r\leq\rho^*\\
\{\rho_r\}\cup[0,\rho^*] & \mbox{ if } & \rho_r\geq\rho^*
\ea
\right.
\ee
Finally, the stationary set $\mathcal S$ reduces to the four phases introduced in \cite{dehp}:
\be \label{eq_stat} {\mathcal S}=\left\{
\ba{lll}
\{\rho_s(.)\equiv \rho_l\} & \mbox{if} & \rho_l<\rho_r \mbox{ and }f(\rho_l)<f(\rho_r) \\
& \mbox{or} & \rho^*>\rho_l\geq \rho_r  \mbox{ and } f(\rho_l)\geq f(\rho_r)\\ \\
\{\rho_s(.)\equiv\rho_r\} & \mbox{if} & \rho_l<\rho_r \mbox{ and } f(\rho_l)>f(\rho_r) \\
& \mbox{or} & \rho_l\geq \rho_r>\rho^*  \mbox{ and } f(\rho_r)\geq f(\rho_l)\\ \\
\{\rho_s(.)\equiv\rho^*\} & \mbox{if} & \rho_l\geq\rho^*\geq\rho_r\\
%
\{\rho_s^y(.),\,y\in[0,1]\} & \mbox{if} & \rho_l\leq\rho_r\mbox{ and
} f(\rho_l)=f(\rho_r)
\ea
\right.
\ee
where
\be \label{arbitrary_shock}
\rho_s^y(.)=\rho_l{\bf 1}_{(0,y)}+\rho_r{\bf 1}_{(y,1)},\quad
y\in[0,1] \ee
The four cases in \eqref{eq_stat} are respectively (\cite{dehp}) the low-density
(LD), high-density (HD) and maximal current (MC) phase, and the
phase transition line between LD and HD.\\ \\
\textbf{Static functional.}
We now make the symmetry assumption (see \eqref{symmetry_f})
\be \label{symmetry} h''(\rho)=-h''(\varphi(\rho))\varphi'(\rho)
\ee
Since \eqref{symmetry} and the definition of $I_{(a,b)}$ is unchanged by adding a linear term to the entropy $h$, we may assume
without loss of generality that $h=h(\rho,\rho^*)$ (see \eqref{relative_entropy}) is the relative entropy with respect to $\rho^*$. With this choice
\eqref{symmetry} is equivalent to $h'(\varphi(\rho))=-h'(\rho)$. Let $\theta=h'(\rho)$ be the chemical potential. The Legendre-Fenchel dual $L(\theta)=h^*(\theta)$
can be interpreted microcsopically as the moment generating function under the equilibrium invariant measure of an underlying particle system with equilibrium entropy $h$.
The latter can be thought of as the equilibrium LD rate function for the empirical particle density.
Let $j(\theta)=f(h'^{-1}(\theta))$ be the flux as a function of the chemical potential. Then \eqref{symmetry} is equivalent to 
\be\label{symmetry_2}j(-\theta)=j(\theta)\ee
Let
\be \label{def_k} k(\rho):=\int_0^\rho\varphi(r)h''(r)dr,\quad
K(\rho):=h(\rho)-\rho h'(\rho)+k(\rho) \ee
%
%
It follows from
\eqref{symmetry} that
\be\label{kk}
k(\rho)=\varphi(\rho)h'(\rho)+h[\varphi(\rho)],\quad
K(\rho)=[\varphi(\rho)-\rho]h'(\rho)+h(\rho)+h[\varphi(\rho)]
\ee
%
%
%
and
\be \label{symmetry_K} K(\varphi(\rho))=K(\rho) \ee
$K$ is increasing
on $[0,\rho^*]$ and decreasing on $[\rho^*,K]$, and there exists a unique function
$L:[0,f(\rho^*)]\to\R$ such that
$ K(\rho)=L(f(\rho))$. $L$ is an increasing function.
We now define on $L^{\infty,K}((0,1))$ a generalization $S[.]$ of the stationary LD functional introduced in \cite{dls}.\\ \\
\emph{First case: $\rho_l<\rho_r$.} For $\rho(.)\in
L^{\infty,K}((0,1))$, and $y\in[0,1]$, set
\be\label{def_S_1}
S[\rho(.),y]:=\int_0^y[h(\rho(x))-\rho(x)h'(\rho_l)+k(\rho_l)]dx +
\int_y^1[h(\rho(x))-\rho(x)h'(\rho_r)+k(\rho_r)]dx \ee
%
%
\be \label{def_S_2}
S[\rho(.)]:=\inf_{y\in[0,1]}S[\rho(.),y]-\min(K(\rho_l),K(\rho_r))
\ee
We denote by $\mathcal Y[\rho(.)]$ the set of minimizers $y$ in
\eqref{def_S_2}, that is the set of
minimizers $y$ of $y\mapsto \int_0^y\rho(z)dz-\rho_c y$, where
$\rho_c$ is given by
\be\label{def_critical}
\rho_c:=\frac{k(\rho_r)-k(\rho_l)}{h'(\rho_r)-h'(\rho_l)} \ee
\emph{Second case: $\rho_l\geq\rho_r$.} For $F(.)\in
L^{\infty,K}((0,1))$, set
\be \label{def_S_3} S[\rho(.),F(.)]:=\int_0^1 [h(\rho(x))-\rho(x)
h'(F(x))+k(F(x))]dx \ee
\be \label{def_S_4} S[\rho]:=\sup_{F\in{\mathcal
F}(\rho_l,\rho_r)}S[\rho(.),F(.)]-\sup_{\rho\in[\rho_r,\rho_l]}K(\rho)
\ee
where ${\mathcal F}(\rho_l,\rho_r)$ is the set of nonincreasing
functions $F$ on $[0,1]$ such that $\rho_l\geq F(0+)\geq \rho_r\geq
F(1^-)$. As in \cite{dls}, the unique minimizer in \eqref{def_S_4}, which we will denote
by $F_{\rho(.)}$, can be described explicitely by an envelope construction, see beginning of Subsection \ref{lower_relax}.
One shows as in \cite{dls} that $S[\rho(.)]=0$ for
$\rho(.)\in{\mathcal S}$, and $S[\rho(.)]>0$ for
$\rho(.)\not\in{\mathcal S}$.
It is not difficult to see that the functionals \eqref{def_S_2}--\eqref{def_S_4} are continuous with respect to $L^1((0,1))$ norm
and l.s.c. in weak* toplogy $\sigma(L^{\infty,K}(0,1),L^1((0,1))$. \\ \\
\textbf{Quasi-potential.} Let ${\mathcal R}[\rho(.)]$ denote the set of functions
$\rho(.,.)\in L^{\infty,K}((-\infty,0)\times(0,1))$ such that
$\lim_{t\to 0}\rho(t,.)=\rho(.)$, and 
$\lim_{t\to
-\infty}\rho(t,.)=\rho_s(.)\in{\mathcal S}$ in $L^1((0,1))$.
%
%
For $T\in(-\infty,0)$, let ${\mathcal R}_T[\rho(.)]$ denote the set
of functions $\rho(.,.)\in L^{\infty,K}((-\infty,0)\times(0,1))$
such that $\lim_{t\to 0}\rho(t,.)=\rho(.)$ in $L^1((0,1))$, and there exists $\rho_s(.)\in\mathcal S$ for which $\rho(T,.)=\rho_s(.)$ for all $t\leq T$. Define  
\begin{eqnarray}\label{quasi_potential}
V[\rho(.)] & := & \inf_{\rho(.,.)\in{\mathcal
R}[\rho(.)]}I_{(-\infty,0)}[\rho(.,.)]\\ 
V_T[\rho(.)] & := & \inf_{\rho(.,.)\in{\mathcal
R}_T[\rho(.)]}I_{(-\infty,0)}[\rho(.,.)]
\label{quasi_potential_finite}
\end{eqnarray}
We can now state our main results. The first theorem identifies the quasi-potential with the static functional, and states that it is achieved in finite time outside the MC phase. The next two theorems characterize optimal paths for the quasi-potential.
\begin{theorem}
\label{theorem_1}
For every $\rho(.)\in L^{\infty,K}((0,1))$,
\be \label{variational} 
%
S[\rho(.)]=V[\rho(.)]
%
%
\ee
Outside the MC phase, we also have
\be\label{finite_time} 
%
S[\rho(.)]=\min_{T<0}V_T[\rho(.)]
 \ee
\end{theorem}
%
%
%
We denote by $v(\rho^-,\rho^+)$
the Rankine-Hugoniot or velocity of a discontinuity
($\rho^-\neq\rho^+$)
%
%
in \eqref{conservation_law}, given by the flux continuity relation
\be\label{speed_rankine}
f(\rho^+)-v\rho^+=f(\rho^-)-v\rho^-
\ee
\begin{theorem}
\label{theorem_2} Assume $\rho_l<\rho_r$. Then there exists a
mapping
$$(\tilde{\rho}(.)\in L^{\infty,K}((0,1)),\tilde{y}\in\tilde{\mathcal Y}[\tilde{\rho}(.)])\mapsto
\tilde{\rho}(.,.)=M^{\tilde{y}}[\tilde{\rho}(.)]\in
L^{\infty,K}((0,+\infty)\times(0,1))$$
with the following properties, where $\tilde{\mathcal
Y}[\tilde{\rho}(.)]$ denotes the set of maximizers of
$y\mapsto\int_0^y\tilde{\rho}(z)dz-\rho_c
y$.\\ \\
(0) $M^{\tilde{y}^1}[\tilde{\rho}(.)]\neq M^{\tilde{y}^2}[\tilde{\rho}(.)]$ for two distinct $\tilde{y}^1,\tilde{y}^2\in\tilde{\mathcal Y}[\tilde{\rho}(.)]$.\\ \\
(1) $\rho(.,.)\in L^{\infty,K}((-\infty,0)\times(0,1))$ achieves equality in
\eqref{quasi_potential} and \eqref{variational} iff. it is of the form
\be \label{reversal} \rho(t,x)=\tilde{\rho}(-t,1-x)\ee
where $\tilde{\rho}(.,.)=M^{1-y}[\tilde{\rho}(.)]$ for some $y\in\mathcal Y[\rho(.)]$, with $\tilde{\rho}(.)=\rho(1-.)$.
Then we have $\mathcal Y[\rho(t,.)=\{y_t\}$ for all $t<0$, where $y_t=1-\tilde{y}_{-t}$, $y_0=y$, with $\tilde{y}_.$ given by \eqref{def_tildey} below.\\ \\
(2) $M^{\tilde{y}}[\tilde{\rho}(.,.)]$ is the unique weak solution of
\eqref{conservation_law} in
$\emsol((0,+\infty)\times(0,1))$ that satisfies the following conditions (a)--(c). Besides, 
$\tilde{\mathcal Y}[\tilde{\rho}(t,.)]=\{\tilde{y}_t\}$ holds for all $t>0$, with $\tilde{y}_t$ given by \eqref{def_tildey}.\\ \\
(a) $\tilde{\rho}(t,.)\to\tilde{\rho}(.)$ in $L^1((0,1))$ as $t\to 0$.\\ \\
(b) Let $\Gamma=\{
(t,\tilde{y}_t):t\in(0,\tilde{\theta}^{\tilde{y}})\}$, where
\be\label{def_tildey}\tilde{y}_t:=\min(1,(\tilde{y}+vt)^+),\quad
t\in[0,+\infty)\ee
\be \label{def_vshock} v:=v(\varphi(\rho_l),\varphi(\rho_r)) \ee
\be \label{def_time} \tilde{\theta}^{\tilde{y}}:=
\max(-\tilde{y}/v,(1-\tilde{y})/v)\in[0,+\infty]\ee
with $0/0=0$, is the time for $\tilde{y}_.$ to reach a boundary.
Then $\tilde{\rho}(.,.)$ is an entropy solution to
\eqref{conservation_law} in
$[(0,\tilde{\theta}^{\tilde{y}})\times(0,1)]\backslash\Gamma$, and satisfies
\be\label{reverse_bln} \tilde{\rho}(t,0^+)\in{\mathcal
E}^+(\rho_r),\quad \tilde{\rho}(t,1^-)\in\bln^-(\rho_l)\ee
\be \label{antishock_0}
\tilde{\rho}(t,\tilde{y}_t-0)=\varphi(\rho_l)>\tilde{\rho}(t,\tilde{y}_t+0)=\varphi(\rho_r)
\ee
(c) If $\tilde{\theta}^{\tilde{y}}<+\infty$, i.e. $f(\rho_l)\neq f(\rho_r)$ or $\tilde{y}\in\{0,1\}$, $\tilde{\rho}(.,.)$ is an entropy solution to
\eqref{conservation_law} in
$(\tilde{\theta}^{\tilde{y}},+\infty)\times (0,1)$, and satisfies
boundary conditions
\be \label{reverse_bln_2}
\tilde{\rho}(t,\tilde{y}_\infty^\sigma)=\rho_b, \quad
\tilde{\rho}(t,(1-\tilde{y}_\infty)^{-\sigma})\in{\mathcal
E}^{-\sigma}(\rho_b)\ee
where
\be \label{final_boundary}
\tilde{y}_\infty=\tilde{y}_{\tilde{\theta}^{\tilde{y}}}=\left\{
\ba{lll}
0 & \mbox{if} & f(\rho_l)<f(\rho_r)\mbox{ or
}f(\rho_l)=f(\rho_r),\,\tilde{y}=0\\
1 & \mbox{if} & f(\rho_l)>f(\rho_r)\mbox{ or
}f(\rho_l)=f(\rho_r),\,\tilde{y}=1
\ea
\right. \ee
is the boundary reached by $\tilde{y}_.$ at time
$\tilde{\theta}^{\tilde{y}}$,
$\sigma=+$ and $-\sigma=-$ if $\tilde{y}_\infty=0$, $\sigma=-$
and $-\sigma=+$ if $\tilde{y}_\infty=1$, $\rho_b=\rho_l$ in the
first case of \eqref{final_boundary}, and $\rho_b=\rho_r$ in the
second case of \eqref{final_boundary}.
%
%
%
%
%
%
%
%
%
\end{theorem}
\begin{theorem}
\label{theorem_3}
Assume $\rho_l\geq\rho_r$. Then there is a unique 
${\rho}(.,.)\in L^{\infty,K}((-\infty,0)\times(0,1))$ that achieves equality in \eqref{quasi_potential} and \eqref{variational}, which is defined as follows. Let
\be \label{initial_tildeg}
\tilde{G}(0,x)=\varphi[F_{\rho(.)}(1-x)]
\ee
Define
$\tilde{G}(.,.)$ on $(0,+\infty)\times\R$ as the entropy solution to
\eqref{conservation_law} with initial datum
\be \label{barg} \bar{G}(0,.)=\varphi(\rho_r){\bf
1}_{(-\infty,0)}+\tilde{G}(0,.){\bf 1}_{(0,1)}+\varphi(\rho_l){\bf
1}_{(1,+\infty)} \ee
The restriction of $\tilde{G}(.,.)$ to $(0,+\infty)\times(0,1)$ is the entropy solution to \eqref{conservation_law} with initial datum $\tilde{G}(0,.)$ and BLN
boundary conditions
\be \label{bln_tildeg}
\tilde{G}(t,0^+)\in\bln^+(\varphi(\rho_r)),\quad
\tilde{G}(t,1^-)\in\bln^-(\varphi(\rho_l))
\ee
Then $\rho(.,.)$ is given by \eqref{reversal}, where
$\tilde{\rho}(t,x)$ is the unique entropy solution to
\eqref{conservation_law} with initial datum
\be\label{initial_tilderho}
\tilde{\rho}(0,.)=\rho(1-.)\ee
and BLN boundary data
\be \label{bln_f} \tilde{\rho}(t,0^+)\in{\mathcal
E}^+[\varphi(\tilde{G}(t,0^+))], \quad
\tilde{\rho}(t,1^-)\in\bln^-[\varphi(\tilde{G}(t,1^-))]
\ee
Besides, \eqref{initial_tildeg} remains true at later
times, i.e.:
\be \label{later_tildeg}
\tilde{G}(t,x)=\varphi[F_{\rho(-t,1-.)}(1-x)],\quad \forall
t>0,\,x\in(0,1)
\ee
\end{theorem}
A functional closely related to \eqref{functional_bulk} was introduced in \cite{jen,var} as the LD rate function for TASEP on $\Z$.
Namely,
\be\label{jv}
J^0_{(a,b)}[\rho(.,.)]=\left\{
\ba{lll}
+\infty & \mbox{if} & \rho(.,.)\not\in\emsol_h((a,b)\times(0,1))\\
\dsp\mu^+_h((a,b)\times(0,1)) & \mbox{if} & \rho(.,.)\in \emsol_h((a,b)\times(0,1))\ea
\right.
\ee
where $\emsol_h(\Omega)$ is the set of
weak solutions of \eqref{conservation_law} such that $\mu_h$ lies in $\overline{M}(\Omega)$.
By convexity,
\be\label{compare_bulk}I^0_{(a,b)}\geq J^0_{(a,b)}\quad\mbox{on
}L^{\infty,K}((a,b)\times(0,1))\ee
If $f$ is strictly convex or concave, the result of \cite{dow}
implies that zeroes of $J^0_{(a,b)}$ are exactly
entropy solution to \eqref{conservation_law} in $(a,b)\times(0,1)$. For such $f$ one has
\be\label{compare_bulk_bv} J^0_{(a,b)}=I^0_{(a,b)}\quad\mbox{ on }{BV}_{loc}((a,b)\times(0,1))
\ee
%
%
where ${BV}_{loc}(\Omega)$ denotes the set of functions $\rho(.,.)\in L^{\infty}(\Omega)$ such that $\partial_t\rho$ and $\partial_x\rho$ lie in ${M}(\Omega)$, the set of Radon measures on $\Omega$.
If $f$ is neither strictly convex nor strictly concave, $I^0$ and $J^0$ are quite different, as a single entropy is not enough to characterize entropy solutions (\cite{lf}). 
%
In fact we will prove our results more generally for a bulk functional $I^0_{(a,b)}:L^{\infty,K}((a,b)\times(0,1))\to[0,+\infty]$ which enjoys the following
properties: (i)
$J^0_{(a,b)} \leq I^0_{(a,b)}$ on $L^{\infty,K}((a,b)\times(0,1))$;
(ii) $I^0_{(a,b)}=J^0_{(a,b)}$  on
${BV}_{loc}((a,b)\times(0,1))$;
(iii) $(I^0)^{-1}_{(a,b)}([0,+\infty))\subset\emsol((a,b)\times(0,1))$.\\ \\
%
%
We will denote by
$J_{(a,b)}$ the functional obtained by replacing $I^0$ with $J^0$ in \eqref{decompose_functional}. 
Since we do not know if $J^0$ satisfies assumption iii), in view of \eqref{functional_left}--\eqref{functional_right}, $J_{(a,b)}$ is not defined a priori on $L^{\infty,K}((a,b)\times(0,1))$, but only on $\emsol((a,b)\times(0,1))$.
Our results would be valid for $J_{(a,b)}$ if
condition (iii) could be proved in this case. The result of \cite{dow} suggests the possible conjecture 
$\emsol_h((a,b)\times(0,1))=\emsol((a,b)\times(0,1))$ for any strictly convex entropy $h$, which would imply iii) for $J_{(a,b)}$. 
\section{The case $\rho_l<\rho_r$}\label{sec_shock}
\subsection{Analysis of entropy production}
\label{proof_inequality}
Let $\rho(.,.)\in\emsol((a,b)\times(0,1))$. It will be convenient to extend it by giving it constant uniform values $\rho_0,\rho_1$, respectively for $x<0$ and $x>1$. 
This extension is in general no longer a weak solution of \eqref{conservation_law}. However, by \cite[Proposition 2.1]{cf}, it still lies in the set $\gsol((a,b)\times\R)$
of generalized entropy solutions in the sense of \cite{dow2}, i.e. functions $\rho(.,.)\in L^{\infty,K}((a,b)\times\R)$ such that $\mu_\eta[\rho(.,.)]\in M((a,b)\times\R)$ for every $C^2$ entropy $\eta$.\\ \\
Recall that a bounded Borel function $F$ on $\R^d$ is said to have an approximate limit $l_n F(x)$ at $x\in\R^d$ along the normal $n\in\R^d$, $||n||=1$,  if
\be\label{def_approxlim}
\lim_{\varepsilon\to 0}\varepsilon^{-d}\int_{B(x,\varepsilon)\cap H_n(x)}|F(y)-l_n F(x)|dy=0
\ee
where $B(x,\varepsilon)$ denotes the open ball of radius $\varepsilon$ centered at $x$ and $H_n(x)=\{y\in\R^d:\,n.(y-x)>0\}$ is the normal hyperplane containing $x$ and $n$.
$x$ is called a jump point of $F$ if, for some $n$,  approximate limits $l_{\pm n} F(x)$ exist and are distinct. One can show that the pair $\{n,-n\}$ is uniquely determined. $x$ is called a point of vanishing mean oscillation (VMO) of $F$ if
\be\label{def_vmo}
\lim_{\varepsilon\to 0}\varepsilon^{-2}\int_{B(x,\varepsilon)}|F(y)-\bar{F}(x,\varepsilon)|dy=0
\ee
where $\bar{F}(x,\varepsilon)$ is the mean value of $F$ on $B(x,\varepsilon)$.\\ \\
Let $(\eta,q)$ be a $C^2$ entropy-flux pair. By \cite{dow2}, $\rho(.,.)\in\gsol((a,b)\times\R)$ has the following property: there are subsets $J$ and $V$ of $(a,b)\times\R$ such that the complement of $J\cup V$ has $\mathcal H^1$-measure $0$,
$J$ is rectifiable with a local normal $n=(n_t,n_x)$, $\rho(.,.)$ has approximate limits $\rho^+\neq\rho^-$ on $J$ ($\rho^+$ being the limit in the positive $n$-half-space), and $\rho(.,.)$ has VMO on $V$. The trace of entropy production on the jump set $J$ is 
\be\label{trace_prod}
\mu_\eta[\rho(.,.)]\llcorner J=\pi_\eta(\rho^-,\rho^+,n)\mathcal H^1\llcorner J
\ee
where 
$$
\pi_\eta(\rho^-,\rho^+,n) = n_x[q(\rho^+)-q(\rho^-)]-n_t[\eta(\rho^+)-\eta(\rho^-)]
$$
For $(\eta,q)=(id,f)$, \eqref{conservation_law} yields $n_x\neq 0$, $\mathcal H^1$-a.e. on $J\cap\{(-\infty,0)\times(0,1)\}$, and the Rankine-Hugoniot relation
\eqref{speed_rankine} for the local speed $v:=-n_t/n_x$. By convention we orient the normal so that $n_x>0$. Note that $n_x\neq 0$ implies no jump
along time lines. Together with \cite{pan} this implies $\emsol((a,b)\times(0,1))\subset C^0([a,b],L^1((0,1))$.
We now focus on $(\eta,q)=(h,g)$ and set
\begin{eqnarray}
\pi(\rho^-,\rho^+) & = & [g(\rho^+)-v(\rho^-,\rho^+)\rho^+]-[g(\rho^-)-v(\rho^-,\rho^+)\rho^-]\\
\label{localprod_11}
& = & -\int_{\rho^-}^{\rho^+}h''(r)[f(r)-f(\rho^\pm)-v(\rho^-,\rho^+)(r-\rho^\pm)]
\end{eqnarray}
Let
\be\label{localprod_2}
\sigma(\rho,\xi):=\int_{\rho_l}^{\rho_r}h''(r)[f(\varphi(r))-f(\rho)-\xi(\varphi(r)-\rho)]dr
\ee
Notice that $\sigma(\rho^-,\xi)=\sigma(\rho^+,\xi)$ if $\xi=v(\rho^-,\rho^+)$. 
\begin{lemma}
\label{lemma_1}
Let $\rho(.,.)\in\emsol((a,b)\times(0,1))$, and $\Gamma=\{(t,y_t):\,t\in(a,b)\}$, where
$y_.:[a,b]\to[0,1]$ is a Lipschitz trajectory. Set $S_t=S[\rho(t,.),y_t]$. Then
\begin{eqnarray}
S_b-S_a-J_{(a,b)}[\rho(.,.)] & = &
\int_a^{b}d(t)dt-\mu^-[(a,b)\times(0,1)\backslash\Gamma]\nonumber\\
& - & \int_a^b 1_{(0,1)}(y_t)\pi(\rho(t,y_t-),\rho(t,y_t+))^-dt\label{difference_gamma}
\end{eqnarray}
where 
\begin{eqnarray}\label{decomp_d}
d(t) &  :=  & d^\circ(t)+\partial d(t)\\
\label{dcirc}
d^\circ(t) & := & 1_{(0,1)}(y_t)\sigma(\rho(t,y_t\pm),\dot{y}_t)
\end{eqnarray}
and $\partial d(t)$ is given respectively when $y_t\in(0,1)$, $y_t=0$, $y_t=1$, by
\be\label{boundaryterm_1}
[g(\rho(t,0^+);\rho_l)-i^l(\rho(t,0^+);\rho_l)]+[-g(\rho(t,1^-);\rho_r)-i^r(\rho(t,1^-);\rho_r)]\label{boundaryterm_gamma_bulk}
\ee
\be\label{boundaryterm_2}
[g(\rho(t,0^+);\rho_r)-i^l(\rho(t,0^+);\rho_l)]+[-g(\rho(t,1^-);\rho_r)-i^r(\rho(t,1^-);\rho_r)]\label{boundaryterm_gamma_left}
\ee
\be\label{boundaryterm_3}
[g(\rho(t,0^+);\rho_l)-i^l(\rho(t,0^+);\rho_l)]+[-g(\rho(t,1^-);\rho_l)-i^r(\rho(t,1^-);\rho_r)]\label{boundaryterm_gamma_right}
\ee
\end{lemma}
\begin{remark}\label{remark_limits}
The above limits can be understood as approximate limits provided the last integral in \eqref{difference_gamma} is restricted to $\Gamma\cap J$. Due to the extension of $\rho(.,.)$, VMO points at the boundaries have approximate limit 
inside, and jump points at the boundaries have a normal parallel to the space axis.
Hence approximate spatial limits are defined a.e. at the boundaries. 
Alternatively, since $\rho(.,.)\in\qsol((a,b)\times(0,1))$, 
the limits in Lemma \ref{lemma_1} exist in the $L^1$ sense (\cite{pan2}): 
\be\label{lonelims}
{\rm ess}\lim_{\varepsilon\to 0}\int_I|\rho(t,\varepsilon)-\rho(t,0^+)|dt={\rm ess}\lim_{\varepsilon\to 0}\int_I|\rho(t,y_t+\varepsilon)-\rho(t,y_t+0)|dt=0
\ee 
for every bounded interval $I\subset(a,b)$.
\end{remark}
%
%
%
%
%
%
We need a version of Lemma \ref{lemma_1} for a mollified version of $S[.]$. Let $\psi:[0,1]\to[0,1]$ be a nonincreasing
function such that $\psi(0)=1$ and $\psi(1)=0$. We
set
\begin{eqnarray}\label{mollified_S}
S_\psi[\rho(.)] & := & \int_0^1 \left\{h(\rho(x))+\psi(t,x)[-\rho(x)h'(\rho_l)+k(\rho_l)]\right\}dx\nonumber\\
& + & \int_0^1 \left\{h(\rho(x))+[1-\psi(t,x)][-\rho(x)h'(\rho_r)+k(\rho_r)]\right\}dx
\label{mollified_S}\\
& = & \int_0^1 -\psi'(y)S[\rho(.),y]dy
\end{eqnarray}
We denote by $\overline{BV}_{loc}(\Omega)$  the set of functions $\rho(.,.)\in L^{\infty}(\Omega)$ such that $\partial_t\rho$ and $\partial_x\rho$ lie in $\overline{M}(\Omega)$.
\begin{lemma}\label{lemma_11}
Let 
%
$\psi\in\overline{BV}_{loc}((a,b)\times(0,1))$.
%
Assume
$\psi(t,0)=1$ and $\psi(t,1)=0$ for a.e. $t\in(a,b)$, $\partial_x\psi\leq 0$, $|\partial_t\psi|\ll-\partial_x\psi$, and $\partial_x\psi(dt,dx)=\partial_x\psi(t,dx)dt$.
Let 
$S_t=S_{\psi(t,.)}[\rho(t,.)]$. 
Then
\be\label{difference_psi}
S_b-S_a-J_{(a,b)}[\rho(.,.)] = \int_a^b d_\psi(t)dt
-\mu^-[(a,b)\times(0,1)]\ee
where (with $\partial d$ given by \eqref{boundaryterm_1})
\begin{eqnarray}
\label{decomp_dpsi}
d_\psi(t) &  :=  & d_\psi^\circ(t)+\partial d(t)\\
\label{dpsicirc}
d_\psi^\circ(t) & := & \int_0^1 \sigma\left(\rho(t,x),
\frac{-\partial_t\psi(t,x)}{\partial_x\psi(t,x)}\right)[-\partial_x\psi](t,dx)\end{eqnarray}
\end{lemma}
The next lemma will be combined with \eqref{difference_gamma}--\eqref{difference_psi} to show that $\partial d(t)$, $d^\circ(t)$ and $d_\psi^\circ(t)$ are nonpositive for suitable choices of $y_.$ and $\psi$, and derive conditions from these quantities being equal to $0$.
\begin{lemma}
\label{inequalities}
Let $\rho^-,\rho^+\in[0,K]$. Then:\\ \\
(o) $\rho^-<\rho^+$ (resp. $>$) iff.
$\pi(\rho^-,\rho^+)<0$ (resp. $>$).\\ \\
(i) (a) Assuming $\rho^-\leq\rho^+$,
\be \label{inequality_1}
\sigma(\rho^\pm,v(\rho^-,\rho^+))+\pi(\rho^-,\rho^+)\leq 0
\ee
with equality iff. $\rho^-=\varphi(\rho_r)$
and $\rho^+=\varphi(\rho_l)$. (b)
\be\label{uniformbound_sigma}
\sup_{\rho\in[0,K]}\sigma(\rho,f'(\rho))<0
\ee
(ii) $i^l(\rho,\rho_l)\geq g(\rho,\rho_l)$, with equality iff.
$\rho\in\bln^-(\rho_l)$; $i^r(\rho,\rho_r)\geq
-g(\rho,\rho_r)$, with equality iff. $\rho\in\bln^+(\rho_r)$.\\ \\
(iii) Assume $\rho\leq\rho^*$. Then $-g(\rho,\rho_l) \leq
i^r(\rho,\rho_r)$, and
\be \label{equivalence_1}
-g(\rho,\rho_l)=i^r(\rho,\rho_r)\Leftrightarrow f(\rho_l)\leq
f(\rho_r)\mbox{ and }\rho=\rho_l \ee
Under \eqref{equivalence_1}, we have $-g(\rho,\rho_l)
=i^r(\rho,\rho_r)=0$
.\\ \\
(iv) Assume $\rho\geq\rho^*$. Then $g(\rho,\rho_r)\leq
i^l(\rho,\rho_l)$, and
\be \label{equivalence_2} g(\rho,\rho_r)=
i^l(\rho,\rho_l)\Leftrightarrow f(\rho_l)\geq f(\rho_r)\mbox{ and
}\rho=\rho_r \ee
Under \eqref{equivalence_2}, we have
$g(\rho,\rho_r)=i^l(\rho,\rho_l)=0$.
\end{lemma}
In the proof of Lemma \ref{inequalities} and various other places we use the following lemma, which is immediate from assumptions \eqref{properties_flux} on $f$,
\eqref{def_entropy_flux} and \eqref{symmetry_f}.
\begin{lemma}
\label{lemma_g}
(i) Let $\rho_0\leq\rho^*$. Then
$g(\rho,\rho_0)>0$ for $\rho\leq\rho^*,\,\rho\neq\rho_0$; and
$g(\rho,\rho_0)< 0$ for $\rho\geq\varphi(\rho_0),\,\rho\neq\rho_0$.\\ \\
(ii) Let $\rho_0\geq\rho^*$. Then
$g(\rho,\rho_0)<0$ for $\rho\geq\rho^*,\,\rho\neq\rho_0$; and
$g(\rho,\rho_0)> 0$ for $\rho\leq\varphi(\rho_0),\,\rho\neq\rho_0$.
\end{lemma}
\begin{proof}{lemma}{lemma_1}
Apply the generalized Green's formula (\cite{cf}) to the divergence-measure fields $(h(\rho),g(\rho))$ on $(a,b)\times(0,1)$, and $(\rho,f(\rho))$ on either side of $\Gamma$.
For the latter, use \eqref{conservation_law} and note that $\dot{y}_t=0$ a.e. on $y^{-1}(0)\cup y^{-1}(1)$. The result then follows 
from simple computations using  \eqref{def_entropy_flux}, \eqref{symmetry_f},
\eqref{def_k} and \eqref{speed_rankine}. We eventually use \eqref{trace_prod} to decompose $\mu^-((a,b)\times(0,1))$ arising from the first Green's formula into the second and third term on the r.h.s. of \eqref{difference_gamma}.
\end{proof}
\mbox{}\\
\begin{proof}{lemma}{lemma_11}
The proof is similar to that of Lemma \ref{lemma_1}, but here the second Green's formula must be applied to $\psi(t,x)[\rho(t,x),f(\rho(t,x))]$ on $(a,b)\times(0,1)$.
\end{proof}
\mbox{}\\
\begin{proof}{lemma}{inequalities}
\mbox{}\\ \\
\emph{Proof of (o).}
The claim follows from \eqref{localprod_11} and strict convexity (resp. concavity) of $h$
(resp. $f$).\\ \\
\emph{Proof of (i).}
(b) Follows from strict concavity of $f$ and strict convexity of $h$. For (a),
setting $\rho=\varphi(r)$ in \eqref{localprod_2} and
and using \eqref{symmetry}, we obtain, with $\xi=v(\rho^-,\rho^+)$,
\begin{eqnarray*}
\sigma(\rho^\pm,\xi)+\pi(\rho^-,\rho^+) & = &
\int_{[\varphi(\rho_r),\varphi(\rho_l)]\backslash[\rho^-,\rho^+]}h''(r)[f(r)-f(\rho^\pm))-\xi(r-\rho^\pm))]dr\\
&-&\dsp\int_{[\rho^-,\rho^+]\backslash[\varphi(\rho_r),\varphi(\rho_l)]}h''(r)[f(r)-f(\rho^\pm))-\xi(r-\rho^\pm))]dr
\end{eqnarray*}
$\rho\mapsto f(\rho)-\xi\rho$ is a strictly concave function that
vanishes for $\rho=\rho^\pm$, because $\xi=v(\rho^-,\rho^+)$. Hence
the integrand is negative for
$\rho\not\in[\rho^-,\rho^+]$, and positive for
$\rho\in(\rho^-,\rho^+)$. This yields \eqref{inequality_1}, and
shows that equality occurs if and only if $\varphi(\rho_r)=\rho^-$
and
$\varphi(\rho_l)=\rho^+$.\\ \\
\emph{Proof of (ii).} We establish the result for $i^l$, the proof for $i^r$ being similar.  We distinguish two cases:\\ \\
First case: $\rho_l\leq\rho^*$. Then by \eqref{left_1} and
\eqref{adset_left}, we have $g(\rho,\rho_l)=i^l(\rho,\rho_l)$ if
$\rho\leq\rho^*$, that is $\rho\in\bln^-(\rho_l)$. When
$\rho\geq\varphi(\rho_l)$ and $\rho\neq\rho_l$, \eqref{left_1} and
Lemma \ref{lemma_g} imply $g(\rho,\rho_l)<0=i^l(\rho,\rho_l)$. For
$\rho^*<\rho<\varphi(\rho_l)$ we have, by \eqref{relative_flux} and
\eqref{symmetry_f},
$i^l(\rho,\rho_l)-g(\rho,\rho_l)=g(\varphi(\rho))-g(\rho)> 0$, where the inequality follows from (o) with $(\eta,q)=(h,g)$. \\ \\
Second case: $\rho_l>\rho^*$. By \eqref{left_2} and
\eqref{adset_left}, we have $g(\rho,\rho_l)=i^l(\rho,\rho_l)$ if
$\rho=\rho_l$  or $\rho\leq\varphi(\rho_l)$, that is
$\rho\in\bln^-(\rho_l)$. By \eqref{left_2} and Lemma
\ref{lemma_g}, $g(\rho,\rho_l)<0=i^l(\rho,\rho_l)$ if
$\rho\geq\rho^*$ and $\rho\neq\rho_l$. For
$\varphi(\rho_l)<\rho<\rho^*$, \eqref{left_1}, \eqref{relative_flux}
\eqref{symmetry_f} and Lemma \ref{lemma_g} imply
$i^l(\rho,\rho_l)-g(\rho,\rho_l)=g(\varphi(\rho),\rho_l)>0$.\\ \\
\emph{Proof of (iii).} We distinguish two cases:\\ \\
First case: $\rho_l\leq\rho^*$. For $\rho\leq\rho^*$, by Lemma
\ref{lemma_g},
\be \label{sandwich}-g(\rho,\rho_l)\leq 0\leq i^r(\rho,\rho_r) \ee
This is an equality for $\rho=\rho_l$ and  $f(\rho_l)\leq
f(\rho_r)$. Indeed, the latter means either
$\rho_l\leq\rho_r\leq\rho^*$, or $\rho_r\geq\rho^*$ and
$\rho_l\leq\varphi(\rho_r)$. Then, by
\eqref{right_1}--\eqref{right_2} and Lemma \ref{lemma_g},
$$
-g(\rho_l,\rho_l)=0=i^r(\rho_l,\rho_r)
$$
Conversely, assume $-g(\rho,\rho_l)=i^r(\rho,\rho_r)$. By
\eqref{sandwich}, we must have
\be \label{wemusthave} g(\rho,\rho_l)=0=i^r(\rho,\rho_r) \ee
Since $\rho\leq\rho^*$, by Lemma \ref{lemma_g}, the first equality
in \eqref{wemusthave} implies $\rho=\rho_l$. Then, again by Lemma
\ref{lemma_g}, the second equality $i^r(\rho_l,\rho_r)=0$ in
\eqref{wemusthave} holds iff $f(\rho_l)\leq f(\rho_r)$. \\ \\
Second case: $\rho_l>\rho^*$. If $\rho\leq\varphi(\rho_l)$, then
Lemma \ref{lemma_g} implies $-g(\rho,\rho_l)<0\leq
i^r(\rho,\rho_r)$. Assume now $\varphi(\rho_l)\leq\rho\leq\rho^*$.
Then
\begin{eqnarray*}
g(\rho,\rho_r)-g(\rho,\rho_l) & = &
h'(\rho_r)f(\rho_r)-g(\rho_r)-[h'(\rho_l)f(\rho_l)-f(\rho_l)]\\
& - & (h'(\rho_r)-h'(\rho_l))f(\rho) \\
& = & \int_{\rho_l}^{\rho_r}h''(r)(f(r)-f(\rho))dr
\end{eqnarray*}
Since $\varphi(\rho_l)\leq\rho\leq\rho^*$ and
$r\in[\rho_l,\rho_r]$, the integrand is negative. Hence
$-g(\rho,\rho_l)< -g(\rho,\rho_r)\leq i^r(\rho,\rho_r)$.\\ \\
The proof of (iv) is symmetric to that of (iii) and thus omitted.
\end{proof}
\subsection{$S\leq V$ and characterization of minimizers}
\label{proof_uniqueness}
In this subsection, we establish $S\leq V$ in \eqref{variational}, and prove that every $\rho(.,.)$ that achieves equality in \eqref{quasi_potential} and \eqref{variational} necessarily satisfies conditions of Theorem \ref{theorem_2}. In the next subsection we shall prove actual existence of such $\rho(.,.)$. 
Before going into details, we explain the main idea. Let $\rho(.,.)\in\emsol((-\infty,0),\times(0,1))$ with $\lim_{t\to-\infty}\rho(t,.)=\rho_s(.)\in\mathcal S$,
$\rho(0,.)=\rho(.)$, and 
$\{y^*\}=\mathcal Y[\rho_s(.)]$, that is, $y^*=1$ if $\rho_s(.)=\rho_l$, $y^*=0$ if $\rho_s(.)=\rho_r$, $y^*=y$ in the case \eqref{arbitrary_shock}.
Let $(y_t,\,t\leq 0)$ be a $[0,1]$-valued a Lipschitz path such that $\lim_{t\to-\infty}y_t=y^*$. Then
$$ S[\rho(.)]= S[\rho(.)]-S[\rho_s(.)]\leq S[\rho(0,.),y_0]-S[\rho_s(.),y^*]$$
If we find a path $y_.$ such that $d^\circ(t)\leq 0$ and $\partial d(t)\leq 0$ in \eqref{difference_gamma} (with $a\to-\infty$ and $b=0$), we obtain the upper bound in \eqref{variational}.
It follows from \eqref{difference_gamma} and Lemma \ref{inequalities} that $d^\circ(t)\leq 0$ and $\partial d(t)\leq 0$ if the path $y_.$ enjoys the following properties:\\ \\
(i) Whenever $\rho(t,y_t^\pm)$ are different, $\dot{y}_t$ is the Rankine-Hugoniot speed \eqref{speed_rankine}.\\ \\
(ii) Whenever $\rho(t,y_t^\pm)$ are equal, $\dot{y}_t$ is the characteristic speed $f'(\rho(t,y_t\pm))$.\\ \\
(iii) $\rho(t,0^+)\geq\rho^*$ a.e. on $\{y_t=0\}$, $\rho(t,0^+)\leq\rho^*$ a.e. on $\{y_t=1\}$.\\ \\
(iv) $y_.$ sees only shocks, i.e. $\rho(t,y_t-)\leq\rho(t,y_t+)$ a.e.\\ \\
Note that (i) is automatically satisfied by any path $y_.$, and
(iii) means that $y_.$ stays at a boundary iff. it is driven against it by the characteristic speed.
Equality in \eqref{difference_gamma} implies that the measure term vanishes and $d^\circ(t)=\partial d(t)=0$. The former gives entropicity of $\tilde{\rho}(.,.)$ outside $\Gamma$. The latter
yields conditions on $\rho(t,y_t^\pm)$, $\rho(t,0^+)$, $\rho(t,1^-)$. Using Lemma \ref{inequalities} one sees that these conditions are exactly those of Theorem \ref{theorem_2}.\\ \\ 
%
A path satisfying (i)--(iv) is a (maximal or minimal) generalized forward characteristic (g.f.c.) of $\rho(.,.)$ in a sense slightly wider than \cite{daf} (where no boundary and only entropy solutions are considered). Its existence can be established with the arguments of \cite[Theorem 3.2]{daf} if we know that limits $\rho(t,x\pm)$ exist in usual sense for a.e. $t$.
Unfortunately, in $\emsol((-\infty,0]\times(0,1))$, we only have these limits in the weaker sense \eqref{lonelims}, and we are not able
to obtain existence of the g.f.c. in this case. Therefore the actual proof is a local, somewhat more technical version of the above idea, that uses the mollified version
\eqref{difference_psi} of \eqref{difference_gamma}.
However, for paths that achieve equality in \eqref{quasi_potential} and \eqref{variational}, we can a priori establish existence of genuine limits in certain domains, and \eqref{difference_gamma}
will be used to reach some of the conclusions.
\\ \\ 
The first lemma is a simple consequence of the fact that $\emsol((-\infty,0)\times(0,1))$ is contained in $C^0((-\infty,0],L^1((0,1))$. Then, all but the existence part of Theorem \ref{theorem_2}
is contained in Propositions \ref{variational_1}--\ref{lemma_uniqueness_2} below. 
\begin{lemma}\label{closed_minima}
Let $\rho(.,.)\in\emsol((-\infty,0)\times(0,1))$. Then the following set is closed:
\be\label{minimizer_set}\mathcal M[\rho(.,.)]:=\{(t,y)\in(-\infty,0]\times[0,1]:\,y\in\mathcal Y[\rho(t,.)]\}\ee
\end{lemma}
\begin{proposition}
\label{variational_1}
Let $-\infty< s<t<0$, $\rho(.,.)\in\emsol((-\infty,0)\times(0,1))$.
Then 
\be \label{inequality_production_2} S[\rho(t,.)]-S[\rho(s,.)]\leq
J_{(s,t)}[\rho(.,.)] \ee
In particular, if $\rho(.)\in L^{\infty,K}((0,1))$ and
${\rho}(.,.)\in{\mathcal R}[\rho(.)]$,
\be\label{upper_bound_variational} S[\rho(.)]\leq
J_{(-\infty,0)}[{\rho}(.,.)] \ee
Assume \eqref{inequality_production_2} is an equality. Then 
for a.e. $\theta\in (s,t)$, 
$\mathcal Y[\rho(\theta,.)]$ is reduced to a single point $y$ such that $(\theta,y)\in J$ if $y\in(0,1)$, and the following hold, respectively
if $y\in(0,1)$, $y\in(0,1]$, $y\in[0,1)$, $y=1$, $y=0$:
\begin{eqnarray}
\label{inside_shock} {\rho}(\theta,y-) & = &
\varphi(\rho_r)<{\rho}(\theta,y+)=\varphi(\rho_l)\\
\label{boundary_1} \rho(\theta,0^+) & \in & {\mathcal
E}^-(\rho_l)\\
\label{boundary_2}\rho(\theta,1^-)& \in & \bln^+(\rho_r)\\
\label{boundary_3}
\rho(\theta,1^-) & =& \rho_l,\quad f(\rho_l)\leq f(\rho_r)\\
\label{boundary_4}
\rho(\theta,0^+) & = & \rho_r,\quad f(\rho_l)\geq f(\rho_r)
\end{eqnarray}
Besides, 
\be\label{vanish_outside}\mu^-(\Omega)=0\ee
where $\Omega$ is the open set defined by
\be\label{def_Omega}
\Omega:=  \{(\theta,y)\in(-\infty,0)\times(0,1):\,y\not\in\mathcal Y[\rho(\theta,.)]\}\ee
%
%
%
%
%
%
%
%
%
\end{proposition}
\begin{proposition}
\label{lemma_uniqueness_1}
Assume $\rho(.,.)\in{\mathcal R}[\rho(.)]$ achieves
equality in \eqref{quasi_potential} and \eqref{variational}.  Then 
$\tilde{\rho}(.,.)\in
L^{\infty,K}((0,+\infty)\times(0,1))$ given by \eqref{reversal}
satisfies conditions (a)--(c) of
Theorem \ref{theorem_2} for some $\tilde{y}\in\tilde{\mathcal Y}[\tilde{\rho}(.)]$, where  $\tilde{\rho}(.)=\rho(1-.)$.
\end{proposition}
\begin{proposition}
\label{lemma_uniqueness_2}
Let $\tilde{\rho}\in L^{\infty,K}((0,1))$. Then for every
$\tilde{y}\in\tilde{\mathcal Y}[\tilde{\rho}(.)]$, there exists at
most one $\tilde{\rho}(.,.)\in L^{\infty,K}((0,+\infty)\times(0,1))$
satisfying conditions (a)--(c) of Theorem \ref{theorem_2}. If
$\tilde{\rho}(.,.)$ exists, we have $\tilde{\mathcal
Y}[\tilde{\rho}(t,.)]=\{\tilde{y}_{t}\}$ for every $t>0$. Besides,
${\rho}(.,.)\in L^{\infty,K}((-\infty,0)\times(0,1))$ given by
\eqref{reversal} lies in $\mathcal R_T[\rho(.)]$ for some
$T<+\infty$, for which it achieves equality in \eqref{quasi_potential_finite} and \eqref{finite_time}. 
If $\tilde{\rho}(.,.)$ exists for two distinct values of $\tilde{y}$, statement (0) of Theorem \ref{theorem_2} holds for these values.  
\end{proposition}
For the proof of Proposition \ref{lemma_uniqueness_1}, we need a local version of \cite[Theorem 3.2]{daf} to deal with generalized characteristics. We do not repeat the proof, which is a straightforward variant of the original one. Let $\rho(.,.)\in L^{\infty,K}((-\infty,0)\times\R)$. 
Recall that $y_.$ is called a Filippov solution (\cite{fil}) of
\be\label{diff_eq}\dot{y}_t=f'(\rho(t,y_t))\ee
if it satisfies the differential inclusion
\be\label{differential_inclusion}
\dot{y}_t\in[
{\rm ess}\liminf_{x\to y_t}f'(\rho(t,x)),{\rm ess}\limsup_{x\to y_t}f'(\rho(t,x))
]
\ee
A forward (resp. backward) solution issued from $(t_0,y_0)\in(-\infty,0)\times\R$ is a solution $y_.$ defined on $[t_0,t_1]$ (resp. $[t_1,t_0]$), where $t_0\leq t_1\leq 0$ (resp. $t_1\leq t_0$), such that $y_{t_0}=y_0$. Any  forward (resp. backward) solution issued from $(t_0,y_0)$ lies in a fan between a unique lower forward (resp. backward) solution and a unique upper forward (resp. backward) solution. Because $f'$ is bounded there is no blowup in \eqref{diff_eq}, hence the upper and lower forward (resp. backward) solution are defined on $[t_0,0]$ (resp. $(-\infty,t_0]$).
\begin{lemma}
\label{lemma_daf}
Let $(t_0,y_0)\in(-\infty,0)\times\R$, and $y_.$ be the upper or lower backward (resp. forward) Filippov solution of \eqref{diff_eq} issued from $(t_0,y_0)$.
Suppose $\mathcal T$ is a subset of $(-\infty,0)$ such that $\rho(.,.)$ is an entropic (resp. anti-entropic) solution of \eqref{conservation_law}, i.e. $\mu^+[\rho(.,.)]=0$
(resp. $\mu^-[\rho(.,.)]=0$)
in a neighborhood of $\{(t,y_t):\,t\in\mathcal T\}$.
Then $\rho(t,x^+)=\rho(t,x^-)$ for a.e. $t\in\mathcal T$.
\end{lemma}
The underlying heuristics goes as follows. Suppose e.g. $\rho(.,.)$ is anti-entropic and $y_.$ is the upper forward solution. The only possible discontinuities of $\rho(.,.)$ along $y_.$ are antishocks, i.e. (since $f$ is strictly concave) $\rho^-=\rho(t,y_t^-)>\rho^+=\rho(t,y_t^+)$. Then the Rankine-Hugoniot local speed  $\dot{y}_t=v(\rho^-,\rho^+)<f'(\rho^+)$. Thus by perturbing $y_.$ to the right one could obtain a new bigger solution of \eqref{diff_eq} that contradicts maximality of $y_.$.\\ \\
\textbf{Remark.} Limits in Lemma \ref{lemma_daf} exist in usual sense for all $t$ because either ${\rho}(t,x)$ or its space-time reversal is locally an entropy solution and $f$ is uniformly concave, thus has locally bounded space variation. \\
%
%
%
%
%
%
%
%
\begin{proof}{proposition}{variational_1}
\eqref{upper_bound_variational} follows from
\eqref{inequality_production_2} because $S[.]$ vanishes on ${\mathcal S}$ and  is
continuous w.r.t. $L^1$ norm.
For $\theta\in(s,t)$, let
\be\label{difference_prod}D(\theta)=S[\rho(t,.)]-S[\rho(\theta,.)]-J_{(\theta,t)}[\rho(.,.)]\ee
and
\be\label{derivative_prod}
\Delta(\theta):=
\liminf_{\varepsilon\to 0}
\frac{D(\theta)-D(\theta+\varepsilon)}{\varepsilon}\ee
Let $y\in\mathcal Y[\rho(\theta,.)]$. Then
$$
D(\theta)-D(\theta+\varepsilon)\leq S[\rho(\theta+\varepsilon,.),y_{\theta+\varepsilon}]-S[\rho(\theta,.),y_\theta]-J_{(\theta,\theta+\varepsilon)}[\rho(.,.)]
$$
for any Lipschitz path $y_.$ on $[\theta,\theta+\varepsilon]$ such that $y_\theta=y$. In particular, taking the constant path $y_.\equiv y$ in \eqref{difference_gamma}, we see that $\varepsilon^{-1}[D(\theta)-D(\theta+\varepsilon)]$ is uniformly bounded above.
Hence, Proposition \ref{variational_1} follows if we show
that $\Delta(\theta)\leq 0$ a.e. on $(s,t)$, and that $\Delta(\theta)=0$ implies (outside a null subset)  $\mathcal Y[\rho(\theta,.)]=\{y\}$, where
$\rho(\theta,.)$ and $y$ satisfy the conditions of Proposition \ref{variational_1}.\\ \\
In the sequel we choose the values $\rho_0,\rho_1$ of $\rho(.,.)$ outside $x\in(0,1)$ such that $\rho_0<\rho^*<\rho_1$.
Let $J_s$, resp. $J_a$,  denote the shock and antishock set of $\rho(.,.)$ in $(-\infty,0)\times(0,1)$, defined as the set of $(t,x)\in J\cap((-\infty,0)\times(0,1))$ 
such that $\rho(t,x-)<\rho(t,x+)$ (resp. $>$).
Let $\mathcal T$ denote the total subset of points $\theta\in(s,t)$ such that
$\theta$ is a Lebesgue point of $\rho(\theta,0^+)$ and $\rho(\theta,1^-)$, and
%
%
$\{\theta\}\times(0,1)\subset J_s\cup J_a\cup V$.
%
In the following we consider $(\theta,y)$ with $\theta\in\mathcal T$ and $y\in \mathcal Y[\rho(\theta,.)]$. We set $\rho^\pm=\rho(\theta,y\pm)$ whenever $(\theta,y)\in J_s\cup J_a$.
Let $\xi_0\in\R$, whose value will be chosen below.
For $\delta> 0$, we define  functions $\psi_{\pm\delta}:[\theta,\theta+\varepsilon]\times[0,1]\to [0,1]$ by
$\psi_{\pm\delta}(\theta,.)=1_{[0,y]}$ and, for $u>\theta$,
\begin{eqnarray}\label{def_moll_1}
\psi_{\pm\delta}(u,x) & := & \Psi_{\pm\delta}((x-y)/(u-\theta))\\
\label{def_moll_2}
\Psi_{\delta}(\xi) & := & 1_{(-\infty,\xi_0)}(\xi)+\frac{\xi_0-\xi}{\delta}1_{(\xi_0,\xi_0+\delta)}(\xi)\\
\label{def_moll_21} \Psi_{-\delta}(\xi) & := &
1_{(-\infty,\xi_0-\delta)}(\xi)+\frac{\xi_0-\delta-\xi}{\delta}1_{(\xi_0-\delta,\xi_0)}(\xi)
\end{eqnarray}
so that
\begin{eqnarray}\label{derivative_psi}
-\partial_x\psi_{\pm\delta}(u,x) & = & \frac{1}{u-\theta}\frac{1}{\delta}1_{I(\xi_0,\xi_0+\delta)}\left(
\frac{x-y}{u-\theta}
\right)\\
\label{speed_psi}
-\frac{\partial_u\psi_{\pm\delta}(u,x)}{\partial_x\psi_{\pm\delta}(u,x)} & = & \frac{x-y}{u-\theta}1_{I(\xi_0,\xi_0+\delta)}\left(
\frac{x-y}{u-\theta}
\right)
\end{eqnarray}
with $I(a,b):=(\min(a,b),\max(a,b))$.
Since $y\in\mathcal Y[\rho(\theta,.)]$ we have, with $\psi=\psi_{\pm\delta}$,
\be\label{use_minimizer}
D(\theta)-D(\theta+\varepsilon)\leq S_{\psi(\theta+\varepsilon,.)}[\rho(\theta+\varepsilon,.)]-S_{\psi(\theta,.)}[\rho(\theta,.)]-J_{(\theta,\theta+\varepsilon)}[\rho(.,.)]
\ee
Define
\begin{eqnarray}
d_\psi^1(u) &: = & \int_0^1\frac{1}{(u-\theta)\delta}1_{(\xi_0,\xi_0+\delta)}\left(
\frac{x-y}{u-\theta}
\right)\sigma[\rho(t,x),\xi_0]dx\label{def_dpsi1}\\
d_\psi^2(u) & := & \int_0^1\frac{1}{(u-\theta)\delta}1_{(\xi_0,\xi_0+\delta)}\left(
\frac{x-y}{u-\theta}
\right)\sigma[\rho(t,x),f'(\rho(t,x))]dx\label{def_dpsi2}
%
%
%
%
\end{eqnarray}
%
%
%
Since $\sigma(\rho,\xi)$ is uniformly Lipschitz we have, for some
constant $C=C(f,h)>0$,
\begin{eqnarray}
\abs{
\varepsilon^{-1}\int_\theta^{\theta+\varepsilon}d_\psi^\circ(u)du-\varepsilon^{-1}\int_\theta^{\theta+\varepsilon} d_\psi^1(u)du
}& \leq  & C\delta\label{replacement_1}\\
\varepsilon^{-1}\int_\theta^{\theta+\alpha\varepsilon}d_\psi^1(u)du & \leq & C\alpha\label{replacement_2}
\end{eqnarray}
with $d_\psi^\circ(.)$ defined in \eqref{dpsicirc}. In the following  we apply either \eqref{difference_psi} with $\psi=\psi_{\pm\delta}$ and $\xi_0$ as specified, or \eqref{difference_gamma} with $y_.$ as specified. The constant $W:=1+||f'||_\infty$ will appear in several places.\\ \\
{\em First case.} Assume $(\theta,y)\in V$. Let
$\xi_0=f'(\bar{\rho})$, where $\bar{\rho}$ is the mean value of $\rho(.,.)$ in the ball $B_{\theta,y,\varepsilon}$ of radius $W\varepsilon$ centered at $(\theta,y)$. Then
$$
\ba{ll}
& \dsp\abs{
\varepsilon^{-1}\int_{\theta+\alpha\varepsilon}^{\theta+\varepsilon}
d_\psi^1(u)du-\varepsilon^{-1}\int_{\theta+\alpha\varepsilon}^{\theta+\varepsilon}d_\psi^2(u)du
}\\
\leq  & \dsp C\alpha^{-1}\delta^{-1}\varepsilon^{-2}\int_\theta^{\theta+\varepsilon}\int_{y+\xi_0\varepsilon}^ {y+(\xi_0+\delta)\varepsilon}\abs{\rho(u,x)-\bar{\rho}}dxdu
\ea
$$
For small $\delta>0$,
the domain of the above space-time integral is contained in $B_{\theta,y,\varepsilon}$. Thus, by VMO property, this integral vanishes as $\varepsilon\to 0$.
Letting $\varepsilon\to 0$ and then $(\delta,\alpha)\to(0,0)$ in \eqref{difference_psi}, using (ii) of Lemma \ref{inequalities} for the boundary terms,
 and \eqref{replacement_1}--\eqref{replacement_2}, we find that $\Delta(\theta)$ is bounded above by \eqref{uniformbound_sigma}.\\ \\
{\em Second case.} Assume $(\theta,y)\in J_a\cup J_s$.
Let $\xi_0=v(\rho^-,\rho^+)$ if $(\theta,y)\in J_s$, $\xi_0=f'(\rho^+)$ if $(\theta,y)\in J_a$.
Since $\sigma(\rho,\xi)$ is uniformly Lipschitz,
\be\label{replacement_3}
\ba{ll}
& \dsp\abs{
\varepsilon^{-1}\int_{\theta+\alpha\varepsilon}^{\theta+\varepsilon}d_\psi^1(u)du-(1-\alpha)\sigma(\rho^+,\xi_0)
}\\
\leq  & \dsp C\alpha^{-1}\delta^{-1}\varepsilon^{-2}\int_\theta^{\theta+\varepsilon}\int_{y+\xi_0\varepsilon}^{y+(\xi_0+\delta)\varepsilon}\abs{\rho(u,x)-\rho^+}dxdu
\ea
\ee
%
%
%
%
%
Since $n_t=-v(\rho^-,\rho^+)n_x$, 
the domain of the above integral lies in the half-plane containing $(\theta,y)$ and $n$.
This holds also in the case $(\theta,y)\in J_a$, because $\rho^->\rho^+$ implies $\xi_0>v(\rho^-,\rho^+)$ by strict concavity of $f$.
Thus, by approximate limit property, the  r.h.s. of \eqref{replacement_3} vanishes as $\varepsilon\to 0$. Thus in the case $(\theta,y)\in J_a$ we find as previously $\Delta(\theta)$ bounded above by the l.h.s. of \eqref{uniformbound_sigma}. 
%
%
Let $\mathcal T'$ denote the set of $\theta\in\mathcal T$  such that there exists $z\in\mathcal Y[\rho(\theta,.)]$ for which $(\theta,z)\in V\cup J_a$.
In view of the above, to establish $\Delta(\theta)\leq 0$ in remaining cases, we may assume without loss of generality that $\theta\in\mathcal T\backslash\mathcal T'$.
Once the inequality is established for such $\theta$, it will also follow that $\mathcal T'$ is negligible.\\ \\
We now consider $(\theta,y)\in J_s$.
We bound the measure term in \eqref{difference_psi} as follows. Let 
$\mathcal Z$
be an arbitrary finite subset of $\mathcal Y_\theta:=\mathcal Y[\rho(\theta,.)]\cap(0,1)$. 
For each $z\in\mathcal Z$, let $n(\theta,z)$ denote the local normal at $J$ and $v(\theta,z)=-n_t(\theta,z)/n_x(\theta,z)$ the local 
velocity. 
Let  $H_1$ and $H_2$  denote nonnegative  continuous functions  on $[0,+\infty)$ supported and not identically $0$ on $[0,1]$, with $H_1(0)=1$ and
$H_2(0)=0$. Set
\begin{eqnarray}\chi_\varepsilon(t,x,v) & = & H_2[\varepsilon^{-1}t^+]H_1[\varepsilon^{-1}|x-vt|]\\
\chi_\varepsilon(t,x,\mathcal Z) & = & \sum_{z\in\mathcal Z}\chi_\varepsilon(t-\theta,x-z,v(\theta,z))
\end{eqnarray}
for $(t,x)\in\R^2$. Then, for $\varepsilon>0$ small enough, 
\be\label{bound_measure_term}
-\mu^-[(\theta,\theta+\varepsilon)\times(0,1)]   \leq   -\int \chi_\varepsilon(t,x,\mathcal Z)d\mu(t,x)\\
\ee
%
%
%
%
We apply Green's formula  to the r.h.s. of \eqref{bound_measure_term},  and let $\varepsilon\to 0$ and $\mathcal Z$ grow to $\mathcal Y_\theta$. Using approximate limit property at $(\theta,z)$ and (o) of Lemma \ref{inequalities}, one obtains
$$\limsup_{\varepsilon\to 0}\varepsilon^{-1}(-\mu^-)[(\theta,\theta+\varepsilon)\times(0,1)]\leq\sum_{z\in\mathcal Y_\theta}\pi[\rho(\theta,z-),\rho(\theta,z+)]$$
Letting next $(\delta,\alpha)\to(0,0)$, we have
\begin{eqnarray}\label{green_limit}
\Delta(\theta)& \leq & \partial d(\theta)+\sigma(\rho^+,v(\rho^-,\rho^+))+\pi(\rho^-,\rho^+)\\
& + & \sum_{z\in\mathcal Y_\theta\backslash\{y\}}\pi[\rho(\theta,z-),\rho(\theta,z+)]\nonumber
\end{eqnarray}
By (o)--(ii) of Lemma \ref{inequalities}, we have $\Delta(\theta)\leq 0$, and $\Delta(\theta)=0$ implies
\eqref{inside_shock}, \eqref{boundary_1}--\eqref{boundary_2},  and $\mathcal Y[\rho(\theta,.)]\backslash\{y\}=\emptyset$. 
Now, instead of \eqref{bound_measure_term}, we write 
\begin{eqnarray}
-\mu^-[(\theta,\theta+\varepsilon)\times(0,1)] &  \leq &  -\int \chi_\varepsilon(t-\theta,x-y)d\mu(t,x)\label{outside}\\
& - & \mu^-[(\theta,\theta+\varepsilon)\times((0,1)\backslash(y-W\varepsilon,y+W\varepsilon))]\nonumber
%
%
\end{eqnarray}
%
%
%
%
%
It follows that
\begin{eqnarray*}
\Delta(\theta)& \leq & \partial d(\theta)+\sigma(\rho^+,v(\rho^-,\rho^+))+\pi(\rho^-,\rho^+)\\
& - & \limsup_{\varepsilon\to 0}\varepsilon^{-1}\mu^-[(\theta,\theta+\varepsilon)\times((0,1)\backslash(y-W\varepsilon,y+W\varepsilon))]
\end{eqnarray*}
Thus $\Delta(\theta)=0$ implies the following, which will be used in the last part of the proof:
\be\label{zerolimits}\lim_{\varepsilon\to 0}\varepsilon^{-1}\mu^-[(\theta,\theta+\varepsilon)\times((0,1)\backslash(y-W\varepsilon,y+W\varepsilon))]=0
\ee
{\em Third case.}  $y=0\in\mathcal Y[\rho(\theta,.)]$,
Suppose first $(\theta,0)\not\in J$. By VMO property,
\be\label{vmo_plus}
\lim_{\varepsilon\to 0}\varepsilon^{-2}\int_\theta^{\theta+\varepsilon}\int_0^{W\varepsilon}\abs{\rho(u,x)-\rho_0}dxdu=0
\ee
We apply \eqref{difference_psi} with $\psi=\psi_\delta$, $\xi_0=f'(\rho_0)>0$. 
%
%
Proceeding as in the first case one shows that $\Delta(\theta)$
is bounded above by \eqref{uniformbound_sigma}.\\ \\
Suppose next that $(\theta,0)\in J$. Then (see Remark \ref{remark_limits}) $(n_t=0,n_x=1)$ is the normal to $J$ at $(\theta,0)$, and
\be\label{trace_dow_boundary}
\lim_{\varepsilon\to 0}\varepsilon^{-2}\int_\theta^{\theta+\varepsilon}\int_0^\varepsilon\abs{\rho(u,x)-\rho(\theta,0^+)}dxdu=0
\ee
If $\rho(\theta,0^+)<\rho^*$, we proceed as above, with $\xi_0=f'(\rho(\theta,0^+))>0$,
to show that $\Delta(\theta)$ is bounded above by \eqref{uniformbound_sigma}.
Finally, assume $\rho(\theta,0^+)\geq\rho^*$. We then apply \eqref{difference_gamma} with  $y_.\equiv 0$, thus $d^\circ(.)\equiv 0$
and $\partial d(.)$ is given by \eqref{boundaryterm_gamma_left}. We bound the measure term as in \eqref{bound_measure_term}.
Since $\theta\in\mathcal T$, we arrive at
\be\label{conclusion_boundarymin}
\Delta(\theta)\leq\partial d(\theta)+\sum_{z\in\mathcal Y_\theta}\pi(\rho(\theta,z-),\rho(\theta,z+))
\ee
By (o), (ii) and (iv) of Lemma \ref{inequalities}, we have $\Delta(\theta)\leq 0$, and
$\Delta(\theta)=0$ implies \eqref{boundary_2}, \eqref{boundary_4} and  $\mathcal Y[\rho(\theta,.)]\cap(0,1)=\emptyset$.\\ \\
{\em Fourth case.} $y=1\in\mathcal Y[\rho(\theta,.)]$. This case is similar to the third one, here we have to use \eqref{difference_psi} with $\psi=\psi_{-\delta}$, or \eqref{difference_gamma} with $y_.\equiv 1$. We arrive at \eqref{conclusion_boundarymin}, and  $\Delta(\theta)=0$ implies \eqref{boundary_1}, \eqref{boundary_3} and $\mathcal Y[\rho(\theta,.)]\cap(0,1)=\emptyset$.\\ \\ 
%
%
We now prove that,
if \eqref{inequality_production_2} is an equality, $\mu^-$ vanishes on $\Omega$ (which is open because of Lemma \ref{closed_minima}).
It follows from Lemma \ref{closed_minima} that 
$\bar{y}_.$ and $\underline{y}_.$ defined by 
\be\label{ext_minima}\bar{y}_\theta  :=  \sup{\mathcal Y}[{\rho}(\theta,.)],\quad
\underline{y}_\theta  := \inf{\mathcal Y}[{\rho}(\theta,.)]\ee
are respectively u.s.c. and l.s.c.
By Proposition \ref{variational_1}, they are continuous at every $\theta\in\mathcal T$, where we have $\bar{y}_\theta=\underline{y}_\theta$.
\label{topofpage}
Applying \eqref{inequality_production_2} on $(u,t)$ with $s<u<t$ shows that inequality
\eqref{inequality_production_2} is  reversed on $(s,u)$ and thus an equality, hence an equality any subinterval of $(s,t)$. Let $y_.$ a Lipschitz-continuous $[0,1]$-valued path on $[a,b]\subset(s,t)$ such that
 $y_a\in\mathcal Y[\rho(a,.)]$ and $y_b\in\mathcal Y[\rho(b,.)]$. Then, by \eqref{difference_gamma},
\be\label{abs_cont}
\mu^-((a,b)\times(0,1))\leq C(b-a)
\ee
where $C$ is a constant depending only on $f$ and $h$.
Given $\eta
\geq 0$, we define a nonnegative measure $M_\eta\leq\mu^-\ll\lambda$ on $(0,+\infty)$ by
$$
M_\eta(A)=\mu^-\left(
\left\{
(t,x)\in A\times(0,1):\,t\in A,\,x\not\in[\underline{y}_t-\eta,\bar{y}_t+\eta]
\right\}
\right)
$$
for every Borel set $A\subset(0,+\infty)$.  
Let $\theta\in\mathcal T$.
Since  $\bar{y}_.$ and $\underline{y}_.$ are continuous at $\theta$,
%
\eqref{zerolimits} implies
\be\label{density_M}
\lim_{\varepsilon\to 0}\varepsilon^{-1}M_\eta((\theta,\theta+\varepsilon))=0
\ee
Together with \eqref{abs_cont} this implies $M_\eta\equiv 0$. Letting $\eta\to 0$ we obtain $M_0\equiv 0$,
i.e. $\mu^-$ vanishes on 
$\Omega':=\{(\theta,y)\in(0,+\infty)\times(0,1):\,y\not\in[\underline{y}_\theta,\bar{y}_\theta]\}$.
Since $\Omega\backslash\Omega'$ is contained in $\mathcal U\times(0,1)$, where $\mathcal U$ is the negligible set
of $t\in(-\infty,0)$ such that $\underline{y}_t<\bar{y}_t$, using \eqref{abs_cont}, we obtain $\mu^-(\Omega)=0$.
\end{proof}
\mbox{}\\
\begin{proof}{proposition}{lemma_uniqueness_1}
%
%
Assume \eqref{upper_bound_variational} is an equality. Then
applying \eqref{inequality_production_2} on $(s,0)$ shows that
$
S[\rho(s,.)] \geq J_{(-\infty,s)}[\rho(.,.)]
$
which must be an equality in view of \eqref{upper_bound_variational}. Thus 
\be\label{subequality}
S[\rho(t_2,.)]-S[\rho(t_1,.)]=J_{(t_1,t_2)}[\rho(.,.)]
\ee
for every $t_1\leq t_2\leq 0$.
%
%
%
Let $\mathcal U:=\{t\in(-\infty,0):\,\mathcal Y[\rho(t,.)]\cap(0,1)\neq\emptyset\}$ and $t_-:=\inf\mathcal U$. We set $t_-=0$ if $\mathcal U=\emptyset$.
By \eqref{boundary_3}--\eqref{boundary_4}, the following occurs for a.e. $t<t_-$: $\mathcal Y[\rho(t,.)]=\{1\}$ if $f(\rho_l)<f(\rho_r)$,
$\mathcal Y[\rho(t,.)]=\{0\}$ if $f(\rho_l)>f(\rho_r)$, $\mathcal Y[\rho(t,.)]\subset\{0,1\}$ if $f(\rho_l)=f(\rho_r)$.
By Lemma \ref{closed_minima}, in the first (resp. second, resp. third) case, for all $t\leq t_-$, $1$ (resp. $0$, resp. at least one element of $\{0,1\}$) lies in $\mathcal Y[\rho(t,.)]$.\\ \\
{\em First case.} $f(\rho_l)\neq f(\rho_r)$. 
For notational simplicity, we treat the case $f(\rho_l)<f(\rho_l)$, which implies $\rho_l<\rho^*$. Then $v:=v(\varphi(\rho_l),\varphi(\rho_r))<0$ (cf. \eqref{speed_rankine}).
Translation to the case $f(\rho_l)>f(\rho_r)$ is straightforward and left to the reader.
By \eqref{subequality} and
\eqref{boundary_1}--\eqref{boundary_3}, for a.e. $t<t_-$, we have $\rho(t,1^-)=\rho_l$,
$\rho(t,0^+)\in\bln^-(\rho_l)$  and $\mathcal Y[\rho(t,.)]=\{1\}$. 
We apply \eqref{difference_gamma} on $(-\infty,t_-]$ with $y_t\equiv 1$. By (ii), (iii) of Lemma \ref{inequalities} we have $d(t)=0$. Thus equality in \eqref{variational} implies
$\mu^-[(-\infty,t_-)\times(0,1)\backslash\Gamma_1)]$=0, where $\Gamma_1=(-\infty,0)\times\{1\}$. Hence $\tilde{\rho}(.,.)$ is an entropy solution to \eqref{conservation_law} on $(-t_-,+\infty)\times(0,1)$ and satisfies
\eqref{reverse_bln_2} (here with $\tilde{y}_\infty=0$ and $\sigma=+$). \\ \\
Now choose an arbitrary $s\in\mathcal U$ and
$y\in\mathcal Y[\rho(s,.)]\cap(0,1)$ and consider the upper forward solution $y_.$ of \eqref{diff_eq} issued from $(s,y)$. Let
$s':=\inf\{t\in(s,0):\,y_t\in\{0,1\}\}$, where the infimum is taken to be $0$ if the set is empty.
We first argue that
%
%
\be \label{neverleaves}y_t\in\mathcal Y[\rho(t,.)]\mbox{ for every }t\in(s,s')\ee
Indeed, assume there exists $s_1\in(s,s')$ such that 
$y_{s_1}\not\in\mathcal Y[\rho(s_1,.)]$. Let 
$s_0:=\sup\{t\in[s,s_1):\,y_t\in\mathcal Y[\rho(t,.)]\}$.
Then $s_0\in[s,s_1)$, $y_{s_0}\in\mathcal Y[\rho(s_0,.)]$, and for $t\in(s_0,s_1)$, $(t,y_t)$ lies in the set \eqref{def_Omega} . By \eqref{vanish_outside},  \eqref{differential_inclusion} and Lemma \ref{lemma_daf}, $\rho(t,y_t^-)=\rho(t,y_t^+)$ for a.e. $t\in(s_0,s_1)$.
We now apply \eqref{difference_gamma} to $\rho(.,.)$ and $y_.$ on  $[s_0,s_1]$. (i),(ii) of Lemma \ref{inequalities}
imply $d(t)<0$, hence
\begin{eqnarray*}
S[\rho(s_1,.)]-S[\rho(s_0,.)] & \leq & S[\rho(s_1,.),y_{s_1}]-S[\rho(s_0,.),y_{s_0}]\\
& < & J_{(s_0,s_1)}[\rho(.,.)]
\end{eqnarray*}
which contradicts \eqref{subequality}.
%
%
\eqref{inside_shock} and \eqref{speed_rankine} imply  $\dot{y}_t=v(\varphi(\rho_r),\varphi(\rho_l))=:v$ a.e. on $(s,s')$. Hence $y_t=y_s+v(t-s)$ for $t\in(s,s')$. Since $s$ was arbitrary, and (by Proposition \ref{variational_1}) $\mathcal Y[\rho(\theta,.)]$ is reduced to a single point for a.e. $\theta$, the above implies $t_->-\infty$ and existence of $t_+\in(t_-,0]$ such that $\mathcal Y[\rho(t,.)]\supset\{1+v(t-t_-)\}$ for $t\in[t_-,t_+]$, and $1+v(t_+-t_-)=0$ if $t_+<0$.
By \eqref{vanish_outside} and \eqref{abs_cont}, $\mu^-[\rho(.,.)]$ vanishes on $\{(t,x)\in(t_-,t_+)\times(0,1):\,x\neq 1+v(t-t_-)\}$. Thus $\tilde{\rho}(.,.)$ is an entropy solution to \eqref{conservation_law} on $\{(t,x)\in(-t_+,-t_-)\times(0,1):\,x\neq v(t+t_-)\}$. \eqref{inside_shock} and \eqref{boundary_1}--\eqref{boundary_2} of Proposition \ref{variational_1}
yield \eqref{reverse_bln} and \eqref{antishock_0} for $\tilde{\rho}(.,.)$ on $(-t_+,-t_-)$. \\ \\
We will now prove $t_+=0$, which will establish the proposition with $\tilde{y}=vt_-$. The corresponding  $\tilde{\theta}^{\tilde{y}}$ in Theorem \ref{theorem_2}
is $-t_-$.
%
%
%
Assume $t_+<0$, then by what precedes we have $0\in\mathcal Y[\rho(t_+,.)]$. Define
$t_-^{1}:=\inf\{t\in(t_+,0):\,\mathcal Y[\rho(t,.)]\cap(0,1)\neq\emptyset\}$. Set $t^1_-=0$ if the preceding set is empty.
By \eqref{vanish_outside}, $\mu^-[\rho(.,.)]$ vanishes on 
$D_1:=(t_+,t^1_-)\times(0,1)$. 
Proceeding as above, if $t^1_-<0$, one can show existence of $t_+^{1}\in(t_-^1,0]$ such that
$\mathcal Y[\rho(t,.)]\supset\{1+v(t-t_-^1)\}$ for $t\in(t_-^1,t_+^1)$, and $\mu^-[\rho(.,.)]$ vanishes on $D_2:=\{(t,x)\in[t^1_-,t^1_+):\,x\neq 1+v(t-t^1_-)\}$.
Let $D=D_1\cup D_2$. Hence,
$\tilde{\rho}(.,.)$ is an entropy solution to \eqref{conservation_law} on the open set $\tilde{D}:=\tilde{D}_1\cup\tilde{D}_2$, where
$\tilde{D}_i:=\Phi(D_i)$ is the space-time reversal (cf. \eqref{reversal_prod}) of $D_i$. Besides, by \eqref{boundary_1}--\eqref{boundary_4}, $\tilde{\rho}$ satisfies \eqref{reverse_bln} on $(-t^1_+,-t^1_-)$ and  \eqref{reverse_bln_2} on $(-t^1_-,-t_+)$, with $\tilde{y}_\infty=0$ and $\sigma=+$. 
%
%
%
For $\varepsilon>0$, set
\be\label{eps_ext}
\rho^\varepsilon(t,x):=\rho_0 1_{(-\infty,\varepsilon)}(x)+\rho(t,x)1_{(\varepsilon,+\infty)}(x)
\ee
where $\rho_0<\rho^*$, and define $S^\varepsilon$ as $S$ in \eqref{def_S_1}--\eqref{def_S_2}, but with $y\in[\varepsilon,1]$ and integral over $[\varepsilon,1]$.
It is straightforward to translate \eqref{difference_gamma} to $S^\varepsilon$, a dynamical functional $J^\varepsilon$ with left boundary $\varepsilon$, and a $[\varepsilon,1]$-valued path. Using \eqref{lonelims}, one can see that
$\lim_{\varepsilon\to 0}J^\varepsilon[\rho^\varepsilon(.,.)]=J[\rho(.,.)]$.
Let $(y_t^\varepsilon:\,t\in[t_+,0])$ denote the upper forward solution of 
\be\label{eps_inclusion}\dot{y}_t=f'(\rho^\varepsilon(t,y_t)) \ee
on $[t_+,0]$ issued from $(t_+,\varepsilon)$. 
Since $\tilde{\rho}(.,.)$ is en entropy solution on $\tilde{D}$, ${\rho}(.,.)$ has locally bounded space
variation on $D$. In particular, if $t\in(t_+,t^1_+)$ and $\varepsilon<1+\min(0,v(t-t^1_-))$,  the limit $\rho(t,\varepsilon^+)$ exists in classical sense. Thus, \eqref{eps_ext} and the Filippov sense of \eqref{eps_inclusion} imply 
\be\label{boundary_filippov}
y_t\geq\varepsilon,\quad
\rho^\varepsilon(t,\varepsilon^+)\geq\rho^*\mbox{ a.e. on }
\{t:\,y_t= \varepsilon\}\ee
Note that, with the upper forward solution of $y_.$ of \eqref{diff_eq} issued from $(t_+,0)$, we could not deduce the analogous property $\rho(t,0^+)\geq\rho^*$ on $\{y_t=0\}$,
because the latter limit is only known in the sense \eqref{lonelims}. This is why we introduced th $\varepsilon$-boundary.
Set $T^\varepsilon=\inf\{t\geq t_+:\,y_t^\varepsilon\in B\}$, where 
$B:=[t_+,t^1_-)\times\{1\}\cup\{(t,1+v(t-t^1_-)):\,t\in[t^1_-,t^1_+]\}$.
Since $y^\varepsilon$ has a Lipschitz constant $||f'||_\infty$ independent of $\varepsilon$, $T:=\liminf\{T^\varepsilon:\,\varepsilon>0\}>0$.
Let $s\in(t_+,T)$, $\mathcal T^\varepsilon:=\{t\in[t_+,s]:\,1>y_t^\varepsilon>\varepsilon\}$,
and $\mathcal T_0^\varepsilon:=\{t\in[t_+,s]:\,y_t^\varepsilon=\varepsilon\}$.  By Lemma \ref{lemma_daf}, 
$\rho^\varepsilon(t,y^\varepsilon_{t^-})=\rho^\varepsilon(t,y^\varepsilon_{t^+})$ a.e. on $\mathcal T^\varepsilon$.
Applying \eqref{difference_gamma} to $y^\varepsilon_.$, using \eqref{boundary_filippov}, \eqref{boundary_2} and (i), (ii), (iv) of Lemma \ref{inequalities}, we obtain
$$S^\varepsilon[\rho^\varepsilon(s,.)]-S^\varepsilon[\rho^\varepsilon(t_+,.),\varepsilon]\leq J^\varepsilon_{[t_+,T']}[\rho^\varepsilon(.,.)]
$$
%
%
%
%
%
Since $(y^\varepsilon_.:,\varepsilon>0)$ has uniform Lipschitz constant it is relatively compact
in $C^0([t_+,T])$. Let $y_.$ be a limit point, hence $y_{t_+}=0$. Since $0\in\mathcal Y[\rho(t_+,0)]$, letting $\varepsilon\to 0$ yields 
$$S[\rho(s,.),y_s]-S[\rho(t_+,.)]\leq J_{[t_+,s]}[\rho(.,.)]$$
for all $s\in[t_+,T]$.
This and \eqref{subequality} imply
$y_s\in\mathcal Y[\rho(s,.)]$. Since $\mathcal Y[\rho(s,.)]$ is a singleton for a.e. $s$ and $B$ is contained in the set \eqref{minimizer_set}, we obtain a contradiction with the continuity of $y_.$.
\\ \\
{\em Second case.} $f(\rho_l)=f(\rho_r)$. By \eqref{vanish_outside}, $\mu^-[\rho(.,.)]$ vanishes on $(-\infty,t_-)\times(0,1)$. 
For $y\in\{0,1\}$, Let $t_y:=\inf\{t\in(-\infty,t_-):\,y\in\mathcal Y[\rho(t,.)]\}$. Set $t_y=t_-$ if the corresponding set is empty. For at least one $y\in\{0,1\}$ we have
$t_y=-\infty$. We assume e.g. $y=0$ and leave the similar case $y=1$ to the reader. Let $\tau<t_-$ such that $0\in\mathcal Y[\rho(\tau,.)]$. We proceed as above, now defining $\rho^\varepsilon$, $S^\varepsilon$ and $J^\varepsilon$ 
with boundaries both at $\varepsilon$ and $1-\varepsilon$. \eqref{eps_ext} becomes 
\be\label{epss_ext}
\rho^\varepsilon(t,x):=\rho_0 1_{(-\infty,\varepsilon)}(x)+\rho(t,x)1_{(\varepsilon,1-\varepsilon)}(x)+\rho_1 1_{(1,+\infty)}(x)
\ee
with $\rho_0<\rho^*<\rho_1$.
Let $y_.^\varepsilon$ be the upper forward solution of \eqref{eps_inclusion}
on $[\tau,t_-)$ issued from $(\tau,\varepsilon)$. Now \eqref{boundary_filippov} becomes 
\be\label{now_becomes}y_t^\varepsilon\in[\varepsilon,1-\varepsilon],\,\rho^\varepsilon(t,\varepsilon^+)\geq\rho^*\mbox{ on }\{y_t^\varepsilon=\varepsilon\},\, \rho^\varepsilon(t,(1-\varepsilon)^-)\leq\rho^*\mbox{ on }\{y_t^\varepsilon=1-\varepsilon\}
\ee
Let $s\in[\tau,t_-]$, $\mathcal T^\varepsilon:=\{t\in[\tau,s]:\,y_t^\varepsilon\in(\varepsilon,1-\varepsilon)\}$, 
$\mathcal T_0^\varepsilon:=\{t\in[\tau,s]:\,y_t^\varepsilon=\varepsilon\}$, 
$\mathcal T_1^\varepsilon:=\{t\in[\tau,s]:\,y_t^\varepsilon=1-\varepsilon\}$. 
Applying Lemma \ref{lemma_daf} on $\mathcal T^\varepsilon$ and \eqref{difference_gamma} on $[\tau,s]$, using \eqref{now_becomes} and (i)--(iv) of Lemma \ref{inequalities}, we obtain
$$
S^\varepsilon[\rho^\varepsilon(t_-.),y^\varepsilon_{t_-}]-S^\varepsilon[\rho^\varepsilon(\tau,.),\varepsilon]  \leq  J^\varepsilon[\rho^\varepsilon(.,.)]
$$
%
%
%
%
Let $y_.$ be a limit point of $y^\varepsilon_.$ in $C^0([\tau,t_-])$ as $\varepsilon\to 0$, then
$$S[\rho(s,.),y_s]-S[\rho(\tau,.)]\leq J_{[\tau,s]}[\rho(.,.)]$$
In view of \eqref{subequality}, this must be an equality with
$y_s\in\mathcal Y[\rho(s,.)]$.
Since $\mathcal Y[\rho(s,.)]$ is a singleton for a.e. $s$ and $y_.$ is continuous, we must have $y_s\equiv 0\in\mathcal Y[\rho(s,.)]$  for all $s\in[\tau,t_-]$.
Since $\tau$ is arbitrarily small, \eqref{boundary_2}--\eqref{boundary_4} imply that on $(-t_-,+\infty)$, $\tilde{\rho}(.,.)$ is an entropy solution that satisfies \eqref{reverse_bln_2}, here with $\tilde{y}_\infty=1$ and $\sigma=-$.\\ \\
If $\mathcal U=\emptyset$, i.e. $t_-=0$, the above agument establishes the proposition with  $\tilde{y}=1-y$.
Assume $\mathcal U\neq\emptyset$, $t_-<0$. Since now we have $v=0$, the same arguments as in the first case show that in fact $\mathcal U=(-\infty,0)$, $t_-=-\infty$, and there exists ${y}\in(0,1)$ such that $\mu^-[\rho(.,.)]$ vanishes on $(-\infty,0)\times[(0,1)\backslash\{y\}]$, $\rho(t,0^+)\in\bln^-(\rho_l)$ and $\rho(t,1^-)\in\bln^+(\rho_r)$ for a.e. $t\in(-\infty,0)$. Thus $\tilde{\rho}(.,.)$ is an entropy solution on $(0,+\infty)\times[(0,1)\backslash\{\tilde{y}:=1-y\}]$ that satisfies boundary conditions \eqref{reverse_bln}. This establishes the proposition with  $\tilde{y}=1-y$. In both cases $t_-=0$ and $t_-=-\infty$, the corresponding $\tilde{\theta}^{\tilde{y}}$ in Theorem \ref{theorem_2} is $-t_-$.
\end{proof}
\mbox{}\\
\begin{proof}{proposition}{lemma_uniqueness_2}
Define $\tilde{\rho}^+(.,.)$, resp. $\tilde{\rho}^-(.,.)$ on
$(0,\tilde{\theta}^{\tilde{y}})\times(-\infty,1)$, resp.
$(0,\tilde{\theta}^{\tilde{y}})\times(0,+\infty)$ by
\begin{eqnarray} \label{def_tildeplus}
\tilde{\rho}^+(t,.) & := & \varphi(\rho_r){\bf
1}_{(-\infty,\tilde{y}+vt)}+\tilde{\rho}(t,.){\bf 1}_{(\tilde{y}+vt,1)}\\
\label{def_tildemin}
\tilde{\rho}^-(t,.) & := &  \varphi(\rho_l){\bf
1}_{(\tilde{y}+vt,+\infty)}+\tilde{\rho}(t,.){\bf 1}_{(0,\tilde{y}+vt)}
\end{eqnarray}
Since $\tilde{\rho}^\pm$ has no jump across the line $x=\tilde{y}+vt$, is
constant on one side, and $\tilde{\rho}$ is an entropy
solution to \eqref{conservation_law} outside this line,
$\tilde{\rho}^\pm$ is an entropy solution on its domain. Besides we
have
\begin{eqnarray} \label{initial_tildeplus}
\tilde{\rho}^+(0,.) & = & \varphi(\rho_r){\bf 1}
_{(-\infty,\tilde{y})}+{\rho}(0,1-.){\bf 1}_{(\tilde{y},1)}\\
\label{initial_tildemin} \tilde{\rho}^-(0,.) & = &
\varphi(\rho_l){\bf 1}
_{(\tilde{y},+\infty)}+{\rho}(0,.1-){\bf 1}_{(0,\tilde{y})}\\
\label{right_bln} \tilde{\rho}^+(t,1^-) & \in & {\mathcal
E}^-(\rho_l)\\
\label{left_bln} \tilde{\rho}^-(t,0^+)& \in & \bln^+(\rho_r)
\end{eqnarray}
where \eqref{right_bln}--\eqref{left_bln} follow from
\eqref{reverse_bln}. Thus $\tilde{\rho}^\pm$ is uniquely determined
as the entropy solution of an initial-boundary problem for
\eqref{conservation_law}. This uniquely determines
$\tilde{\rho}(.,.)$ in $(0,\tilde{\theta}^{\tilde{y}})\times(0,1)$.
Next, since $\rho_b\in\bln^\sigma(\rho_b)$ for any
$\sigma\in\{+,-\}$, \eqref{reverse_bln_2} implies that on
$(\tilde{\theta}^{\tilde{y}},+\infty)\times(0,1)$, $\tilde{\rho}(.,.)$ is the unique
entropy solution to \eqref{conservation_law} with Cauchy datum
$\tilde{\rho}(\tilde{\theta}^{\tilde{y}},.)$ (previously determined) and boundary
conditions $\tilde{\rho}(t,0^+)\in\bln^+(\rho_b)$,
$\tilde{\rho}(t,1^-)\in\bln^-(\rho_b)$.\\ \\
Statement (0) of Theorem \ref{theorem_2} follows because, assuming existence of $\tilde{\rho}^i(.,.):=M^{\tilde{y}^i}[\tilde{\rho}(.)]$ for $\tilde{y}^1\neq \tilde{y} ^2$, 
: if both $\tilde{y}^i\in(0,1)$, both $\tilde{\rho}^i$ have a single antishock along two two different curves. If exactly one $\tilde{y}^i\in(0,1)$, or if they both lie in $\{0,1\}$ and $f(\rho_l)\neq f(\rho_r)$, one $\tilde{\rho}^i$ has an antishock and the other does not. If both $\tilde{y}^i\in\{0,1\}$ and $f(\rho_l)=f(\rho_r)$,
asumme e.g. $\tilde{y}^1=0$, $\tilde{y}^2=1$. Then $\tilde{\rho}^1(t,0^+)=\rho_l$ a.e. while $\tilde{\rho}^2(t,0^+)\in\bln^+(\rho_r)$ a.e. Since $\rho_l<\rho^*<\rho_r$, 
we have (see \eqref{adset_left_2}--\eqref{adset_right_2}) $\rho_l\not\in\bln^+(\rho_r)$, which implies $\tilde{\rho}^1\neq\tilde{\rho}^2$.
\\ \\
Let us now prove that ${\rho}(.,.)$  given by \eqref{reversal} lies
in $\mathcal R_T[\rho(.)]$ for some $T>0$. If $\tilde{\theta}^{\tilde{y}}<+\infty$,
this follows from Theorem \ref{th_stat} and $\rho_b\neq\rho^*$. Assume now $f(\rho_l)=f(\rho_r)$,
$\tilde{\theta}^{\tilde{y}}=+\infty$, $y\in(0,1)$. Since
$\varphi(\rho_l)=\rho_r\in\bln^-(\rho_r)$ and
$\varphi(\rho_r)=\rho_l\in\bln^+(\rho_l)$, by
\eqref{reverse_bln} and \eqref{antishock_0}, the restriction of
$\tilde{\rho}$ to $(0,+\infty)\times(0,\tilde{y})$ is an entropy solution
to \eqref{conservation_law}  with BLN boundary conditions
$\tilde{\rho}(t,0^+)\in\bln^+(\rho_r)$,
$\tilde{\rho}(t,\tilde{y}^-)\in\bln^-(\rho_r)$. Similarly,
the restriction of $\tilde{\rho}$ to $(0,+\infty)\times(\tilde{y},1)$ is
an entropy solution to \eqref{conservation_law}  with boundary
conditions $\tilde{\rho}(t,\tilde{y}^+)\in\bln^+(\rho_l)$,
$\tilde{\rho}(t,1^-)\in\bln^-(\rho_l)$. The result follows
from Theorem \ref{th_stat}, $\rho_l\neq\rho^*$, $\rho_r\neq\rho^*$.
In all cases, we have
$\lim_{t\to-\infty}{\rho}(t,.)=\rho^{1-\tilde{y}_\infty}(.)$, with
$\rho^y(.)$ given by
\eqref{arbitrary_shock}.\\ \\
We next prove that $\rho(.,.)$ given by \eqref{reversal} achieves equality in
\eqref{finite_time}. Set
$z_t=1-\tilde{y}_{-t}$.
Let $-\infty<T<0$. 
%
%
By \eqref{difference_gamma} applied to $z_.$, Lemma \ref{inequalities} and conditions (2), (a)--(c) of Theorem \ref{theorem_2},
\be \label{integrate_equality}
S[\rho(0,.),z_0]-S[\rho(T,.),z_t]=J_{(t,0)}[\rho(.,.)] \ee
Since
%
%
$\rho(0,.)=\rho(.)$ and $z_0=y\in{\mathcal Y}[\rho(.)]$,
\eqref{integrate_equality} and \eqref{inequality_production_2} on $(T,0)$ yield
$S[\rho(T,.),z_T]=S[\rho(T,.)]$, i.e.
$z_t\in\mathcal Y[\rho(T,.)]$, or equivalently $\tilde{y}_{-T}\in\tilde{\mathcal Y}[\tilde{\rho}(T,.)]$. For $T$ small enough such that $\rho(.,.)\in\mathcal R_T[\rho(.)]$, \eqref{integrate_equality} yields
equality in \eqref{quasi_potential_finite} and \eqref{finite_time}.\\ \\
We eventually prove that $\tilde{\mathcal Y}[\tilde{\rho}(s,.)]=\{\tilde{y}_s\}$ for all $s>0$. 
Assume the contrary and let $\tilde{y}^1\neq\tilde{y}^2$ for some $s>0$. Up to this point we have already established uniqueness of $M^{\tilde{y}}[\tilde{\rho}(.)]$ and statement (0) of Theorem \ref{theorem_2}.
By anticipation we use the existence of $M^{\tilde{y}}[\tilde{\rho}(.)]$ established independently in the next subsection. We set $\tilde{\rho}^i(t,.)=M^{\tilde{y}^i}[\tilde{\rho}(s,.)](t-s,.)$ for $t>s$ and $\tilde{\rho}^i(t,.)=\tilde{\rho}(t,.)$ for $0\leq t\leq s$. Statement (0) implies $\tilde{\rho}^1\neq\tilde{\rho}^2$. But the uniqueness statement of Theorem \ref{theorem_2} yields $\tilde{\rho}^i(.,.)=\tilde{\rho}(.,.)$ for each $i\in\{1,2\}$. 
\end{proof}
\subsection{$S\geq V$ and existence of a minimizer}
\label{proof_construction}
We construct, for each $\tilde{y}\in\tilde{\mathcal Y}[\tilde{\rho}(.)]$, a
$\tilde{\rho}(.,.)$ satisfying conditions (2), (a)--(c) of Theorem \ref{theorem_2}.\\ \\
\textbf{Construction on the time interval $(0,\tilde{\theta}^{\tilde{y}})$}.
We proceed inductively by constructing $\tilde{\rho}(.,.)$ on intervals $[t_n,t_{n+1}]$. Let
$t_0=0$, $z_n=\tilde{y}_{t_n}=\tilde{y}+v t_n$,
%
%
\be\label{seq_t} t_{n+1}-t_n =
\min\left(\frac{z_n}{V-v},\frac{1-z_n}{V+v}\right) \ee
%
%
%
where $V=\sup\abs{f'}>\abs{v}$ by strict concavity of $f$.
Then $\lim_{n\to\infty}t_n=\tilde{\theta}^{\tilde{y}}$.
%
%
%
%
Set $\tilde{\rho}(0,.)=\tilde{\rho}(.)$ and define $\tilde{\rho}(.,.)$
on $(t_n,t_{n+1})\times(0,1)$ by
\be \label{def_tilde}
\tilde{\rho}(t,.)=\tilde{\rho}^+_n(t-t_n,.){\bf
1}_{(z_n+v(t-t_n),1)}+\tilde{\rho}^-_n(t-t_n,.){\bf
1}_{(0,z_n+v(t-t_n))}\ee
where $\tilde{\rho}^+_n$ and $\tilde{\rho}^-_n$ are the entropy
solutions to \eqref{conservation_law}, respectively on $(-\infty,1)$
and $(0,+\infty)$, with initial/boundary conditions
\begin{eqnarray} \label{initial_tildeplus_2}
\tilde{\rho}^+_n(0,.) & = & \varphi(\rho_r){\bf
1}_{(-\infty,z_n)}+\tilde{\rho}(t_n,.){\bf 1}_{(z_n,1)}\\
\label{initial_tildemin_2}
\tilde{\rho}^-_n(0,.) & = & \varphi(\rho_l){\bf
1}_{(z_n,+\infty)}+\tilde{\rho}(t_n,.){\bf 1}_{(0,z_n)}\\
\label{right_bln_n} \tilde{\rho}_n^+(t,1^-) & \in & {\mathcal
E}^-(\rho_l)\\
\label{left_bln_n} \tilde{\rho}_n^-(t,0^+)& \in & {\mathcal
E}^+(\rho_r)
\end{eqnarray}
\eqref{reverse_bln} follows from
\eqref{right_bln_n}--\eqref{left_bln_n}. 
%
%
%
We prove by induction on $n$ that \eqref{antishock_0}
holds a.e. on $(0,t_n)$ and that
\be \label{min_n}z_n\in\tilde{\mathcal Y}[\tilde{\rho}(t_n,.)]\ee
For $n=0$, \eqref{min_n} follows from $\tilde{y}\in\tilde{\mathcal Y}[\tilde{\rho}(.)]$.
Assume the statements hold for $n$. We first prove that
\eqref{antishock_0} holds up to $t_{n+1}$ i.e. (in view of
\eqref{def_tilde})
\begin{eqnarray} \label{antishock_1}
\tilde{\rho}^+_n(t-t_n,(z_n+v(t-t_n))^+) & = & \varphi(\rho_r)\\
\label{antishock_11} \tilde{\rho}^-_n(t-t_n,(z_n+v(t-t_n))^-) & = &
\varphi(\rho_l) \end{eqnarray}
for a.e. $t\in(t_n,t_{n+1})$.
Let $\tilde{\rho}^{++}_n(.,.)$ and $\tilde{\rho}^{--}_n(.,.)$ be the
entropy solutions to \eqref{conservation_law} on $\R$ with initial
data
\begin{eqnarray} \tilde{\rho}^{++}_n(0,.) & := &
\tilde{\rho}^+_n(0,.){\bf 1}_{(-\infty,1)}+\rho_c{\bf
1}_{(1,+\infty)}\label{extensions_1}\\
\tilde{\rho}^{++}_n(0,.) & := & \tilde{\rho}^-_n(.){\bf
1}_{(0,+\infty)}+\rho_c{\bf 1}_{(-\infty,0)}\label{extensions_2}
\end{eqnarray}
with $\rho_c$ given by \eqref{def_critical}. By
\eqref{extensions_1}--\eqref{extensions_2},
\eqref{initial_tildeplus_2}-- \eqref{initial_tildemin_2} and
\eqref{min_n},
\begin{eqnarray}
\label{obtain_1} \int_{z_n}^z\tilde{\rho}_n^{++}(0,x)dx & \leq &
\rho_c(z-z_n),\quad\forall z>z_n\\
\label{obtain_2}
\int_z^{z_n}\tilde{\rho}_n^{--}(0,x)dx & \geq &
\rho_c(z_n-z),\quad\forall z<z_n
\end{eqnarray}
Given \eqref{extensions_1}--\eqref{extensions_2},
\eqref{obtain_1}--\eqref{obtain_2} and \eqref{def_critical}, we
apply Lemma \ref{lemma_pde} below to $\tilde{\rho}^{++}_n$ with
$r=\varphi(\rho_r)\leq\rho=\rho_c$, and $\tilde{\rho}^{--}_n$, with
$r=\varphi(\rho_l)\geq\rho=\rho_c$. Note that strict convexity of
$f$ implies
\begin{eqnarray*}
v(\varphi(\rho_r),\rho_c) & > &
v(\varphi(\rho_r),\varphi(\rho_l))=v\\
v(r,\rho) & = &
v(\varphi(\rho_l),\rho_c)<v(\varphi(\rho_r),\varphi(\rho_l))=v
\end{eqnarray*}
Thus Lemma \ref{lemma_pde} yields
\begin{eqnarray} \label{traces_ppmm_1} \tilde{\rho}^{++}_n
(t,(z_n+v(t-t_n))^+) & = & \varphi(\rho_r)\\
\label{traces_ppmm_2} \tilde{\rho}^{--} (t,(z_n+v(t-t_n)^-) & = &
\varphi(\rho_l)\end{eqnarray}
By finite propagation property (see e.g. Proposition (2.3.6) of
\cite{ser}) for \eqref{conservation_law}, we have
$\tilde{\rho}^+_n(t,.)=\tilde{\rho}^{++}_n(t,.)$ on
$(-\infty,1-V(t-t_n))$ and
$\tilde{\rho}^-_n(t,.)=\tilde{\rho}^{--}_n(t,.)$ on
$(V(t-t_n),+\infty)$. With
\eqref{traces_ppmm_1}--\eqref{traces_ppmm_2} and \eqref{seq_t}, this
implies \eqref{antishock_1} on $(t_n,t_{n+1})$.
Now we prove \eqref{min_n} for $n+1$. For
$t\in[-t_{n+1},-t_n]$, set $y_t=1-z_n-v(-t-t_n)\in(0,1)$, so
$1-z_{n+1}=y_{-t_{n+1}}$. Define $\rho(.,.)$ by \eqref{reversal}, so \eqref{antishock_1}--\eqref{antishock_11}
and \eqref{right_bln_n}--\eqref{left_bln_n} imply
\eqref{inside_shock}--\eqref{boundary_2} for a.e.
$\theta\in(-t_{n+1},-t_n)$. Applying \eqref{difference_gamma} with this $y_.$, using 
(o)--(ii) of Lemma \ref{inequalities}, we obtain
\be\label{thisy}
S[\rho(-t_n,.),y_{-t_n}]-S[\rho(-t_{n+1},.),y_{-t_{n+1}}]=J_{(-t_{n+1},-t_n)}[\rho(.,.)]
\ee
Then
\eqref{inequality_production_2} and \eqref{min_n} imply
\eqref{min_n} for $n+1$. Besides, summing
\eqref{thisy} over $n$ and letting $n\to\infty$ in
\eqref{min_n}, we have
\begin{eqnarray} \label{summation_n}
S[\rho(.)]-S[\rho(-\tilde{\theta}^{\tilde{y}},.)] & = & J_{(-\tilde{\theta}^{\tilde{y}},0]}[\rho(.,.)]\\
\label{limit_n} \tilde{y}_{\infty} & \in & \tilde{\mathcal
Y}[\tilde{\rho}(-\tilde{\theta}^{\tilde{y}},.)] \end{eqnarray}
\textbf{Construction on the time interval $[\tilde{\theta}^{\tilde{y}},+\infty)$.} If
$\tilde{\theta}^{\tilde{y}}<+ \infty$, which occurs unless $f(\rho_l)=f(\rho_r)$ and
$y\in(0,1)$, we define $\tilde{\rho}$ on $(\tilde{\theta}^{\tilde{y}},+\infty)$ as the
unique entropy solution to \eqref{conservation_law} on $(0,1)$ with
initial datum $\tilde{\rho}(\tilde{\theta}^{\tilde{y}},.)$ and BLN boundary conditions
\be \label{bln_tilde} \tilde{\rho}(t,0^+)\in{\mathcal
E}^+(\rho_b),\quad \tilde{\rho}(t,1^-)\in\bln^-(\rho_b)
\ee
with $\rho_b$ defined in (c) of Theorem \ref{theorem_2}. The second
boundary condition in \eqref{reverse_bln_2} is thus satisfied. Let
us show that the first one holds, i.e.
\be \label{closed_bln} \tilde{\rho}(t,\tilde{y}^\sigma)=\rho_b\ee
with $\sigma$ and $\rho_b$ given in (c). We prove by induction that,
for $n\in\N$ and $t_n:=\tilde{\theta}^{\tilde{y}}+n/V$,
\be \label{min_1} \tilde{y}_\infty\in\tilde{\mathcal Y}[\tilde{\rho}(-t_n,.)]
\ee
and that this implies \eqref{closed_bln} a.e. on $(t_n,t_{n+1})$.
For $n=0$, \eqref{min_1} is simply \eqref{limit_n}.
We consider the case $\tilde{y}_\infty=0$, as $\tilde{y}_\infty=1$
is symmetric. Here $\rho_b=\rho_l$, thus we must show
\be \label{rhol} 
\tilde{\rho}(t,0^+)=\rho_l\ee
for a.e. $t\in(t_n,t_{n+1})$.
Let $\tilde{\rho}^+_n(.,.)$ and $\tilde{\rho}^{++}_n(.,.)$ be the
entropy solutions to \eqref{conservation_law}, respectively on
$(-\infty,1)$ and $\R$, with initial/boundary conditions
\begin{eqnarray} \label{def_tildeplus_rhol}
\tilde{\rho}_n^+(0,.) & = & \tilde{\rho}(t_n,.){\bf
1}_{(0,1)}+\rho_l{\bf 1}_{(-\infty,0)}\\
\label{def_tildepp_rhol} \tilde{\rho}_n^{++}(0,.) & = &
\tilde{\rho}(t_n,.){\bf 1}_{(0,1)}+\rho_l{\bf
1}_{(-\infty,0)}+\rho_c{\bf 1}_{(1,+\infty)}\\
\label{another_bln} \tilde{\rho}_n^+(t,1^-) & \in & {\mathcal
E}^-(\rho_l)
\end{eqnarray}
By \eqref{min_1},
$$
\int_0^x\tilde{\rho}^{++}_n(0,z)dz\leq\rho_ cx
$$
Note that
\be \label{implication_1} \tilde{y}_\infty=0\Rightarrow
f(\rho_l)\leq f(\rho_r)\Rightarrow \rho_l\leq\rho^*\ee
which, by \eqref{def_critical}, implies $\rho_l\leq\rho_c$. Thus
Lemma \ref{lemma_pde} yields
$\tilde{\rho}^{++}_n(t,x)=\rho_l$ for every $x<t v(\rho_l,\rho_c)$.
By \eqref{def_critical} and strict concavity of $f$,
$v(\rho_l,\rho_c)>v(\rho_l,\varphi(\rho_r))\geq 0$,
where the last inequality follows from
$f(\varphi(\rho_r))=f(\rho_r)\geq f(\rho_l)$. Thus
\be \label{bln_pp} \tilde{\rho}^{++}_n(t,0^+)=\rho_l \ee
for a.e. $t>0$. By finite propagation property,
$\tilde{\rho}^+(t,.)$ and $\tilde{\rho}^{++}(t,.)$ coincide on
$(-\infty,1-V(t-t_n))$. Thus, by definition of $t_n$, we also have
\be\label{bln_p}\tilde{\rho}^{+}_n(t,0^+)=\rho_l\in{\mathcal
E}^+(\rho_l)\ee
for a.e. $t\in(0,t_{n+1}-t_n)$.
By \eqref{bln_tilde}, \eqref{def_tildeplus_rhol},
\eqref{another_bln} and \eqref{bln_p}, $\tilde{\rho}(t_n+.,.)$ and
$\tilde{\rho}^+_n$ are entropy solutions of a common
initial-boundary problem for \eqref{conservation_law} on
$(0,t_{n+1}-t_n)\times(0,1)$. Thus they coincide on this domain.
This in particular establishes \eqref{rhol} on $(t_{n},t_{n+1})$.
We now apply \eqref{difference_gamma} on $[-t_{n+1},-t_n]$ with $y_t\equiv 1$.
Using  \eqref{bln_tilde},  \eqref{rhol} and \eqref{implication_1},
with (i), (ii), (iii) of Lemma \ref{inequalities}, we obtain \eqref{thisy}.
Then \eqref{min_1} and \eqref{inequality_production_2}
imply \eqref{min_1} for $n+1$.
\begin{lemma}
\label{lemma_pde}
\mbox{}\\ \\
Let $0\leq r\leq\rho\leq K$, and $u(.,.)$ be the entropy solution to
\eqref{conservation_law} on $\R$, with initial datum $u_0(.)\in
L^{\infty,K}(\R)$ such that
$u_0(x)=r$ for $x<0$ and
$\int_0^xu_0(z)dz\leq \rho x$ for $x>0$
(resp. $\int_x^0 u_0(z)dz\geq -rx$ for $x<0$ and $u_0(x)=\rho$ for
$x>0$). Then
$u(t,x)=r$ for every $x <t v(r,\rho)$ (resp. $u(t,x)=\rho$ for every
$x>t v(r,\rho)$).
\end{lemma}
\begin{proof}{lemma}{lemma_pde}
We prove the first half of the statement, the other one being
similar. We have
\be \label{burgers_hj} u(t,x)=-\partial_x U(t,x) \ee
where $U(t,x)$ denotes the viscosity solution to the
Hamilton-Jacobi equation
\be \label{hj}
\partial_ t U=f(-\partial_x U),\quad U_0(x)=-\int_0^x u_0(z)dz
\ee
given by Hopf's formula,
\be \label{hopf} U(t,x)=\inf_{y\in\R}\left[ U_0(y)+t
f^*\left(\frac{x-y}{t}\right)\right] \ee
where the convex conjugate of the concave function $f$ is given by
\be \label{conjugate}
f^*(v):=\sup_{\rho\in[0,K]}[f(\rho)-v\rho]=f[(f')^{-1}(v)]-v(f')^{-1}(v)
\ee
and satisfies
\be \label{derivative_conjugate} (f^*)'=-(f')^{-1} \ee
Set
$$
H(y)=U_0(y)+t f^*\left(\frac{x-y}{t}\right)
$$
For $y< 0$ we have, by assumption on $u_0(.)$,
$$H(y) =-ry+t
f^*\left(\frac{x-y}{t}\right)$$
hence, by \eqref{derivative_conjugate}
\be \label{derivative_H}
H'(y)=-r+(f')^{-1}\left(\frac{x-y}{t}\right)\ee
Assuming $x\leq tv(r,\rho)$, we have $x-t f'(r)\leq 0$, since
concavity of $f$ and $r\leq\rho$ imply $v(r,\rho)\leq f'(r)$. Thus
$H'(y)\leq 0$ for $y\leq x-t f'(r)$ and $H'(y)\geq 0$ for $x-
tf'(r)\leq y<0$, so that
\be \label{inf_negative}
\inf_{y\in(-\infty,0)}H(y)=H(x-tf'(r))=-rx+tf(r) \ee
For $y>0$ we have, by assumption on $u_0(.)$,
$$
H(y)\geq G(y):=-\rho y+t f^*\left(\frac{x_y}{t}\right)
$$
Two cases may arise: (a) If $x-tf'(\rho)\leq 0$, then we have
\be \label{inf_positive_1} \inf_{y\in(0,+\infty)}H(y)\geq
\inf_{y\in(0,+\infty)}G(y)=G(0)=H(0)\geq\inf_{y\in(-\infty,0)}H(y)
\ee
(b) if $x-tf'(\rho)>0$, then we have
\be \label{inf_positive_2} \inf_{y\in(0,+\infty)}H(y)\geq
\inf_{y\in(0,+\infty)}G(y)=G(x-tf'(\rho))=-\rho x+t f(\rho)> -r x+
tf(r)=\inf_{y\in(-\infty,0)}H(y) \ee
where the last inequality follows from $x<tv(r,\rho)$. From
\eqref{inf_negative} and
\eqref{inf_positive_1}--\eqref{inf_positive_2} we obtain that, for
every $x<tv(r,\rho)$, $U(t,x)=-rx+tf(r)$, and thus $\rho(t,x)=r$
by \eqref{burgers_hj}.
\end{proof}
\section{The case $\rho_l\geq\rho_r$}\label{sec_relax}
\subsection{Analysis of entropy production} \label{preparatory}
\begin{lemma}
\label{lemma_2}
Let
$\rho(.,.)\in \emsol((-\infty,0)\times(0,1))$. Let $\tilde{G}(.,.)$, defined on $(0,+\infty)\times\R$, be the function defined in Theorem \ref{theorem_3}. For $(t,x)\in(-\infty,0)\times(0,1)$, set $F_t(x)=\varphi[\tilde{G}(-t,1-x)]$.
Then, for every $-\infty< a<b\leq 0$,
\be\label{computation_4}
S[\rho(b,.),F_b(.)]-S[\rho(a,.),F_a(.)]-J_{(a,b)}[\rho(.,.)]=
\int_a^b d(t)dt-\mu^-[(a,b)\times(0,1)]
\ee
where $d(t)=\partial d(t)+d^\circ(t)$, with
\begin{eqnarray}\label{computation_4_bulk}
d^\circ(t) & := & -\partial_x F_t(x)h''(F_t(x))\\
& \times & \left\{f(\rho(t,x))-f\circ\varphi(F_t(x))-f'\circ\varphi[F_t(x)](\rho(t,x)-\varphi[F_t(x)])\right\}\nonumber\\
\label{computation_4_boundary}
\partial d(t) & := & g(\rho(t,0^+),F_t(0+))-g(\rho(t,1^-),F_t(1-))
\end{eqnarray}
\end{lemma}
%
%
%
%
%
%
\begin{lemma}
\label{lemma_boundaries}
Assume $\rho_l\geq F^+\geq\rho^*\geq F^-\geq \rho_r$. Then:\\ \\
1) For every $\rho^\pm\in[0,K]$,
\begin{eqnarray}
\label{inequality_boundary_1} g(\rho^+,F^+) & \leq &
i^l(\rho,\rho_l)\\
\label{inequality_boundary_2} \quad -g(\rho^-,F^-) & \leq &
i^r(\rho,\rho_r)
\end{eqnarray}
2) Equality occurs in \eqref{inequality_boundary_1} iff.: either $F^+=\rho_l$ and $\rho^+\in\bln^-(\rho_l)$; or, $F^+<\rho_l$ and
$\rho^+\in\{F^+,\varphi(F^+)\}\subset \bln^-(F^+)$. Equality occurs in \eqref{inequality_boundary_2} iff.: either $F^-=\rho_r$ and $\rho^-\in\bln^+(\rho_r)$; or, $F^->\rho_r$ and
$\rho^-\in\{F^-,\varphi(F^-)\}\subset\bln^+(F^-)$
\end{lemma}
\begin{proof}{lemma}{lemma_boundaries}
We will prove the statement for $\rho^+$ and $F^+$, the other part being similar. The case $F^+=\rho_l$ follows from ii) of Lemma \ref{inequalities}. Assume now $F^+<\rho_l$. If $\rho^+\geq\rho^*$, then by Lemma \ref{lemma_g} and \eqref{left_2},
$$g(\rho^+,F^+)\leq 0=i^l(\rho^+,\rho_l)$$
with equality iff. $\rho^+=F^+$. We assume now $\rho^+<\rho^*$, and distinguish two cases.
i) $\varphi(\rho_l)\leq\rho^+<\rho^*$. Then by \eqref{left_2} and Lemma \ref{lemma_g},
$$g(\rho^+,F^+)-i^l(\rho^+,\rho_l)=g(\varphi(\rho^+),F^+)\leq 0$$
with equality iff. $\varphi(\rho^+)=F^+$.
ii) $\rho^+<\varphi(\rho_l)$, then by \eqref{left_2}, $i^l(\rho^+,\rho_l)=g(\rho^+,\rho_l)$. To study
$g(\rho^+,F^+)-g(\rho^+,\rho_l)$ we note that, for $\rho^*\leq r<\varphi(\rho^+)$,
$$
\frac{d}{dr}g(\rho^+,r)=h''(r)[f(r)-f(\rho^+)]>0
$$
Thus, since $\rho^*\leq F^+\leq\rho_l<\varphi(\rho^+)$, we have $g(\rho^+,F^+)\leq g(\rho^+,\rho_l)$ with equality iff.
$F^+=\rho_l$, in which case $\rho^+\in\bln^-(\rho_l)$.
\end{proof}
\mbox{}\\
\begin{proof}{lemma}{lemma_2}
We write
$$S[\rho(t,.),F_t(.)]=H(t)-J(t)+\kappa(t)$$
where
$$H(t):=\int_0^1 h[\rho(t,x)]dx, \,J(t):=\int_0^1 [\rho(t,x)h'(F_t(x))]dx, \,\kappa(t):=\int_0^1 k[F_t(x)]dx$$
%
%
By the Gauss-Green formula for divergence-measure fields (\cite{cf}),
\be\label{hterm}
H(b)-H(a)=\int_a^b\left[g(\rho(t,0^+))-g(\rho(t,1^-))\right]dt+\mu[(a,b)\times(0,1)]
\ee
Since $\tilde{G}(0,.)$ is nonincreasing, for $t>0$, $\tilde{G}$ is a solution of class $C^1$ constructed by standard characteristics (see \eqref{std_char_1}--\eqref{std_char_2} below). Thus by \eqref{conservation_law} and the generalized product differentiation rule for the divergence-measure field $(\rho(t,x),f(\rho(t,x))$ (\cite[Theorem 3.1]{cf}),
$$
\partial_t[\rho(t,x)h'(F_t(x))]+\partial_x[f(\rho(t,x))h'(F_t(x))]=\rho(t,x)\partial_t h'(F_t(x))+f(\rho(t,x))\partial_x h'(F_t(x))
$$
Thus $h'(F_t(x))(\rho(t,x),f(\rho(t,x))$ is a divergence-measure field, and the Green's formula yields
\begin{eqnarray*}
J(b)-J(a) & = & \int_a^b \left[
f(\rho(t,0^+)h'(F_t(0^+))-f(\rho(t,1^-)h'(F_t(1^-))\right]dt\\
& + &
\int_a^b\int_0^1
\left[\rho(t,x)\partial_t h'(F_t(x))+f(\rho(t,x))\partial_x h'(F_t(x))\right]dxdt
\end{eqnarray*}
Since $\tilde{G}(.,.)$ satisfies \eqref{conservation_law}, we have $\partial_t F_t(x)=-f'(\varphi(F_t(x))\partial_x F_t(x)$.
By \eqref{def_k}, $\partial_t k(F_t)=\varphi(F_t)h''(F_t)\partial_t F_t$, thus the result follows by a simple computation, noting that (since $f\circ\varphi=f$)
\begin{eqnarray*}
\int_0^1 f(\varphi(F_t(x))h''(F_t(x))\partial_x F_t(x)dx & = & [h'(F_t(1^-))f(F_t(1^-))-g(F_t(1^-))]\\
& - & [h'(F_t(0^+))f(F_t(0^+))-g(F_t(0^+))]
\end{eqnarray*}
\end{proof}
\subsection{$S\leq V$  and uniqueness of a minimizer}
\label{unique_relax}
We need the following properties of $\tilde{G}(.,.)$ defined in Theorem \ref{theorem_3}. Property 1) below implies \eqref{bln_tildeg} in Theorem \ref{theorem_3}.
\begin{lemma}
\label{explicit_tildeg}
\mbox{}\\ \\
0) For all $t>0$, $\tilde{G}(t,.)$ is a nonincreasing function.\\ \\
1) For all $t>0$ and $x\in(0,1)$:
a) $\varphi(\rho_l)\leq\tilde{G}(t,x)\leq\varphi(\rho_r)$;
b) $\rho_l\leq\rho^*$  implies $\tilde{G}(t,1^-)=\varphi(\rho_l)$,
 $\rho_l>\rho^*$  implies $\tilde{G}(t,1^-)\leq\rho^*$;
c)  $\rho_r\geq\rho^*$ implies $\tilde{G}(t,0^+)=\varphi(\rho_r)$,
 $\rho_r<\rho^*$ implies $\tilde{G}(t,0^+)\geq\rho^*$.\\ \\
2) For every $x\in[0,1],\tilde{G}(t,x)\to\varphi(\rho_s)$  as $t\to\infty$, where $\rho_s$ is the uniform density defined in \eqref{eq_stat}.
Outside the MC phase (i.e. the third case in \eqref{eq_stat}), there exists $T>0$ such that $\tilde{G}(t,x)=\rho_s$ for all $t\geq T$ and $x\in[0,1]$.
\end{lemma}
\begin{proof}{lemma}{explicit_tildeg}
Since $\tilde{G}(0,.)$ is nonincreasing and $f$ concave, $\tilde{G}(.,.)$ is shock-free and can be constructed by standard characteristics.
For every $(t,x)\in(0,+\infty)\times\R$, there is a unique $y=y(x,t)$ such that
\be\label{std_char_1}
y+tf'(\tilde{G}(0,y^-)) \leq x\leq y+tf'(\tilde{G}(0,y^+))
\ee
and we have
\be\label{std_char_2} \tilde{G}(t,x) =  (f')^{-1}\left(\frac{x-y}{t}\right)\ee
All statements are simple consequences of \eqref{std_char_1}--\eqref{std_char_2}.
\end{proof}
\mbox{}\\ \\
Let $\rho(.,.)\in{\mathcal R}[\rho(.)]$. We apply
\eqref{computation_4} on $(T,0)$.
Then $d^\circ(t)\leq 0$, because $h$ is convex, $F_t$ is nonincreasing, and $f$ is concave and $\partial d(t)\leq 0$ by \eqref{inequality_boundary_1}--\eqref{inequality_boundary_2}, ii) of Lemma \ref{inequalities} and 1), b)--c) of Lemma \ref{explicit_tildeg}
(note that if, for instance, $\rho_l<\rho^*$, \eqref{inequality_boundary_1} does not follow from Lemma \ref{lemma_boundaries} but from ii) of Lemma \ref{inequalities},
because then we have $F^+=F(t,0^+)=\rho_l$ by 1), b) of Lemma \ref{explicit_tildeg}. $\rho_r>\rho^*$ is treated similarly). Since $F_0$ is the maximizing $F$ for $\rho(.)$ in  \eqref{def_S_4}, we have
\be\label{neg_int} S[\rho(.)] \leq  S[\rho(T,.),F_T(.)]+ J_{(T,0)}[\rho(.,.)]\leq S[\rho_T(.)]+J_{(T,0)}[\rho(.,.)]\label{neg_int}
\ee
Sending $T\to-\infty$, 
%
%
%
since $S[\rho_s(.)]=0$,  we obtain $S[\rho(.)]\leq J_{(-\infty,0)}[\rho(.,.)]$, hence $S[\rho(.)]\leq V[\rho(.)]$.
%
%
%
Let us now assume that $\rho(.,.)\in\mathcal R[\rho(.,.)]$ achieves
equality in \eqref{variational}. Then in \eqref{computation_4} we must have $d^\circ(t)=\partial d(t)=0$ a.e. on $(-\infty,0)$. Thus
$$S[\rho(t,.),F_t]=J_{(-\infty,t)}[\rho(.,.)]\geq S[\rho(t,.)]$$
%
%
Hence $F_t=F_{\rho(t,.)}$, that is \eqref{later_tildeg}.
Define $\tilde{\rho}(.,.)$ by \eqref{reversal}. $d^\circ(t)=\partial d(t)=0$ on $(-\infty,0)$ has the following implications:\\ \\
1) Inequalities \eqref{inequality_boundary_1}--\eqref{inequality_boundary_2} must be equalities for a.e. $t<0$, with $\rho^\pm=\rho(t,0^\pm)$.
Then 2) of Lemma \ref{lemma_boundaries}, ii) of Lemma \ref{inequalities} and 1), b)--c) of Lemma \ref{explicit_tildeg} imply $\tilde{\rho}(t,0^+)\in\bln^+[\tilde{G}(t,0^+)]$ and  $\tilde{\rho}(t,1^-)\in\bln^-[\tilde{G}(t,1^-)]$ for every $t>0$ (here again, if $\rho_l<\rho^*$ or $\rho_r>\rho^*$, Lemma \ref{inequalities} is used instead of Lemma \ref{lemma_boundaries}).
\\ \\
2) $\mu[(-\infty,0)\times(0,1)]=\mu^+[(-\infty,0)\times(0,1)]$, hence $\mu=\mu_h[\rho(.,.)]\geq 0$ on $(-\infty,0)\times(0,1)$, and by \eqref{reversal_prod}, $\mu_h[\tilde{\rho}(.,.)]\leq 0$ on $(0,+\infty)\times(0,1)$. Since \eqref{conservation_law} is invariant by space-time reversal, $\tilde{\rho}$ is also a weak solution. Thus, by
\cite{dow}, $\tilde{\rho}(.,.)$ is an entropy solution to \eqref{conservation_law}
on $(0,+\infty)\times(0,1)$.\\ \\
Since $\rho(.,.)\in C^0((-\infty,0),L^1(0,1))$, 1) and 2) imply that $\tilde{\rho}(.,.)$ is the entropy solution to \eqref{conservation_law} with BLN boundary conditions \eqref{bln_f} and initial datum \eqref{initial_tilderho}.
Thus, if $\rho\in\mathcal R[\rho(.)]$ achieves equality in \eqref{variational}, it is necessarily (and uniquely) determined by the construction of Theorem \ref{theorem_3}, and satisfies \eqref{later_tildeg}.
\subsection{$S\geq V$ and existence of a minimizer}
\label{lower_relax}
Let us first recall the construction of the maximizing function $F(.)$ in \eqref{def_S_3}.
Let $R(.)$ be a continuous function on $[0,1]$. and $-\infty\leq\alpha\leq\beta\leq +\infty$. We define the $(\alpha,\beta)$-truncated convex hull
$\mathcal C_{\alpha,\beta}R$ of $R(.)$ as the upper envelope of all convex functions $S(.)$ on $[0,1]$ such that $S\leq R$ and
$\alpha\leq S'(0^+)\leq S'(1^-)\leq\beta$.
A straightforward extension of a computation of \cite{dls} shows that the supremum in
\eqref{def_S_4} is achieved uniquely for $F=F_{\rho(.)}$ given by
\be \label{optimal_F}
\varphi[F_{\rho(.)}(x)]=\frac{d}{dx}\mathcal C_{\varphi(\rho_l),\varphi(\rho_r)}R(x)
\ee
where $R(x):=\int_0^x\rho(y)dy$.
Alternatively, we have
$$\mathcal C_{\alpha,\beta}R(x)=\sup_{\theta\in[\alpha,\beta]\cap\R}[\theta x-R^*(\theta)]$$
where
$$R^*(\theta):=\sup_{x\in[0,1]}[\theta x-R(x)]$$
is the Legendre-Fenchel transform of $R(.)$. $\mathcal C_{-\infty,+\infty}R=R^{**}$ is the usual (untruncated) convex hull.  One-sided derivatives $(\mathcal C_{\alpha,\beta}R)'(x^\pm)\in\bar{\R}$
are such that $[(\mathcal C_{\alpha,\beta}R)'(x^-),(\mathcal C_{\alpha,\beta}R)'(x^+)]$ is the set of maximizers of the concave function $\theta\mapsto\theta x-R^*(\theta)$ on $[\alpha,\beta]$. Note that
$(\mathcal C_{\alpha,\beta}R)'(x^\pm)=(\mathcal C R)'(x^\pm)_\alpha^\beta$, where $x_\alpha^\beta:=\max[\alpha,\min(x,\beta)]$.
%
%
The above definitions are easily adapted to the truncated concave hull $\mathcal C^{\alpha,\beta}$ for $\alpha\geq\beta$,
by replacing suprema with infima and maximizers with minimizers.\\ \\
We shall need a Hopf-Lax type formula derived in \cite{jvg}
for the entropy solution to \eqref{conservation_law} on
$(0,+\infty)\times(0,1)$ with Cauchy datum $\rho_0(.)\in
L^\infty((0,1))$ and time-dependent boundary conditions 
\be\label{time_bln}
\rho(t,0^+)\in\bln^+[\rho_l(t)],\quad
\rho(t,1^-)\in\bln^-[\rho_r(t)]\ee
a.e. in the sense of \cite{bln}, where 
$\rho_l(.)$ and $\rho_r(.)$ are boundary data in $L^{\infty,K}(0,+\infty)$.
%
%
\begin{theorem}(\cite{jvg})
\label{th_jvg} For $x,y\in[0,1]$ and $t>0$, let $C(y,x,t)$ denote
the set  of differentiable paths $\gamma:[0,t]\to[0,1]$ such that
$\gamma(0)=y$, $\gamma(1)=x$ and $f'(K)\leq \dot{\gamma}_s\leq
f'(0)$ a.e. Let $R(0,x)=\int_0^x\rho(0,y)dy$ and define
\begin{eqnarray}
R(t,x)& := & \sup_{y\in[0,1],\gamma_.\in C(y,x,t)}\left\{
R(0,y)-\int_0^t 1_{\{0\}}(\gamma_s)f[\min(\rho_l(s),\rho^*)]ds\right.\nonumber\\
& - &  \int_0^t 1_{\{1\}}(\gamma_s)f[\max(\rho_r(s),\rho^*)]ds\nonumber\\
& + & \left.\int_0^t
1_{(0,1)}(\gamma_s)f^*(\dot{\gamma}_s)ds\right\}\label{variational_bln}
\end{eqnarray}
where $f^*(\theta):=\inf_{\rho\in[0,K]}[\theta\rho-f(\rho)]$ is the
(concave) Fenchel-Legendre transform of $f$. Then $\partial_x
R(t,x)$ exists a.e. and is the entropy solution to
\eqref{conservation_law} with initial data $\rho(0,.)$ and boundary data \eqref{time_bln}. Besides, the supremum in \eqref{variational_bln} can be
restricted to the subset of $C(y,x,t)$ consisting of piecewise
linear paths with constant slope between the boundaries.
\end{theorem}
%
%
We now show that $\rho(.,.)\in{\mathcal R}[\rho(.)]$ constructed in Theorem \ref{theorem_3} achieves equality in \eqref{quasi_potential} and \eqref{variational}, and that
outside the MC phase it lies in $\mathcal R_T[\rho(.,.)]$ (thus achieving equality in \eqref{quasi_potential_finite} and \eqref{finite_time}).
We interpret the restriction of $\tilde{G}(.,.)$ to $(0,+\infty)\times(0,1)$ as the entropy solution to \eqref{conservation_law} on $(0,1)$ with carefully chosen
BLN boundary data and use  Theorem \ref{th_jvg}. Set $\tilde{R}(0,x):=\int_0^x\tilde{\rho}(0,y)dy$. By \eqref{initial_tildeg}, we have that $\tilde{G}(0,x)=\frac{d}{dx}\tilde{S}(0,x)$, where
\be\label{tilders}\tilde{S}(0,.)=\mathcal C^{\varphi(\rho_r),\varphi(\rho_l)}\tilde{R}(0,.)
\ee
Following \eqref{adset_left_2}--\eqref{adset_right_2}, 
it is trivial that $\tilde{G}(t,0^+)\in\bln^+[\tilde{G}(t,0^+)]$ and $\tilde{G}(t,1^-)\in\bln^-[\tilde{G}(t,1^-)]$. On the other hand,
using 1), b)--c) of Lemma \ref{explicit_tildeg}, one can see that
if $\rho_r<\rho^*$, then we also have $\tilde{G}(t,0^+)\in\bln^+[\varphi(\tilde{G}(t,0^+))]$. Likewise if $\rho_l>\rho^*$ we also have
$\tilde{G}(t,1^-)\in\bln^-[\varphi(\tilde{G}(t,1^-))]$. Choosing the suitable boundary condition in each case,
Theorem \ref{th_jvg} yields $\tilde{G}(t,.)=\partial_x\tilde{S}(t,.)$, where
\be\label{variational_bln_tildeg}
\tilde{S}(t,x) :=  \sup_{y\in[0,1],\gamma_.\in C(y,x,t)}\left\{
S(0,y)-\int_0^t 1_{\{0,1\}}(\gamma_s)f[\tilde{G}(s,\gamma_s)]ds + \int_0^t 1_{(0,1)}(\gamma_s)f^*(\dot{\gamma}_s)ds\right\}
\ee
On the other hand, by definition, $\tilde{G}(t,x)$ is the entropy solution to \eqref{conservation_law} on $(0,+\infty)\times\R$ with initial
datum \eqref{barg}. It is thus also given by the standard Lax-Hopf formula, hence for $x\in[0,1]$,
\begin{eqnarray}
\tilde{S}(t,x) & := &  \sup_{y\in[0,1],\gamma_.\in \bar{C}(y,x,t)}\left\{
\bar{S}_0(y)+ \int_0^t 1_{(0,1)}(\gamma_s)f^*(\dot{\gamma}_s)d
\right\}\nonumber\\
& = &  \sup_{y\in[0,1]} \left\{
\bar{S}_0(y)+t f^{*}\left(\frac{x-y}{t}\right)\right\}\label{variational_std_tildeg}
\end{eqnarray}
where $\bar{S}(0,x)=\int_0^x \bar{G}(0,y)dy$, and $\bar{C}(y,x,t)$ is the set of differentiable paths $\gamma_.:[0,t]\to\R$ such that $\gamma_0=y$, $\gamma_t=x$ and $f'(K)\leq\dot{\gamma}_s\leq f'(0)$ a.e. on $[0,t]$. \eqref{variational_std_tildeg} has a unique minimizer $y=y(x,t)$ given by \eqref{std_char_1}, and satisfying \eqref{std_char_2}.
Applying again Theorem \ref{th_jvg} to $\tilde{\rho}(.,.)$ defined in Theorem \ref{theorem_3}, 
we have $\tilde{\rho}(t,x)=\partial_x \tilde{R}(t,x)$, with
\begin{eqnarray}
\tilde{R}(t,x)& := & \sup_{y\in[0,1],\gamma_.\in C(y,x,t)}\left\{
\tilde{R}(0,y)-\int_0^t 1_{\{0\}}(\gamma_s)f[\min(\varphi[\tilde{G}(s,0^+)],\rho^*)]ds\right.\nonumber\\
& - &  \int_0^t 1_{\{1\}}(\gamma_s)f[\max(\varphi[\tilde{G}(s,1^-)],\rho^*)]ds\nonumber\\
& + & \left.\int_0^t 1_{(0,1)}(\gamma_s)f^*(\dot{\gamma}_s)ds\right\}\label{variational_bln_tilderho}
\end{eqnarray}
We are now ready to establish the following properties, from which the result of this subsection will follow easily.
\begin{lemma}
\label{lemma_existence}
For every $t>0$, \eqref{later_tildeg} holds, or equivalently,
\be\label{orequivalently}\tilde{S}(t,.)=\mathcal C^{\varphi(\rho_r),\varphi(\rho_l)}\tilde{R}(t,.)\ee
\end{lemma}
\begin{corollary}
\label{corollary_existence}
For a.e. $t<0$, equality holds in \eqref{inequality_boundary_1}--\eqref{inequality_boundary_2}, with
$\rho^+=\rho(t,0^+)$, $\rho^-=\rho(t,1^-)$, $F^+=F_t(0^+)$, $F^-=F_t(1^-)$.
\end{corollary}
\begin{proof}{corollary}{corollary_existence}
We treat \eqref{inequality_boundary_1}, the argument being similar for \eqref{inequality_boundary_2}.
Assume first that $\rho_l\leq\rho^*$, then 1) of Lemma \ref{explicit_tildeg} implies $\varphi(\tilde{G}(-t,1^-))=\rho_l$. The result then follows from ${\rho}(t,0^+)=\tilde{\rho}(-t,1^-)\in\bln^-[\varphi(\tilde{G}(-t,1^-))]$ and ii) of Lemma \ref{inequalities}.
Assume now $\rho_l>\rho^*$. If $\tilde{G}(-t,1^-)=\varphi(\rho_l)$, we argue as above.  If $\tilde{G}(-t,1^-)>\varphi(\rho_l)$,
Lemma \ref{lemma_existence} implies
$\rho(t,0^+)=\tilde{\rho}(-t,1^-)\geq\tilde{G}(-t,1^-)$.
On the other hand, by 1), b) of Lemma \ref{explicit_tildeg}, $\varphi[\tilde{G}(-t,1^-)]\geq\rho^*$.
This and ${\rho}(t,0^+)=\tilde{\rho}(-t,1^-)\in\bln^-[\varphi(\tilde{G}(-t,1^-))]$ imply (see \eqref{adset_left})
${\rho}(t,0^+)\in\{\varphi(\tilde{G}(-t,1^-)),\tilde{G}(-t,1^-)\}=\{F_t(0^+),\varphi(F_t(0^+))\}$. The result then follows from 2) of lemma \ref{lemma_boundaries}.
\end{proof}
\mbox{}\\
\begin{proof}{lemma}{lemma_existence}
We check properties (1)--(3) below, which imply the result.\\ \\
1) $\tilde{S}(t,.)$ is a concave function, $\varphi(\rho_r)\geq\partial_x\tilde{S}(t,0^+)\geq \partial_x\tilde{S}(t,1^-)\geq\varphi(\rho_l)$, and
\be\label{condition_a}\tilde{R}(t,x)\leq\tilde{S}(t,x)\ee
The first two properties follow from $\tilde{G}(t,.)=\partial_x\tilde{G}(t,.)$ and 0) of Lemma \ref{explicit_tildeg}. \eqref{condition_a} follows from
 \eqref{variational_bln_tildeg}--\eqref{variational_bln_tilderho}, $\tilde{R}(0,.)\leq\tilde{S}(0,.)$ (which is implied by  \eqref{tilders}), and the inequalities
 (cf. \eqref{symmetry_f})
$$
f[\min(\varphi[\tilde{G}(s,0^+)],\rho^*)]\geq f(\tilde{G}(s,0^+)),\quad
f[\max(\varphi[\tilde{G}(s,1^-)],\rho^*)]\geq f(\tilde{G}(s,1^-))
$$
(2) $\tilde{S}(t,.)=\tilde{R}(t,.)$ on $(0,1)\backslash I_t$, where $I_t$ denotes the union of all (relatively) open subintervals of $(0,1)$ on which $\tilde{S}(t,.)$ is affine. To prove this we write $I_t=(0,1)\cap\mathcal I_t$, where $\mathcal I_t$ denotes the union of all open subintervals of $\R$ on which $\tilde{S}(t,.)$ is affine. By \eqref{std_char_1}--\eqref{std_char_2}),
\be\label{transport_linear}
\mathcal I_t=\{x\in\R:\,y(x,t)\in \mathcal I_0\}
\ee
Assume now that $x\in(0,1)\backslash I_t$.
We first claim that $y=y(x,t)\in[0,1]\backslash I_0$. Indeed, suppose e.g. that $y<0$, then \eqref{barg} and \eqref{std_char_1}--\eqref{std_char_2}
imply that $x=y+t f'(\varphi(\rho_l))$, hence $\rho_l>\rho^*$, and $\tilde{G}(t,.)=\varphi(\rho_l)$ in a neighborhood of $x$, which would imply $x\in I_t$.
A similar argument holds if we suppose $y>1$. From $y\in[0,1]\backslash I_0$ and
\eqref{tilders}, it follows that
$\tilde{S}(0,y)=\tilde{R}(0,y)$. The maximizing path in \eqref{variational_std_tildeg} is $\gamma_s=y+sf'(\rho)$, with $\rho=\tilde{G}(t,x)$ given by $\eqref{std_char_2}$.
This path lies in $C(y,x,t)$, thus the same $y$ and $\gamma_.$ also produce the maximum in \eqref{variational_bln_tildeg}. Since $\gamma_.$ does not see the boundaries (except possibly at time $0$), it produces the same value in \eqref{variational_bln_tildeg} and \eqref{variational_bln_tilderho}.
Thus, in view of \eqref{condition_a}, $y$ and $\gamma_.$ must achieve the supremum also in \eqref{variational_bln_tilderho}, and
$\tilde{R}(t,x)=\tilde{S}(t,x)$.\\ \\
(3) (a) If $\partial_x\tilde{S}(t,0^+)<\varphi(\rho_r)$, then $\tilde{S}(t,0^+)=\tilde{R}(t,0^+)$, and
(b) If $\partial_x\tilde{S}(t,1^-)>\varphi(\rho_r)$, then $\tilde{S}(t,1^-)=\tilde{R}(t,1^-)$. 
To prove this we set
%
\begin{eqnarray*}
y_0 & := & \sup\{x\geq 0:\,\bar{G}(0,x^-)=\varphi(\rho_r)\}\\
y_1 & := & \inf\{x\leq 1:\,\bar{G}(0,x^+)=\varphi(\rho_l)\}
\end{eqnarray*}
Thus
$\tilde{G}(0,y_0-)=\varphi(\rho_r)\geq\tilde{G}(0,y_0+)$ and $\tilde{G}(0,y_1-)\geq\tilde{G}(0,y_1+)=\varphi(\rho_l)$.
Let $t_0=y_0/f'(\varphi(\rho_r))^-$ and $t_1=(1-y_1)/f'(\varphi(\rho_l))^+$, with the convention $0/0=+\infty$. 
Hence $t_0=+\infty$ iff. $\rho_r\geq\rho^*$, and $t_1=+\infty$ iff. $\rho_l\leq\rho^*$.
One can see from 
%
\eqref{std_char_1}--\eqref{std_char_2} that
%
%
\begin{eqnarray*}
t_0 & := & \inf\{t>0:\,\tilde{G}(t,0^+)<\varphi(\rho_r)\}\\
t_1 & := & \inf\{t>0:\,\tilde{G}(t,1^-)>\varphi(\rho_l)\}
\end{eqnarray*}
and, if $t_0$ (resp. $t_1$) is finite
\begin{eqnarray}\label{noflat_r}
\tilde{G}(t_0,0^+) & = & \varphi(\rho_r)>\tilde{G}(t,x),\quad\forall x>0\\
\tilde{G}(t_0,0^+) & = & \varphi(\rho_l)<\tilde{G}(t,x),\quad\forall x>0\label{noflat_l}
\end{eqnarray}
In particular, we have nothing to check for $x=0$ and $t<t_0$, resp. $x=1$ and $t<t_1$. If $t_0=t_1=+\infty$ we are done. Otherwise, we assume e.g.
$t_0<t_1$. The above arguments imply that condition 3) is checked (since void) for $t\in(0,t_0)$. Thus for such $t$, \eqref{orequivalently} holds.
Since $\tilde{\rho}(.,.)$ and $\tilde{G}(.,.)$ lie in $C^0((0,+\infty),L^1(0,1))$, \eqref{orequivalently} holds also at $t=t_0$.
We claim that
\be
\label{contact_initial_1}\tilde{R}(t_0,0^+) = \tilde{S}(t_0,0^+)
\ee
Indeed, assume $\tilde{R}(t_0,0^+)<\tilde{S}(t_0,0^+)$. Then there exists $x_0>0$ such that
$\tilde{R}(t_0,x)<\tilde{S}(t_0,x)$ for $x\in(0,x_0)$. This implies $\tilde{S}(t_0,.)$ is affine, and thus $\tilde{G}(t_0,.)$ constant, on $(0,x_0)$, in contradiction
with \eqref{noflat_r}. 
%
%
Let $t>t_0$. We consider \eqref{variational_bln_tilderho} with time origin at $t_0$ and initial condition $\tilde{R}(t_0,.)$.
Since $t_0<+\infty$ we have (by 1), b) of Lemma \ref{explicit_tildeg}) $\tilde{G}(t,0^+)\geq\rho^*$. Then
the special path $\gamma_s\equiv 0$ for $s\in[t_0,t]$ in \eqref{variational_bln_tilderho}
yields
$$\tilde{R}(t,0^+)\geq\tilde{R}(t_0,0^+)-\int_{t_0}^t f(\tilde{G}(s,0^+))ds$$
On the other hand, $\tilde{S}$ is the viscosity solution
of the Hamilton Jacobi equation
\be\label{hj}
\partial_t\tilde{S}(t,x)+f(\partial_x\tilde{S}(t,x))=0
\ee
Integrating \eqref{hj} on the space interval $(0,+\infty)$ and then on the time interval $[t_0,t]$, and using $\tilde{G}=\partial_x\tilde{S}$, we obtain
$$
\tilde{S}(t,0^+)-\tilde{S}(t_0,0^+)=-\int_{t_0}^t f(\tilde{G}(s,0^+))ds
$$
and thus (by \eqref{contact_initial_1}) $\tilde{S}(t,0^+)\leq\tilde{R}(t,0^+)$. With \eqref{condition_a}, this implies $\tilde{S}(t,0^+)=\tilde{R}(t,0^+)$ for all $t>t_0$. 
Thus condition 3) holds on $[t_0,t_1)$.
It follows that \eqref{orequivalently} holds on the time interval $[t_0,t_1)$. If $t_1=+\infty$ we are done. Otherwise, by continuity as above, \eqref{orequivalently}
holds at $t_1$. Then we proceed as above to show that \eqref{contact_initial_1} holds at time $t_1$, as well as
\be
\label{contact_initial_2}\tilde{R}(t_1,1^-) = \tilde{S}(t_1,1^-)
\ee
and obtain $\tilde{R}(t,1^-) = \tilde{S}(t,1^-)$ for all $t>t_1$, whence condition 3) and \eqref{orequivalently} on this interval.
\end{proof}
\mbox{}\\ \\
We now conclude the proof of \eqref{variational}--\eqref{finite_time} and Theorem \ref{theorem_3} as follows.\\ \\
i) $\rho(.,.)$ defined by \eqref{reversal} achieves
%
%
equality in \eqref{quasi_potential} and\eqref{variational}, and equality in \eqref{quasi_potential_finite} and \eqref{finite_time} outside the MC phase.
Indeed, we apply \eqref{computation_4} on $(T,0)$ with $T<0$.
By \eqref{reversal_prod}, $\mu^-[\rho(t,.)]=0$.  Lemma \ref{lemma_existence} implies $d^\circ(t)=0$ and  Corollary
\ref{corollary_existence} implies $\partial d(t)=0$.  Let $F(.,.)$  be defined by \eqref{later_tildeg}. By Lemma \ref{lemma_existence},  $F(t,.)$ achieves supremum in \eqref{def_S_4} for $\rho(t,.)$. Thus \eqref{computation_4} yields
\be\label{lower_bound_relax}
S[\rho(.)]-S[\rho(T,.)]= J_{(T,0)}[\rho(.,.)]
\ee
for every $T>0$, hence $S[\rho(.)]\geq J_{(-\infty,0)}[\rho(.,.)]$. \\ \\
%
%
%
%
ii)  Outside the MC phase, $\rho(.,.)\in{\mathcal R}_T[\rho(.)]$ for some $T>0$.
Indeed, by 2) of Lemma \ref{explicit_tildeg}, there exists $\tau>0$ such that,
for $t>\tau$, $\tilde{\rho}(.,.)$ is the entropy solution to \eqref{conservation_law} with Cauchy datum $\tilde{\rho}(\tau,.)$ and boundary conditions
$\tilde{\rho}(t,0^+)\in\bln^+(\rho_s)$, $\tilde{\rho}(t,1^-)\in\bln^-(\rho_s)$. By Theorem \ref{th_stat}, since $\rho_s\neq\rho^*$, $\tilde{\rho}(.,.)$ relaxes in finite time to the steady state $\rho_s$.\\ \\
iii) In the MC phase, we still have $\rho(.,.)\in{\mathcal R}[\rho(.)]$. 
%
%
%
Indeed, let $\varepsilon<\min(\rho^*,K-\rho^*)$. By 2) of Lemma \ref{explicit_tildeg}, we can find $\tau>0$ such that
$\rho^*-\varepsilon\leq\varphi(\tilde{G}(t,0^+))\leq \rho^*+\varepsilon$ and $\rho^*-\varepsilon\leq\varphi(\tilde{G}(t,1^-))\leq \rho^*+\varepsilon$
for $t>\tau$.
For $t>\tau$, let  $\tilde{\rho}^+(t,.)$ be the entropy
solution to \eqref{conservation_law} with Cauchy datum $\tilde{\rho}^+(\tau,.)=\tilde{\rho}(\tau,.)$ at time $\tau$,
and boundary conditions $\tilde{\rho}^+(t,0^+)\in\bln^+(\rho^*+\varepsilon)$, $\tilde{\rho}^+(t,1^-)\in\bln^-(\rho^*+\varepsilon)$.
Similarly, define $\tilde{\rho}^-(t,.)$ with boundary data $\rho^*-\varepsilon$. By monotonicity of \eqref{conservation_law} with respect to boundary conditions,
$\tilde{\rho}^-(t,.)\leq\tilde{\rho}(t,.)\leq\tilde{\rho}^+(t,.)$.
By Theorem \ref{th_stat}, since $\rho^*\pm\varepsilon\neq\rho^*$, there exists $\tau'>\tau$ such that
$
\tilde{\rho}^\pm(t,.)\equiv\rho^*\pm\varepsilon
$
for all $t>\tau'$. Thus, for $t>\tau'$, $\rho^*-\varepsilon<\tilde{\rho}(t,.)<\rho^*+\varepsilon$.
This establishes pointwise (and $L^1$) convergence of $\tilde{\rho}(t,.)$ to $\rho_s=\rho^*$ as $t\to+\infty$, hence $\rho(.,.)\in\mathcal R[\rho(.)]$.\\ \\
\textbf{Acknowledgements.} I thank Thierry Bodineau and Mauro Mariani for useful discussions and indicating references \cite{bd1,bd2,bbmn,mar}.
This work was developed within ANR projet LHMSHE (ANR grant BLAN07-2184264), whose support is acknowledged. Part of this work was developed at IHP, whose hospitality is acknowledged, during  Fall 2008 IHP trimester ''Interacting particle systems, statistical mechanics and probability theory''.
%
%
%
\begin{appendix}
\section{Variational expression for the dynamic action functional}\label{dynvar}
Recall from \cite{bbmn} that $(F,G)$ is called an entropy-flux sampler for \eqref{conservation_law}, where
$F(t,x,\rho)$ and $G(t,x,\rho)$ are functions of class $C^2$ with compact support in $(a,b)\times[0,1]\times[0,K]$,  iff., for each $(t,x)$, $(F(t,x,.),G(t,x,.))$ is an entropy-flux pair. We add boundary conditions \eqref{boundary_data} to this definition as follows. We say $(F,G)$ is boundary compatible iff., for $x\in\{0,1\}$, $(F(t,x,.),G(t,x,.))$ is a boundary entropy-flux pair in the sense of \cite{ott}, i.e.
\be\label{boundarypair_1}
F(t,x,\rho_b(x)) = \partial_\rho F(t,x,\rho_b(x))=0
\ee
\be\label{boundarypair_2}
G(t,x,\rho_b(x)) = \partial_\rho G(t,x,\rho_b(x))=0
\ee
where the boundary data $\rho_b(.)$ is given here by $\rho_b(.)=\rho_l 1_{\{0\}}+\rho_r 1_{\{1\}}$.
The $F$-sampled entropy production  of $\rho(.,.)\in L^{\infty,K}((a,b)\times(0,1))$ is defined (\cite{bbmn}) as
\be\label{sample_prod}
P_{F,\rho(.,.)}  :=  -\int_a^b\int_0^1\left[(\partial_t F)(t,x,\rho(t,x))+(\partial_x G)(t,x,\rho(t,x))\right]dxdt
\ee
Let $h$ be a given uniformly convex entropy.
%
%
%
%
%
We denote by
$\overline{\mathcal P}_h(0,1)$ the set of boundary compatible entropy samplers $F$ such that $\partial^2_\rho F(t,x,\rho)\leq h''(\rho)$. 
\begin{theorem}\label{th_vari}
For every $\rho(.,.)\in L^{\infty,K}((a,b)\times(0,1))$,
\be
I_{(a,b)}[\rho(.,.)]=\sup\{P_{F,\rho(.,.)}:\,F\in\overline{\mathcal P}_h(0,1)\}\label{vari_2}
%
\ee
\end{theorem}
%
%
%
\begin{corollary}\label{cor_comp}
for every $c\in [0,+\infty)$, $I_{(a,b)}^{-1}((-\infty,c])$ is compact with respect to the local $L^1$ topology on $L^{\infty,K}((a,b)\times(0,1))$.
\end{corollary}
\begin{proof}{corollary}{cor_comp}
For each $F\in\overline{\mathcal P}_h(0,1)$,  $\rho(.,.)\mapsto P_{F,\rho}$ is continuous, thus $I_{(a,b)}$ is lower-semicontinuous.
We show that $(I^0_{(a,b)})^{-1}((-\infty,c])\supset (I_{(a,b)})^{-1}((-\infty,c])$ is relatively compact. Consider functions
$\eta\in C^2([0,K])$ not necessarily assumed convex, that we call entropy by extension. The entropy flux (by extension) is still defined by \eqref{def_entropy_flux}. 
We may extend \eqref{production} and \eqref{kinetic} to such functions  (write $\eta=\eta_1-\eta_2$, where $\eta_i$ are convex functions of class $C^2$). Let $\Omega$ be a bounded open subset of $(a,b)\times(0,1)$. Since $h''$ is uniformly convex, by \eqref{compare_bulk}, 
$\mu_\eta[\rho(.,.)](\Omega)$ is bounded on $(I^0_{(a,b)})^{-1}((-\infty,c])$. Then the compensated compactness argument following \cite[Proposition 9.2.2]{ser}
yields the desired result (notice that Tartar's equation \cite[(9.15)]{ser} is stable by uniform convergence of entropies, and thus holds for $C^1$ extended entropies
including $f$). 
\end{proof}
\mbox{}\\ \\
For the proof of Theorem \ref{th_vari}, we need the
\begin{lemma}\label{lemma_vari}
Let $\Phi(\rho_0)$ denote the set of entropy-flux pairs $(\eta,q)$ such that 
\be\label{boundarypair}\eta(\rho_0)=\eta'(\rho_0)=0,\quad q(\rho_0)=q'(\rho_0)=0\ee
Then
\begin{eqnarray}\label{sup_ent_1}
i^l(\rho,\rho_l) & = & \sup\{q(\rho):\,(\eta,q)\in\Phi(\rho_l),\,\eta''\leq h''\}\\
\label{sup_ent_2}
i^l(\rho,\rho_r) & = & \sup\{-q(\rho):\,(\eta,q)\in\Phi(\rho_r),\,\eta''\leq h''\}
\end{eqnarray}
The suprema in \eqref{sup_ent_1}--\eqref{sup_ent_2} are respectively achieved by the following entropies
$\eta^l_{\rho_l,\rho}$ and $\eta^r_{\rho_r,\rho}$:
\begin{eqnarray}
\label{optimal_ent_1}
\eta^l_{\rho,\rho_l}(u) & := & \int_{\rho_l}^u(u-v)h''(v)1_{(0,+\infty)}[q_v(\rho,\rho_l)]dv\\ 
\eta^r_{\rho,\rho_r}(u) & := & \int_{\rho_r}^u(u-v)h''(v)1_{(0,+\infty)}[-q_v(\rho,\rho_r)]dv
\label{optimal_ent_2}
\end{eqnarray}
\end{lemma}
\mbox{}\\
\begin{proof}{lemma}{lemma_vari}
The result follows from \eqref{decomp_relative}
applied to $(\eta,q)\in\Phi(\rho_0)$. To maximize $q(\rho)$ and achieve $i^l(\rho,\rho_l)$, one has to choose the  entropy $\eta\in\Phi(\rho_l)$ so that
$\eta''(v)=h''(v)1_{(0,+\infty)}[q_v(\rho,\rho_l)]$, and similarly for $i^r(\rho,\rho_r)$.
\end{proof}
\mbox{}\\
\begin{proof}{theorem}{th_vari}
\mbox{}\\ \\
{\em Step one.} We prove that the r.h.s. of \eqref{vari_2} if finite iff. $\rho(.,.)\in\emsol((a,b)\times(0,1))$, and is then dominated by the l.h.s.
Assume first that $\rho(.,.)\in\emsol((a,b)\times(0,1))$. Then $I_{(a,b)}[\rho(.,.)]$ is finite because $v\mapsto m(v;dt,dx)$ is bounded.
Using the generalized Green's formula (\cite{cf}), one obtains the following boundary extension of the formula established in \cite[Proposition 2.3]{bbmn}:
\begin{eqnarray}\label{kinetic_prod_1}
P_{F,\rho(.,.)} & = & \int_0^K\int\int_{(a,b)\times(0,1)} F''(t,x,v)m(v;dt,dx)dv\\
& + & \int_a^b[G(t,0,\rho(t,0^+))-G(t,1,\rho(t,1^-))]dt\label{kinetic_prod_2}
\end{eqnarray} 
Inequality $\geq$ in \eqref{vari_2} then follows from $F''(t,x,v)\leq h''(v)$, $G(t,0,.)\in\Phi_h(\rho_l)$, $G(t,1,.)\in\Phi_h(\rho_r)$ and Lemma \ref{lemma_vari}.
Assume now that the r.h.s. of \eqref{vari_2} is finite. We first show that $\rho(.,.)$ must be a weak solution of \eqref{conservation_law}. Assuming the contrary, there exists
$\varphi\in C^\infty_K((a,b)\times(0,1))$ such that $P_{F,\rho}\neq 0$,  with $F(t,x,\rho)=C\varphi(t,x)\rho$. Then the r.h.s. of \eqref{vari_2} dominates
the supremum over entropy samplers $C F(t,x,\rho)$, where $C\in\R$, which is $+\infty$.
Now, fix a $C^2$ entropy $\eta$ such that $\eta''\leq h''$. Let  $F(t,x,\rho)=\varphi(t,x)\eta(\rho)$, where $\varphi$ varies in $C^\infty_K((a,b)\times(0,1))$, $0\leq\varphi\leq 1$. Then $P_{F,\rho}=\mu_\eta[\rho(.,.)](\varphi)$ remains bounded over $\varphi$, i.e. $\mu_\eta[\rho(.,.)]\in\overline{M}((a,b)\times(0,1))$. Since $h''\geq c>0$, this is true for any $C^2$ entropy $\eta$. Thus $\rho(.,.)\in\emsol((a,b)\times(0,1))$.\\ \\
{\em Step two}. Assuming $\rho(.,.)\in\emsol((a,b)\times(0,1))$, we show that the r.h.s. of \eqref{vari_2} dominates the l.h.s. Let $\varepsilon>0$ and $\alpha_\varepsilon:[0,1]\to[0,1]$ be a smooth cutoff function such that $\alpha_\varepsilon(x)=1$ for $x\in[0,\varepsilon/2]$ and
$\alpha_\varepsilon(x)=0$ for $x\in[\varepsilon,1]$. Set $\beta_\varepsilon(x)=\alpha_\varepsilon(1-x)$ and $\gamma_\varepsilon=1-\alpha_\varepsilon-\beta_\varepsilon$.
We define a regularized (w.r.t. $\rho$) version of $\eta^l_{\rho_r,\rho}$ in \eqref{optimal_ent_1} by convolution
\be\label{optimal_ent_reg}
\bar{\eta}^{l,\varepsilon}_{\rho_l,\rho}(u)=\int\eta^l_{\rho_l,\rho-r}(u)\chi(r/\varepsilon)dr
\ee
where $\chi$ is a standard convolution kernel. To make the above meaningful we naturally extend $\eta^l_{\rho_l,\rho}$ for $\rho\not\in[0,K]$ by setting
$\eta^l_{\rho_l,\rho}=\eta^l_{\rho_l,K}$ for $\rho>K$, $\eta^l_{\rho_l,\rho}=\eta^l_{\rho_l,0}$ for $\rho<0$.
Following  \eqref{optimal_ent_1}--\eqref{optimal_ent_2} and \eqref{boundarycost_1}--\eqref{boundarycost_2}, $\bar{\eta}^{l,\varepsilon}_{\rho_l,\rho}(u)$ and
$\bar{\eta}^{r,\varepsilon}_{\rho_r,\rho}(u)$ are $H$-Lipschitz functions of $u$ with $H=h'(K)-h'(0)$, and continuous (hence uniformly continuous) functions of $(\rho,u)$.
Thus convergences 
\be\label{uniform_eta}\bar{\eta}^{l,\varepsilon}_{\rho_l,\rho}(u)\to\eta^l_{\rho_l,\rho}(u),\quad
\bar{\eta}^{r,\varepsilon}_{\rho_r,\rho}(u)\to\eta^r_{\rho_r,\rho}(u)\ee
are uniform with respect to $(\rho,u)\in[0,K]^2$.
We also regularize boundary traces $\rho(t,0^+)$ and $\rho(t,1^-)$ by convolution, setting 
\begin{eqnarray}\label{trace_reg_1}
\bar{\rho}^\varepsilon(t,0^+) & := & \int\rho(t-s,0^+)\chi(s/\varepsilon)ds\\
\bar{\rho}^\varepsilon(t,1^-) & := & \int\rho(t-s,1^-)\chi(s/\varepsilon)ds
\label{trace_reg_2}
\end{eqnarray}
We may view the mapping $v\mapsto m(v;dt,dx)$ as a measure $M(dt,dx,dv)=m(v;dt,dx)dv$ on $[0,K]\times(a,b)\times(0,1)$.
Then $M^+(dt,dx,dv)=m^+(v;dt,dx)dv$.
Let $\psi^\varepsilon\in C^\infty_K((a,b)\times(\varepsilon,1-\varepsilon)\times[0,K])$, $0\leq\psi^\varepsilon\leq 1$, such that
\be\label{almost_ppart}
\int_0^K\int_{(a,b)\times(0,1)}\psi^\varepsilon(t,x,v)m(v;dt,dx)dv\geq 
\int_0^K\int_{(a,b)\times(\varepsilon,1-\varepsilon)}m^+(v;dt,dx)dv-\varepsilon
\ee
We define $\tilde{F}^\varepsilon(t,x,.)$ as the unique entropy (given by \eqref{decomp_relative}) $\eta(.)\in\Phi(\rho_0)$ such that $\eta''(v)=h''(v)\psi^\varepsilon(t,x,v)$, where $\rho_0\in[0,K]$ is arbitrarily chosen (but independent of $(t,x))$.
Let $\delta>0$. 
We apply \eqref{kinetic_prod_1} to the entropy sampler
\begin{eqnarray}\label{approx_sampler_1}
F^{\varepsilon,\delta}(t,x,\rho)& := & \alpha_{\varepsilon}(x)\gamma_{\varepsilon}\left(\frac{t-a}{b-a}\right)\bar{\eta}^{l,\delta}_{\rho_l,\bar{\rho}_{\delta}(t,0^+)}(\rho)\\
& + & \beta_{\varepsilon}(x)\gamma_{\varepsilon}\left(\frac{t-a}{b-a}\right)\bar{\eta}^{r,\delta}_{\rho_l,\bar{\rho}_{\delta}(t,1^-)}(\rho)\label{approx_sampler_1}\\
& + & \gamma_{\varepsilon}(x)\tilde{F}^{\varepsilon}(t,x,\rho)
\label{approx_sampler_3}
\end{eqnarray}
We use \eqref{almost_ppart}, the triangle inequality
\begin{eqnarray*}
|\bar{\eta}^{l,\delta}_{\rho_l,\bar{\rho}}(\rho^+)-i^l(\rho^+,\rho_l)| & \leq & 
|\bar{\eta}^{l,\delta}_{\rho_l,\bar{\rho}}(\rho^+)-\bar{\eta}^{l,\delta}_{\rho_l,\bar{\rho}}(\bar{\rho})|\label{triangle_1}\\
& + & |\bar{\eta}^{l,\delta}_{\rho_l,\bar{\rho}}(\bar{\rho})-i^l(\bar{\rho},\rho_l)|\label{triangle_2}\\
& + & |i^l(\bar{\rho},\rho_l)-i^l(\rho^+,\rho_l)|\label{triangle_3}
\end{eqnarray*}
(note that $i^l(\bar{\rho},\rho_l)=\bar{\eta}^{l}_{\rho_l,\bar{\rho}}(\bar{\rho}))$ with $\bar{\rho}=\bar{\rho}^{\delta}(t,0^+)$, $\rho^+=\rho(t,0^+)$, and
the similar decomposition with $i^r(\rho^-,\rho_r)$, $\rho^-=\rho(t,1^-)$ and $\bar{\rho}=\bar{\rho}^{\delta}(t,1^-)$. From these we obtain that, in the limit
$\lim_{\varepsilon\to 0}\lim_{\delta\to 0}$, $P_{F^{\varepsilon,\delta},\rho(.,.)}$ converges to the l.h.s. of \eqref{vari_2}.
\end{proof}
\section{Explicit minimizers for uniform $\rho(.)$ when $\rho_l<\rho_r$}
\label{section_explicit_1}
In this and the next section, explicit minimizers are computed by interaction of Riemann waves, see e.g. \cite{lp} for details on such computations.
Let $\rho_c$ be given by \eqref{def_critical}.
We consider $\rho(.)\equiv\rho$, for some constant
$\rho\in[0,K]$.\\ \\
(1) Assume $\rho_l<\rho_r$ and $f(\rho_l)\leq f(\rho_r)$ (thus $\rho_l<\rho^*$). Then:\\ \\
(a) If $\rho\leq\rho_l$, then $y=1$, $\tilde{\theta}^{\tilde{y}}=0$,
$\tau^y=1/f'(\rho_l)$, $\rho_s(.)\equiv\rho_l$,
\be \label{optimal_tilde_1a} \tilde{\rho}(t,x)=\left\{
\ba{lll}
\rho_l & \mbox{if} & x<t f'(\rho)\\ \\
(f')^{-1}(x/t) & \mbox{if} &
\min(tf'(\rho_l),1)<x<\min(tf'(\rho),1)\\ \\
\rho & \mbox{if} & \min(tf'(\rho),1)<x
\ea
\right. \ee
(b) If $\rho_l\leq \rho\leq \rho^*$ and $\rho<\rho_c$, then $y=1$,
$\tilde{\theta}^{\tilde{y}}=0$, $\tau^y=1/v(\rho,\rho_l)$, $\rho_s(.)\equiv\rho_l$,
\be \label{optimal_tilde_1b} \tilde{\rho}(t,.)=\rho_l{\bf
1}_{(0,tv(\rho,\rho_l))}+\rho{\bf 1}_{(tv(\rho,\rho_l),1)},\quad
t\in[0,\tau^y]\ee
%
%
with $v(\rho,\rho_l)>0$.\\ \\
(c) If $\rho^*<\rho<\rho_c$, then $y=1$, $\tilde{\theta}^{\tilde{y}}=0$, $\tau^y$
will be defined below, $\rho_s(.)\equiv\rho_l$,
\be \label{optimal_tilde_1c_1}
\tilde{\rho}(t,x)=\left\{
\ba{lll} \rho_l & \mbox{if} & 0<x<tv(\rho,\rho_l)\\
\\
\rho & \mbox{if} & tv(\rho,\rho_l)<x<1+tf'(\rho)\\
\\
(f')^{-1}((x-1)/t) & \mbox{if} & 1+tf'(\rho)<x<1
\ea
\right. \ee
for $t<[v(\rho_l,\rho)-f'(\rho)]^{-1}:=t_1$,
\be \label{optimal_tilde_1c_2}
\tilde{\rho}(t,x)=\left\{
\ba{lll} \rho_l & \mbox{if} & 0<x<x_t\\ \\
(f')^{-1}((x-1)/t) & \mbox{if} & x_t<x<1
\ea
\right. \ee
for $t_1<t<\tau^y$, where $x_t$ is defined for $t\geq t_1$ by
\be \label{def_interface_c}
x_{t_1}=t_1 v(\rho,\rho_l)=1+t_1f'(\rho),\quad
\dot{x}_t=v\left[\rho_l,(f')^{-1}((x_t-1)/t)\right]>v(\rho,\rho_l)>0
\ee
and $\tau^y$ is the time at which $x_t=1$.\\ \\
(d) If $\rho_c<\rho<\rho_r$ and $f(\rho_l)=f(\rho_r)$, then $y=0$,
$\tilde{\theta}^{\tilde{y}}=0$, $\tau^y=-1/v(\rho,\rho_r)$, $\rho_s(.)\equiv\rho_r$,
\be \label{optimal_tilde_1d}
\tilde{\rho}(t,x)=\rho{\bf 1}_{(0,1+tv(\rho,\rho_r))}+\rho_r{\bf
1}_{(1+tv(\rho,\rho_r),1)}
\ee
with $v(\rho,\rho_r)<0$.\\ \\
(e) If $\rho_r<\rho$ and $f(\rho_l)=f(\rho_r)$, then $y=0$,
$\tilde{\theta}^{\tilde{y}}=0$, $\tau^y=-1/f'(\rho_r)$, $\rho_s(.)\equiv\rho_r$,
\be \label{optimal_tilde_1e}
\tilde{\rho}(t,x)=\left\{
\ba{lll}
\rho_r & \mbox{if} & \max(1+tf'(\rho_r),0)<x<1 \\ \\
(f')^{-1}((x-1)/t) & \mbox{if} &
\max(1+tf'(\rho),0)<x<\max(1+tf'(\rho_r),0)\\ \\
\rho & \mbox{if} & x<\max(1+tf'(\rho),0)
\ea
\right.
\ee
(f) If $\rho>\rho_c$ and $\rho_l<\rho_r\leq\rho^*$, then $y=0$,
$\tilde{\theta}^{\tilde{y}}=-1/v(\varphi(\rho_r),\varphi(\rho_l))$, $\tau^y$ is
defined below, $\rho_s(.)\equiv\rho_l$,
\be \label{optimal_tilde_1f_1}
\tilde{\rho}(t,x)=\left\{
\ba{lll}
(f')^{-1}((x-1)/t) & \mbox{if} & 1+tf'(\varphi(\rho_r))<x<1\\ \\
\varphi(\rho_r) & \mbox{if} &
1+tv(\varphi(\rho_r),\varphi(\rho_l))<x<1+tf'(\varphi(\rho_r))\\
\\
\varphi(\rho_l) & \mbox{if} &
[1+tv(\rho,\varphi(\rho_l))]^+<x<[1+tv(\varphi(\rho_r),\varphi(\rho_l))]^+\\
\\
\rho & \mbox{if} & 0<x<[1+tv(\rho,\varphi(\rho_l))]^+
\ea
\right.
\ee
for $t<\tilde{\theta}^{\tilde{y}}$,
\be \label{optimal_tilde_1f_2}
\tilde{\rho}(t,x)=\left\{
\ba{lll}
\rho_l & \mbox{if} & x<(t-\tilde{\theta}^{\tilde{y}})v(\rho_l,\varphi(\rho_r))\\ \\
\rho & \mbox{if} &
tv(\rho_l,\varphi(\rho_r))<x<1+tf'(\varphi(\rho_r))\\ \\
(f')^{-1}((x-1)/t) & \mbox{if} & 1+tf'(\varphi(\rho_r))<x<1
\ea
\right.
\ee
for
$\dsp\tilde{\theta}^{\tilde{y}}<t<t_1:=\frac{1+\tilde{\theta}^{\tilde{y}}v(\rho_l,\varphi(\rho_r))}{v(\rho_l,\varphi(\rho_r))-f'(\varphi(\rho_r))}$,
\be \label{optimal_tilde_1f_3}
\tilde{\rho}(t,x)=\left\{
\ba{lll}
\rho_l & \mbox{if} & 0<x<x_t \\ \\
(f')^{-1}((x_t-1)/t) & \mbox{if} & x_t<x<1
\ea
\right.
\ee
for $t_1<t<\tau^y$,
where $x_t$ is defined for $t\geq t_1$ by
\be \label{def_interface_f}
x_{t_1}=1+t_1 f'(\varphi(\rho_r)),\quad
\dot{x}_t=v\left[\rho_l,(f')^{-1}((x_t-1)/t)\right]>v(\rho_l,\varphi(\rho_r))>0
\ee
and $\tau^y$ is the time at which $x_t=1$.\\ \\
(g) If $\rho>\rho_c$ and $\rho_l<\rho^*<\rho_r$, then $y=0$,
$\tilde{\theta}^{\tilde{y}}=-1/v(\varphi(\rho_r),\varphi(\rho_l))$,
$\tau^y=\tilde{\theta}^{\tilde{y}}+1/v(\rho_l,\varphi(\rho_r))$,
$\rho_s(.)\equiv\rho_l$,
\be \label{optimal_tilde_1g_1}
\tilde{\rho}(t,x)=\left\{
\ba{lll}
\varphi(\rho_r) & \mbox{if} &
1+tv(\varphi(\rho_r),\varphi(\rho_l))<x<1\\
\\
\varphi(\rho_l) & \mbox{if} &
[1+tv(\rho,\varphi(\rho_l))]^+<x<[1+tv(\varphi(\rho_r),\varphi(\rho_l))]^+\\
\\
\rho & \mbox{if} & 0<x<[1+tv(\rho,\varphi(\rho_l))]^+
\ea
\right.
\ee
for $t<\tilde{\theta}^{\tilde{y}}$,
\be \label{optimal_tilde_1g_2}
\tilde{\rho}(t,x)=\rho_l{\bf
1}_{(0,tv(\rho_l,\varphi(\rho_r))}(x)+\varphi(\rho_r){\bf
1}_{(tv(\rho_l,\varphi(\rho_r),1)}(x) \ee
for $\tilde{\theta}^{\tilde{y}}<t<\tau^y$.\\ \\
(h) If $\rho=\rho_c<\rho^*$, and $f(\rho_l)=f(\rho_r)$, then
$y\in[0,1]$ is arbitrary, $\tilde{\theta}^{\tilde{y}}=+\infty$, $\tau^y$ is defined
below, $\rho_s(.)=\rho_s^y(.)$,
\be \label{optimal_tilde_1h_1}
\tilde{\rho}(t,x)=\left\{
\ba{lll}
\rho_c & \mbox{if} & \min(1-y+tv(\rho_l,\rho_c),1)<x<1\\ \\
\rho_l & \mbox{if} & 1-y<x<\min(1-y+tv(\rho_l,\rho_c),1)\\ \\
\rho_r & \mbox{if} & 1-y+tv(\rho_c,\rho_r)<x<1-y\\ \\
\rho_c & \mbox{if} & tf'(\rho_c)<x<1-y+tv(\rho_c,\rho_r)\\ \\
(f')^{-1}(x/t) & \mbox{if} & 0<x<tf'(\rho_c)
\ea
\right.
\ee
for $0<t<t_1:=(1-y)(f'(\rho_c)-v(\rho_c,\rho_r))$,
\be \label{optimal_tilde_1h_1}
\tilde{\rho}(t,x)=\left\{
\ba{lll}
\rho_c & \mbox{if} & \min(1-y+tv(\rho_l,\rho_c),1)<x<1\\ \\
\rho_l & \mbox{if} & 1-y<x<\min(1-y+tv(\rho_l,\rho_c),1)\\ \\
\rho_r & \mbox{if} & x_t<x<1-y\\ \\
(f')^{-1}(x/t) & \mbox{if} & 0<x<x_t
\ea
\right.
\ee
for $t_1<t<\tau^y$, where $\tau^y=\max(t_2,y/v(\rho_l,\rho_c))$,
$t_2$ is the time at which $x_t=0$, and $x_t$ is defined for
$t\geq t_1$ by
\be \label{def_interface_h}
x_{t_1}=t_1f'(\rho_c),\quad
\dot{x}_t=v\left[(f')^{-1}(x_t/t),\rho_r\right]<v(\rho_c,\rho_r)<0
\ee
(i) If $\rho=\rho_c>\rho^*$, and $f(\rho_l)=f(\rho_r)$, then
$y\in[0,1]$ is arbitrary, $\tilde{\theta}^{\tilde{y}}=+\infty$, $\tau^y$ is defined
below, $\rho_s(.)=\rho_s^y(.)$,
\be \label{optimal_tilde_1i_1}
\tilde{\rho}(t,x)=\left\{
\ba{lll}
(f')^{-1}((x-1)/t) & \mbox{if} & 1+tf'(\rho_c)<x<1\\ \\
\rho_c & \mbox{if} & 1-y+tv(\rho_l,\rho_c)<x<1+tf'(\rho_c)\\ \\
\rho_l & \mbox{if} & 1-y<x<1-y+tv(\rho_l,\rho_c)\\ \\
\rho_r & \mbox{if} & \max(1-y+tv(\rho_c,\rho_r),0)<x<1-y\\ \\
\rho_c & \mbox{if} & 0<x<\max(1-y+tv(\rho_r,\rho_c),0)
\ea
\right.
\ee
for $0<t<t_1:=y/(v(\rho_l,\rho_c)-f'(\rho_c))$,
\be \label{optimal_tilde_1i_2}
\tilde{\rho}(t,x)=\left\{
\ba{lll}
(f')^{-1}((x-1)/t) & \mbox{if} & x_t<x<1\\ \\
\rho_l & \mbox{if} & 1-y<x<x_t\\ \\
\rho_r & \mbox{if} & \max(1-y+tv(\rho_c,\rho_r),0)<x<1-y\\ \\
\rho_c & \mbox{if} & 0<x<\max(1-y+tv(\rho_r,\rho_c),0)
\ea
\right.
\ee
for $t_1<t<\tau^y$, where
$\tau^y=\max(t_2,(y-1)/v(\rho_r,\rho_c))$, $t_2$ is the time at
which $x_t=1$, and $x_t$ is defined for $t\geq t_1$ by
\be \label{def_interface_i}
x_{t_1}=1+t_1f'(\rho_c),\quad
\dot{x}_t=v\left[(f')^{-1}((x_t-1)/t),\rho_l\right]>v(\rho_c,\rho_l)>0
\ee
(j) If $\rho=\rho_c=\rho^*$, and $f(\rho_l)=f(\rho_r)$, then
$y\in[0,1]$ is arbitrary, $\tilde{\theta}^{\tilde{y}}=+\infty$,
$\tau^y=\max(y/v(\rho_l,\rho_c),(y-1)/v(\rho_r,\rho_c))$,
$\rho_s(.)=\rho_s^y(.)$,
\be \label{optimal_tilde_1j}
\tilde{\rho}(t,x)=\left\{
\ba{lll}
\rho_c & \mbox{if} & \min(1-y+tv(\rho_l,\rho_c),1)<x<1\\ \\
\rho_l & \mbox{if} & 1-y<x<\min(1-y+tv(\rho_l,\rho_c),1)\\ \\
\rho_r & \mbox{if} & \max(1-y+tv(\rho_c,\rho_r),0)<x<1-y\\ \\
\rho_c & \mbox{if} & 0<x<\max(1-y+tv(\rho_r,\rho_c),0)
\ea
\right.
\ee
(k) If $f(\rho_l)<f(\rho_r)$ and $\rho=\rho_c\geq\rho^*$, then
$y\in[0,1]$ is arbitrary,
$\tilde{\theta}^{\tilde{y}}=(1-y)/v(\varphi(\rho_r),\varphi(\rho_l))$, $\tau^y$ is
defined below, $\rho_s(.)\equiv\rho_l$,
\be \label{optimal_tilde_1k_1}
\tilde{\rho}(t,x)=\left\{
\ba{lll}
\rho_c & \mbox{if} & 0<x<\max(1-y+tv(\rho_c,\varphi(\rho_l)),0)\\
\\
\varphi(\rho_l) & \mbox{if} &
\max(1-y+tv(\rho_c,\varphi(\rho_l)),0)<x<1-y+tv(\varphi(\rho_l),\varphi(\rho_r))\\
\\
\varphi(\rho_r) & \mbox{if} &
1-y+tv(\varphi(\rho_l),\varphi(\rho_r))<x<1-y+tv(\varphi(\rho_r),\rho_c)\\
\\
\rho_c & \mbox{if} &
1-y+tv(\varphi(\rho_r),\rho_c)<x<1+tf'(\rho_c)\\ \\
(f')^{-1}((x-1)/t) & \mbox{if} & 1+tf'(\rho_c)<x<1
\ea
\right.
\ee
for $t<t_1:=\min(\tilde{\theta}^{\tilde{y}},t'_1)$, where
$$
%
t'_1:= y/[v(\varphi(\rho_r),\rho_c)-f'(\rho_c)]
$$
\be \label{optimal_tilde_1k_2}
\tilde{\rho}(t,x)=\left\{
\ba{lll}
\rho_c & \mbox{if} & 0<x<\max(1-y+tv(\rho_c,\varphi(\rho_l)),0)\\
\\
\varphi(\rho_l) & \mbox{if} &
1-y+tv(\rho_c,\varphi(\rho_l))<x<\max(1-y+tv(\varphi(\rho_l),\varphi(\rho_r)),0)\\
\\
\varphi(\rho_r) & \mbox{if} &
1-y+tv(\varphi(\rho_l),\varphi(\rho_r))<x<x_t\\
\\
(f')^{-1}((x-1)/t) & \mbox{if} & x_t<x<1
\ea
\right.
\ee
if $t'_1<t<\tilde{\theta}^{\tilde{y}}$, where $x_t$ is defined by
\be \label{def_interface_k2}
x_{t'_1}=1+t'_1f'(\rho_c),\quad
\dot{x}_t=v\left[(f')^{-1}((x_t-1)/t),\varphi(\rho_r)\right]
\ee
with
\be \label{speed_interface_k2}
v(\varphi(\rho_r),\varphi(\rho_l))<v\left[(f')^{-1}((x_t-1)/t),\varphi(\rho_r)\right]<f'(\varphi(\rho_r))<0
\ee
and
$$
(f')^{-1}((x_t-1)/t)\in(\varphi(\rho_r),\rho_c)]
$$
\be \label{optimal_tilde_1k_3}
\tilde{\rho}(t,x)=\left\{
\ba{lll}
\rho_c & \mbox{if} & 0<x<tv(\rho_l,\varphi(\rho_r))\\
\\
\varphi(\rho_r) & \mbox{if} & tv(\rho_l,\varphi(\rho_r))<x<x_t\\
\\
(f')^{-1}((x-1)/t) & \mbox{if} & x_t<x<1
\ea
\right.
\ee
if $t'_1<\tilde{\theta}^{\tilde{y}}<t<t_2$, where $t_2$ is the time $t$ at which
$x_t=tv(\rho_l,\varphi(\rho_r))$,
\be \label{optimal_tilde_1k_4}
\tilde{\rho}(t,x)=\left\{
\ba{lll}
\rho_l & \mbox{if} & 0<x<x'_t\\
\\
(f')^{-1}((x-1)/t) & \mbox{if} & x'_t<x<1
\ea
\right.
\ee
for $t'_1<\tilde{\theta}^{\tilde{y}}<t_2<t<\tau^y$, where $x'_t$ is defined by
\be \label{def_interface_k4}
x'_{t_2}=x_{t_2},\quad
\dot{x'}_t=v\left[(f')^{-1}((x'_t-1)/t),\rho_l\right]>v(\rho_l,\rho_c)>0
\ee
and $\tau^y$ is the time $t$ at which $x'_t=1$,
\be \label{optimal_tilde_1k_5}
\tilde{\rho}(t,x)=\left\{
\ba{lll}
\rho_l & \mbox{if} & 0<x<tv(\rho_l,\varphi(\rho_r))\\
\\
\\
\varphi(\rho_r) & \mbox{if} &
tv(\rho_l,\varphi(\rho_r))<x<1-y+tv(\varphi(\rho_r),\rho_c)\\
\\
\rho_c & \mbox{if} &
1-y+tv(\varphi(\rho_r),\rho_c)<x<1+tf'(\rho_c)\\ \\
(f')^{-1}((x-1)/t) & \mbox{if} & 1+tf'(\rho_c)<x<1
\ea
\right.
\ee
for $\tilde{\theta}^{\tilde{y}}<t<\min(t'_1,t_2)$, where $t_2$ is the time $t$ at
which $tv(\rho_l,\varphi(\rho_r))=1-y+tv(\varphi(\rho_r),\rho_c)$,
\be \label{optimal_tilde_1k_6}
\tilde{\rho}(t,x)=\left\{
\ba{lll}
\rho_l & \mbox{if} & 0<x<(t-t_2)v(\rho_l,\rho_c)+t_2v(\rho_l,\varphi(\rho_r))\\
\\
\\
\rho_c & \mbox{if} &
(t-t_2)v(\rho_l,\rho_c)+t_2v(\rho_l,\varphi(\rho_r))<x<1+tf'(\rho_c)\\
\\
(f')^{-1}((x-1)/t) & \mbox{if} & 1+tf'(\rho_c)<x<1
\ea
\right.
\ee
if $\tilde{\theta}^{\tilde{y}}<t_2<t<t'_2$, where $t'_2<t'_1$ is the time $t$ at
which
$$
1+tf'(\rho_c)=(t-t_2)v(\rho_l,\rho_c)+t_2v(\rho_l,\varphi(\rho_r))
$$
\be \label{optimal_tilde_1k_7}
\tilde{\rho}(t,x)=\left\{
\ba{lll}
\rho_l & \mbox{if} & 0<x<x''_t\\
\\
(f')^{-1}((x-1)/t) & \mbox{if} & x''_t<x<1\\
\ea
\right.
\ee
for $t'_2<t<\tau^y$, where $x''_t$ is defined by
\be \label{def_interface_k7}
x''_{t'_2}=1+t'_2f'(\rho_c),\quad
\dot{x''}_t=v\left[(f')^{-1}((x''_t-1)/t),\rho_l\right]>v(\rho_l,\rho_c)>0
\ee
and $\tau^y$ is the time $t$ at which $x''_t=1$,
\be \label{optimal_tilde_1k_8}
\tilde{\rho}(t,x)=\left\{
\ba{lll}
\rho_l & \mbox{if} & 0<x<tv(\rho_l,\varphi(\rho_r))\\
\\
\varphi(\rho_r) & \mbox{if} & tv(\rho_l,\varphi(\rho_r))<x<x_t\\
\\
(f')^{-1}((x-1)/t) & \mbox{if} & x_t<x<1\\
\ea
\right.
\ee
if $\tilde{\theta}^{\tilde{y}}<t'_1<t<t''_2$, where $x_t$ is defined by
\eqref{def_interface_k2}, and $t''_2$ is the time $t$ at which
$x_t=tv(\rho_l,\varphi(\rho_r))$,
\be \label{optimal_tilde_1k_9}
\tilde{\rho}(t,x)=\left\{
\ba{lll}
\rho_l & \mbox{if} & 0<x<x^{(3)}_t\\
\\
(f')^{-1}((x-1)/t) & \mbox{if} & x^{(3)}_t<x<1\\
\ea
\right.
\ee
for $t''_2<t<\tau^y$,
where $x^{(3)}_t$ is defined by
\be \label{def_interface_k7}
x^{(3)}_{t''_2}=x_{t''_2},\quad
\dot{x^{(3)}}_t=v\left[(f')^{-1}((x^{(3)}_t-1)/t),\rho_l\right]>v(\rho_l,\rho_c)>0
\ee
and $\tau^y$ is the time $t$ at which $x''_t=1$.\\ \\
(l) If $f(\rho_l)<f(\rho_r)$ and $\rho=\rho_c<\rho^*$, then
$y\in[0,1]$ is arbitrary,
$\tilde{\theta}^{\tilde{y}}=(y-1)/v(\varphi(\rho_r),\varphi(\rho_l))$, $\tau^y$ is
defined below, $\rho_s(.)\equiv\rho_l$,
\be \label{optimal_tilde_1l_1}
\tilde{\rho}(t,x)=\left\{
\ba{lll}
(f')^{-1}(x/t) & \mbox{if} & 0<x<tf'(\rho_c)\\ \\
\rho_ c & \mbox{if} &
tf'(\rho_c)<x<1-y+tv(\rho_c,\varphi(\rho_l))\\ \\
\varphi(\rho_l) & \mbox{if} &
1-y+tv(\rho_c,\varphi(\rho_l))<x<1-y+tv(\varphi(\rho_r),\varphi(\rho_l))\\
\\
\varphi(\rho_r) & \mbox{if} &
1-y+tv(\varphi(\rho_r),\varphi(\rho_l))<x<\min(1,1-y+tv(\varphi(\rho_r),\rho_c))\\
\\
\rho_c & \mbox{if} & \min(1,1-y+tv(\varphi(\rho_r),\rho_c))<x<1
\ea
\right.
\ee
for $t<t_1:=(1-y)/[f'(\rho_c)-v(\rho_c,\varphi(\rho_l))]$,
\be \label{optimal_tilde_1l_2}
\tilde{\rho}(t,x)=\left\{
\ba{lll}
(f')^{-1}(x/t) & \mbox{if} & 0<x<x_t\\ \\
\varphi(\rho_l) & \mbox{if} &
x_t<x<1-y+tv(\varphi(\rho_r),\varphi(\rho_l))\\
\\
\varphi(\rho_r) & \mbox{if} &
1-y+tv(\varphi(\rho_r),\varphi(\rho_l))<x<\min(1,1-y+tv(\varphi(\rho_r),\rho_c))\\
\\
\rho_c & \mbox{if} & \min(1,1-y+tv(\varphi(\rho_r),\rho_c))<x<1
\ea
\right.
\ee
for $t_1<t<t_2$, where $t_2$ is the time $t$ at which $x_t=0$,
with $x_t$ defined by
\be \label{def_interface_l2}
x_{t_1}=t_1 f'(\rho_c),\quad
\dot{x}_t=v[\varphi(\rho_l),(f')^{-1}(x_t/t)]<v(\rho_c,\varphi(\rho_l))<v(\varphi(\rho_r),\varphi(\rho_l))
\ee
\be \label{optimal_tilde_1l_3}
\tilde{\rho}(t,x)=\left\{
\ba{lll}
\varphi(\rho_l) & \mbox{if} &
x_t<x<1-y+tv(\varphi(\rho_r),\varphi(\rho_l))\\
\\
\varphi(\rho_r) & \mbox{if} &
1-y+tv(\varphi(\rho_r),\varphi(\rho_l))<x<\min(1,1-y+tv(\varphi(\rho_r),\rho_c))\\
\\
\rho_c & \mbox{if} & \min(1,1-y+tv(\varphi(\rho_r),\rho_c))<x<1
\ea
\right.
\ee
for $t_2<t<\tilde{\theta}^{\tilde{y}}$,
\be \label{optimal_tilde_1l_4}
\tilde{\rho}(t,x)=\left\{
\ba{lll}
\rho_l & \mbox{if} & 0<x<(t-\tilde{\theta}^{\tilde{y}})v(\rho_l,\varphi(\rho_r))\\
\\
\varphi(\rho_r) & \mbox{if} &
1-y+tv(\varphi(\rho_r),\varphi(\rho_l))<x<\min(1,1-y+tv(\varphi(\rho_r),\rho_c))\\
\\
\rho_c & \mbox{if} & \min(1,1-y+tv(\varphi(\rho_r),\rho_c))<x<1
\ea
\right.
\ee
if
$$
y\leq\frac{v(\rho_l,\varphi(\rho_r))^{-1}-v(\varphi(\rho_l),\varphi(\rho_r))^{-1}}
{v(\varphi(\rho_r),\rho_c))^{-1}-v(\varphi(\rho_l),\varphi(\rho_r))^{-1}}
$$
and $\tilde{\theta}^{\tilde{y}}<t<\tau^y$, with
$\tau^y=\tilde{\theta}^{\tilde{y}}+1/v(\rho_l,\varphi(\rho_r))$,
\be \label{optimal_tilde_1l_5}
\tilde{\rho}(t,x)=\left\{
\ba{lll}
\rho_l & \mbox{if} & 0<x<(t-\tilde{\theta}^{\tilde{y}})v(\rho_l,\varphi(\rho_r))\\
\\
\varphi(\rho_r) & \mbox{if} &
(t-\tilde{\theta}^{\tilde{y}})v(\rho_l,\varphi(\rho_r))<x<1-y+tv(\varphi(\rho_r),\rho_c)\\
\\
\rho_c & \mbox{if} & 1-y+tv(\varphi(\rho_r),\rho_c)<x<1
\ea
\right.
\ee
if
$$
y>\frac{v(\rho_l,\varphi(\rho_r))^{-1}-v(\varphi(\rho_l),\varphi(\rho_r))^{-1}}
{v(\varphi(\rho_r),\rho_c))^{-1}-v(\varphi(\rho_l),\varphi(\rho_r))^{-1}}
$$
and
$$
\tilde{\theta}^{\tilde{y}}<t<t_3:=(1-y)\frac{v(\varphi(\rho_r,\varphi(\rho_l))-v(\rho_l,\varphi(\rho_r))}
{[v(\varphi(\rho_r,\varphi(\rho_l))][v(\varphi(\rho_l,\varphi(\rho_r))-v(\varphi(\rho_r),c)]}
$$,
\be \label{optimal_tilde_1l_6}
\tilde{\rho}(t,x)=\left\{
\ba{lll}
\rho_l & \mbox{if} & 0<x<t_3v(\rho_l,\varphi(\rho_r))+(t-t_3)v(\rho_l,\rho_c)\\
\\
\rho_c & \mbox{if} &
t_3v(\rho_l,\varphi(\rho_r))+(t-t_3)v(\rho_l,\rho_c)<x<1
\ea
\right.
\ee
if
$$
y>\frac{v(\rho_l,\varphi(\rho_r))^{-1}-v(\varphi(\rho_l),\varphi(\rho_r))^{-1}}
{v(\varphi(\rho_r),\rho_c))^{-1}-v(\varphi(\rho_l),\varphi(\rho_r))^{-1}}
$$
and $t_3<t<\tau^y$, where $\tau^t$ is the time $t$ at which
$$t_3v(\rho_l,\varphi(\rho_r))+(t-t_3)v(\rho_l,\rho_c)=1$$
\section{Explicit minimizers for uniform $\rho(.)$ when $\rho_l>\rho_r$}
\label{section_explicit_2}
Let us assume $\rho(.)=\rho$. Then \eqref{optimal_F} yields
\be \label{envelope_uniform}
F_{\rho(.)}\equiv\min[\rho_l,\max(\rho_r,\varphi(\rho))],\quad
\tilde{G}(0,.)\equiv\max[\varphi(\rho_l),\min(\varphi(\rho_r),\rho))]:=\rho'
\ee
\subsection{The case $\rho_r\leq\rho_l\leq\rho^*$}
In this case, we have
\be \label{tildeg_1}
\tilde{G}(t,x)=\left\{
\ba{lll}
\rho' & \mbox{if} & x<\max(1+tf'(\rho'),0)\\ \\
(f')^{-1}((x-1)/t) & \mbox{if} &
\max(1+tf'(\rho'),0)<x<\max(1+tf'(\varphi(\rho_l)),0)\\ \\
\varphi(\rho_l) & \mbox{if} & \max(1+tf'(\varphi(\rho_l)),0)<x<1
\ea
\right.
\ee
Thus
\be \label{tildeg_left}
\varphi(\tilde{G}(t,0^+))=\left\{
\ba{lll}
\varphi(\rho') & \mbox{if} & t<t_1:=-1/f'(\rho')\\ \\
\varphi[(f')^{-1}(-1/t)] & \mbox{if} &
t_1<t<t_2:=-1/f'(\varphi(\rho_l))\\ \\
\rho_l & \mbox{if} & t>t_2
\ea
\right.
\ee
\be \label{tildeg_right}
\varphi(\tilde{G}(t,1^-))=\rho_l
\ee
From this we obtain:\\ \\
{\em First case: $\varphi(\rho_l)\leq\rho$.} Then $\rho\geq\rho'$,
and
\be \label{tilderho_11}
\tilde{\rho}(t,x)=\left\{
\ba{lll}
\rho & \mbox{if} & 0<x<\max(0,1+tf'(\rho))\\ \\
(f')^{-1}((x-1)/t) & \mbox{if} & \max(0,1+tf'(\rho))<x<1
\ea
\right.
\ee
for $t<t_2$,
\be \label{tilderho_12}
\tilde{\rho}(t,x)=\left\{
\ba{lll}
\rho_l & \mbox{if} & 0<x<x_t\\ \\
(f')^{-1}((x-1)/t) & \mbox{if} & x_t<x<1
\ea
\right. \ee
for $t>t_2$, where $x_t$ is given by
\be \label{def_min_x} x_{t_2}=0,\quad
\dot{x}_t=v[\rho_l,(f')^{-1}((x_t-1)/t)]>v[\rho_l,(f')^{-1}(-1/t)]>0\mbox{
for }t>t_2
\ee
so that $x_t$ reaches $1$ in finite time. Note that at time $t_2$,
the rarefaction wave issued from $1$ has value $\varphi(\rho_l)$
at $x=0$,
and value $<\varphi(\rho_l)$ ahead, which is why a shock starts propagating from $x=0$ towards $x=1$.\\ \\
{\em Second case: $\rho^*<\rho<\varphi(\rho_l)$.} Then
$\rho'=\varphi(\rho_l)$, $\varphi[\tilde{G}(t,0^+)]=\rho_l$ for
all $t>0$, and
\be \label{tilderho_13}
\tilde{\rho}(t,x)=\left\{
\ba{lll}
\rho_l & \mbox{if} & 0<x<tv(\rho_l,\rho)\\ \\
\rho & \mbox{if} & tv(\rho_l,\rho)<x<1+tf'(\rho)\\ \\
(f')^{-1}((x-1)/t) & \mbox{if} & 1+tf'(\rho)<x<1
\ea
\right.
\ee
for $t<t_3:=[v(\rho_l,\rho)-f'(\rho)]^{-1}$,
\be \label{tilderho_14}
\tilde{\rho}(t,x)=\left\{
\ba{lll}
\rho_l & \mbox{if} & 0<x<y_t\\ \\
(f')^{-1}((x-1)/t) & \mbox{if} & y_t<x<1
\ea
\right.
\ee
where $y_t$ is given by
\be \label{def_min_x2}
y_{t_3}=t_3 v(\rho_l,\rho),\quad
\dot{y}_t=v[\rho_l,(f')^{-1}((y_t-1)/t)]>v(\rho_l,\rho)>0
\ee
and thus reaches $1$ in finite time.\\ \\
{\em Third case:
$\rho\leq\rho^*\leq\varphi(\rho_l)\leq\varphi(\rho_r)$.} Then
again $\rho'=\varphi(\rho_l)$ and
$\varphi[\tilde{G}(t,0^+)]=\rho_l$,
\be \label{tilderho_15}
\tilde{\rho}(t,x)=\left\{
\ba{lll}
\rho_l & \mbox{if} & 0<x<\min[1,tv(\rho_l,\rho)]\\ \\
\rho & \mbox{if} & \min[1,tv(\rho_l,\rho)]<x<1
\ea
\right.
\ee
\subsection{The case $\rho_r<\rho^*<\rho_l$}
{\em First case. $\rho\geq\rho^*$.} Then $\rho'\geq\rho^*$,
\be \label{tildeg_21}
\tilde{G}(t,x)=\left\{
\ba{lll}
(f')^{-1}((x-1)/t) & \mbox{if} & \max[0,1+tf'(\rho')]<x<1 \\ \\
\rho' & \mbox{if} & 0<x<\max[0,1+tf'(\rho')]
\ea
\right.
\ee
and thus
\be \label{tildeg0_2}
\varphi[\tilde{G}(t,0^+)]=\left\{
\ba{lll}
\varphi(\rho') & \mbox{if} & t<t_1:=-1/f'(\rho')\\ \\
(f')^{-1}(-1/t) & \mbox{if} & t>t_1
\ea
\right.
\ee
\be \label{tildeg1_2}\varphi[\tilde{G}(t,1^-)]=\rho^* \ee
and
\be\label{tilderho_21}
\tilde{\rho}(t,x)=\left\{
\ba{lll}
(f')^{-1}((x-1)/t) & \mbox{if} & \max[0,1+tf'(\rho')]<x<1\\ \\
\rho' & \mbox{if} & 0<x<\max[0,1+tf'(\rho')]
\ea
\right.
\ee
{\em Second case. $\rho\leq\rho^*$.} Then $\rho'\leq\rho^*$,
\be \label{tildeg_22}
\tilde{G}(t,x)=\left\{
\ba{lll}
(f')^{-1}(x/t) & \mbox{if} & 0<x<\min[1,tf'(\rho')] \\ \\
\rho' & \mbox{if} & \min[1,tf'(\rho')]<x<1
\ea
\right.
\ee
and thus
\be \label{tildeg0_22}
\varphi[\tilde{G}(t,1^-)]=\left\{
\ba{lll}
\varphi(\rho') & \mbox{if} & t<t_1:=1/f'(\rho')\\ \\
(f')^{-1}(1/t) & \mbox{if} & t>t_1
\ea
\right.
\ee
\be \label{tildeg1_22}\varphi[\tilde{G}(t,0^+)]=\rho^* \ee
and
\be\label{tilderho_22}
\tilde{\rho}(t,x)=\left\{
\ba{lll}
(f')^{-1}(x/t) & \mbox{if} & 0<x<\min[1,tf'(\rho')]\\ \\
\rho' & \mbox{if} & \min[1,tf'(\rho')]<x<1
\ea
\right.
\ee
\subsection{The case $\rho^*<\rho_r\leq\rho_l$}
This case is symmetric to $\rho_r<\rho_l<\rho^*$, with the roles
of the boundaries exchanged. Here we have
\be \label{tildeg_3}
\tilde{G}(t,x)=\left\{
\ba{lll}
\rho' & \mbox{if} & \min(1,tf'(\rho'))<x<1\\ \\
(f')^{-1}((x-1)/t) & \mbox{if} &
\min[1,tf'(\varphi(\rho_r))]<x<\min[1,tf'(\rho')]\\ \\
\varphi(\rho_r) & \mbox{if} & 0<x<\min[1,tf'(\varphi(\rho_r))]
\ea
\right.
\ee
Thus
\be \label{tildeg_left_3}
\varphi(\tilde{G}(t,0^+))=\rho_r
\ee
\be \label{tildeg_right_3}
\varphi(\tilde{G}(t,1^-))=\left\{
\ba{lll}
\varphi(\rho') & \mbox{if} & t<t_1:=1/f'(\rho')\\ \\
\varphi[(f')^{-1}(1/t)] & \mbox{if} &
t_1<t<t_2:=1/f'(\varphi(\rho_r))\\ \\
\rho_r & \mbox{if} & t>t_2
\ea
\right.
\ee
From this we obtain:\\ \\
{\em First case: $\rho\leq\varphi(\rho_r)$.} Then $\rho\leq\rho'$,
and
\be \label{tilderho_31}
\tilde{\rho}(t,x)=\left\{
\ba{lll}
(f')^{-1}(x/t) & \mbox{if} & 0<x<\min(1,tf'(\rho))\\ \\
\rho & \mbox{if} & \min(1,tf'(\rho))<x<1
\ea
\right.
\ee
for $t<t_2$,
\be \label{tilderho_32}
\tilde{\rho}(t,x)=\left\{
\ba{lll}
(f')^{-1}(x/t) & \mbox{if} & 0<x<x_t\\ \\
\rho_r & \mbox{if} & x_t<x<1
\ea
\right. \ee
for $t>t_2$, where $x_t$ is given by
\be \label{def_min_x3} x_{t_2}=1,\quad
\dot{x}_t=v[\rho_r,(f')^{-1}(x_t/t)]<v[\rho_r,(f')^{-1}(1/t)]<0\mbox{
for }t>t_2
\ee
so that $x_t$ reaches $0$ in finite time. Note that at time $t_2$,
the rarefaction wave issued from $0$ has value $\varphi(\rho_r)$
at $x=0$, and value $>\varphi(\rho_r)$ behind, which is why a
shock starts propagating from $x=1$ towards $x=0$.\\ \\
{\em Second case: $\varphi(\rho_r)<\rho<\rho^*$.} Then
$\rho'=\varphi(\rho_r)$, $\varphi[\tilde{G}(t,1^-)]=\rho_r$ for
all $t>0$, and
\be \label{tilderho_33}
\tilde{\rho}(t,x)=\left\{
\ba{lll}
\rho_r & \mbox{if} & 1+tv(\rho_r,\rho)<x<1\\ \\
\rho & \mbox{if} & tf'(\rho)<x<1+tv(\rho_r,\rho)\\ \\
(f')^{-1}(x/t) & \mbox{if} & 0<x<tf'(\rho)
\ea
\right.
\ee
for $t<t_3:=[f'(\rho)-v(\rho_r,\rho)]^{-1}$,
\be \label{tilderho_34}
\tilde{\rho}(t,x)=\left\{
\ba{lll}
\rho_r & \mbox{if} & y_t<x<1\\ \\
(f')^{-1}(x/t) & \mbox{if} & 0<x<y_t
\ea
\right.
\ee
where $y_t$ is given by
\be \label{def_min_x4}
y_{t_3}=1+t_3 v(\rho_r,\rho),\quad
\dot{y}_t=v[\rho_r,(f')^{-1}(y_t/t)]<v(\rho_r,\rho)<0
\ee
and thus reaches $0$ in finite time.\\ \\
{\em Third case:
$\varphi(\rho_l)\leq\varphi(\rho_r)\leq\rho^*\leq\rho$.} Then
again $\rho'=\varphi(\rho_r)$ and
$\varphi[\tilde{G}(t,1^-)]=\rho_r$,
\be \label{tilderho_35}
\tilde{\rho}(t,x)=\left\{
\ba{lll}
\rho_r & \mbox{if} & \max[0,1+tv(\rho_r,\rho)]<x<1\\ \\
\rho & \mbox{if} & 0<x<\max[0,1+tv(\rho_r,\rho)]
\ea
\right.
\ee
\end{appendix}
\end{document}